\documentclass[11pt]{article}

\usepackage[margin=1in]{geometry}
\usepackage{times}
\usepackage{microtype}
\usepackage{amsmath,amssymb,amsfonts,amsthm}
\usepackage{graphicx}
\usepackage{enumitem}
\usepackage{authblk}
\usepackage{hyperref}
\usepackage{booktabs}
\usepackage{tikz}
\usepackage{xcolor}
\usepackage{subcaption}
\usepackage[ruled,vlined,linesnumbered]{algorithm2e}
\usepackage{float}
\usepackage{multirow}
\usepackage{comment}
\SetAlgoCaptionSeparator{.}
\SetAlgoCaptionLayout{small}
\SetAlgoNlRelativeSize{-1}
\DontPrintSemicolon
\SetKwComment{Comment}{$\triangleright$\ }{}
\usepackage[T1]{fontenc}
\usepackage[utf8]{inputenc}
\usepackage{newtxtext}
\usepackage{newtxmath}
\hypersetup{
    colorlinks=true,
    linkcolor=blue,
    citecolor=blue,
    urlcolor=blue
}

\newtheorem{theorem}{Theorem}
\newtheorem{lemma}{Lemma}
\newtheorem{corollary}{Corollary}
\newtheorem{proposition}{Proposition}
\theoremstyle{definition}
\newtheorem{definition}{Definition}

\theoremstyle{remark}

\title{The Geometry of Coalition Power:\\
Majorization, Lattices, and Displacement in Multiwinner Elections}

\author[1]{Qian Guo}
\author[2]{Yidan Hu}
\author[1]{Rui Zhang}

\affil[1]{Department of Computer and Information Sciences, University of Delaware}
\affil[2]{Department of Cybersecurity, Rochester Institute of Technology}

\date{}

\begin{document}
\maketitle

\begin{abstract}
How much influence can a coordinated coalition exert in a multiwinner \emph{Top-$k$} election under a positional scoring rule?
We study this question through the lens of the \emph{maximum displacement} problem:
given a coalition of size $m$, how many of the current top-$k$ winners can be forcibly
replaced?

Our first contribution is a structural decomposition showing that coalition power reduces to two independent prefix-majorization constraints—one controlling how strongly the strongest outsiders can be boosted, and one governing how weak winners can be suppressed.
For general positional scoring rules, these prefix inequalities give tight, efficiently checkable necessary conditions---i.e., they exactly capture the coalition's \emph{score-supply} limits in the continuous relaxation.

We obtain an exact discrete characterization when the relevant score ladders are arithmetic-progression (AP) ladders with a common step size $g$ (``common-step AP ladders''), a broad class including Borda, truncated Borda,
$k$-approval/$k$-veto variants, plurality, and other multi-level positional rules (e.g., $3$--$2$--$1$ scoring).
  In this setting we prove an exact
\emph{Majorization–Lattice Theorem}: the coalition’s aggregate feasible score vectors
are precisely the integer points that satisfy (i) the Block–HLP prefix-sum capacity
constraints and (ii) a single global congruence condition modulo~$g$.
For pure Borda ($g=1$), the congruence constraint disappears and feasibility reduces to
a \emph{pure prefix-majorization} test.

This discrete geometry yields a fast $O(k' \log k')$ exact feasibility oracle: given a target
displacement level $k'$, it checks whether a coalition of size $m$ can displace $k'$ winners by combining
a Block--HLP prefix test with a single congruence check. This oracle enables an exact dual-envelope binary
search over $k'$ for computing the \emph{maximum achievable displacement} $k^\ast$.
The resulting algorithm runs in $O\!\left(k(\log k)^2\log(mx)\right)$ time under a fixed Top-$k$ rule and common-step AP scoring ladders, and certifies infeasibility whenever coalition demand exceeds prefix capacities.

Experiments on synthetic Mallows profiles and real PrefLib elections validate the theory: the predicted
cutoff envelopes match the oracle exactly, exhibit clear diminishing returns, and substantially outperform
natural baseline heuristics. Beyond unit-step Borda ($g=1$), we also validate the discrete congruence
phenomenon predicted by Theorem~4 on sparse-step common-step AP rules ($g>1$), where a prefix-only test
admits false positives but the prefix-and-congruence oracle remains exact. The feasibility oracle also scales to extreme election sizes—processing one billion candidates in under 28 seconds—with larger tests limited only by available memory.

\end{abstract}


\maketitle
\thispagestyle{plain}

\section{Introduction}
\label{sec:intro}

Coalitional manipulation—the ability of a group of voters to coordinate
their ballots to alter an election outcome—is a central challenge in
computational social choice. Classic results show that manipulating many
positional scoring rules is NP-hard in the worst case
\cite{bartholdi1989,conitzer2006aaai,faliszewski2010cacm}, yet these
hardness barriers offer limited insight into the \emph{structure} of how
coalitions can redistribute scores under widely used rules such as Borda.
Empirical work further suggests that real elections often exhibit
regularities that make manipulation substantially easier in practice
\cite{aziz2015aamas,azizLee2020aamasEA}.

In modern multiwinner settings—committee formation, peer-review
aggregation, competitive evaluations—the strategic aim is rarely to make a
single candidate win; it is to change the \emph{composition} of the winner
set. We focus on the standard Top-$k$ rule, where the committee size $k$ is
fixed in advance. This assumption is crucial: it creates a well-defined
separating threshold and makes the notion of displacement well posed.

This leads to a fundamental operational question:

\vspace{4pt}
\begin{quote}
\textbf{Maximum Displacement.}
\emph{Given a coalition of size $m$, how many of the current top-$k$ winners
can it forcibly replace?}
\end{quote}
\vspace{4pt}

A coalition achieves displacement level $k'$ if it can cause $k'$ currently
nonwinning candidates to finish above the final cutoff while forcing $k'$ of
the current top-$k$ winners to finish strictly below it. Our goal is to
compute the \emph{maximum achievable displacement} $k^\ast$ for a coalition of
size $m$, where
\[
k^\ast \;:=\; \max\{\,k' : \text{$k'$ is achievable by a coalition of size $m$}\,\}.
\] 
Understanding this requires a precise description of the aggregate
score vectors the coalition can generate, going far beyond classical
worst-case complexity analyses.

\paragraph{Structural decomposition of manipulation.}
Our first contribution is a structural decomposition of coalition influence
at a fixed displacement query level~$k'$. We show that:
\begin{itemize}
\item only the top-$k'$ and bottom-$k'$ positions of each coalition ballot
affect displacement feasibility;
\item any successful manipulation can be transformed into a canonical one in
which all top-$k'$ positions are assigned to the strong outsiders and all
bottom-$k'$ positions to the weak winners; and
\item feasibility separates cleanly into two independent subproblems: one governing the boosts available to strong outsiders and one
governing the suppressions achievable on weak winners.
\end{itemize}
Together with a complementary impossibility theorem, this yields a sharp
\emph{if-and-only-if} description of feasible displacement: feasibility of
both subproblems at cutoff~$B$ is necessary and sufficient for the coalition
to deterministically replace at least~$k'$ winners in \emph{every}
completion, regardless of how the middle positions are filled.

\paragraph{Geometric view: the Block--HLP envelope.}
Our second, and technically central, contribution is a geometric
characterization of coalition power. For any positional scoring rule, the set
of aggregate coalition score allocations forms the Minkowski sum of $m$
permutahedra, one per ballot. Its convex hull is governed by a concise family
of blockwise prefix-sum inequalities generalizing the classical
Hardy--Littlewood--Pólya majorization relations. This \emph{Block--HLP}
envelope exposes the intrinsic geometry of the coalition’s score supply,
independent of any particular manipulation strategy, and provides the
foundation for all subsequent structural and algorithmic results.

\paragraph{Exact discrete structure for common-step AP ladders.}
For the broad and practically important class of scoring rules whose
per-ballot scores follow an \emph{arithmetic-progression (AP) ladder with a
common step size}—including Borda, truncated Borda, $k$-approval, 3--2--1
scoring, and many affine variants—we refine the convex Block--HLP envelope to
an \emph{exact discrete} characterization.  
In this setting, an integer score vector is realizable by the $m$ colluding
ballots if and only if it satisfies:
(i) the Block--HLP prefix-sum capacity constraints, and  
(ii) a single global congruence condition modulo $g$, the common AP step size.  
Hence the feasible region is a clean \emph{prefix-and-congruence lattice}:
the intersection of the Block--HLP polymatroid with one arithmetic residue
class.  
This structure bypasses any need to enumerate ballot permutations and yields an
efficient and exact feasibility oracle.

\paragraph{Algorithmic consequences.}
Combining the structural and geometric insights yields an efficient
algorithm for computing the maximum achievable displacement $k^\ast$ under AP
scoring rules. For each cutoff~$B$, the boost and suppress requirements become
ordered demand vectors; feasibility reduces to checking prefix dominance
together with a single congruence test. Both checks are monotone in $B$, enabling binary search to identify the feasible cutoff interval for each queried displacement level~$k'$.
An outer binary search over $k' \in \{0,\dots,k\}$ then computes $k^\ast$, yielding overall running time
$O\!\left(k(\log k)^2\log(mx)\right)$ in the worst case (since $k'\le k$), with infeasibility certified whenever
coalition demand exceeds prefix capacities.

\paragraph{Contributions.}
This paper makes the following contributions:
\begin{enumerate}
\item \textbf{Structural decomposition.}
      We show that displacement at level $k'$ decomposes into two independent feasibility questions: boosting strong outsiders
      and suppressing weak incumbents.

\item \textbf{Guaranteed displacement.}
      We prove that feasibility of the boost and suppress subproblems at
      cutoff~$B$ is \emph{necessary and sufficient} for a coalition to
      deterministically replace at least~$k'$ winners in \emph{every}
      completion of the nonmanipulated ballots.

\item \textbf{Geometric characterization (continuous envelope via classical majorization).}
      For arbitrary positional scoring rules, we show that the coalition’s continuous feasible region is the
Minkowski sum of ballot-wise permutahedra and therefore admits a succinct Block--HLP (prefix-capacity)
description by classical majorization theory.
This yields a simple and efficiently checkable continuous feasibility oracle.

\item \textbf{Exact discrete structure for common-step AP ladders.}
      When all colluding ballots follow AP ladders with a
      common step size, we prove a \emph{Majorization–Lattice Theorem}:
      an integer score vector is realizable iff it satisfies the Block--HLP
      prefix constraints together with a single global congruence condition.
      This yields an exact prefix-and-congruence lattice of feasible
      aggregates.

\item \textbf{Polynomial-time computation of maximum achievable displacement.}
      Using this lattice structure, we design a fast algorithm for all
      common-step AP rules that computes the maximum achievable displacement
      $k^\ast$ via a dual-envelope binary search over $k'$ supported by
      efficient prefix and congruence oracles.

\end{enumerate}

Overall, our results provide a unified structural and algorithmic framework for quantifying coalition power
in multiwinner elections under positional scoring rules, and yield an exact polynomial-time algorithm in the
common-step AP regime.

\begin{center}
\fbox{
\begin{minipage}{0.95\linewidth}
\textbf{Technical Roadmap.}
(i) We show that maximum displacement decomposes into independent boost and suppress subproblems.
(ii) We characterize the continuous feasibility region via Block--Hardy--Littlewood--Pólya prefix constraints.
(iii) For common-step arithmetic scoring ladders, we prove exact discrete realizability via prefix dominance plus a single congruence condition.
(iv) We leverage monotone feasibility envelopes to compute the maximum achievable displacement $k^\ast$ efficiently.
\end{minipage}}
\end{center}

\paragraph{Roadmap.}
Section~\ref{sec:related} reviews related work.
Section~\ref{sec:model} formalizes the voting model and problem definition.
Section~\ref{sec:decomposition} develops the structural decomposition.
Section~\ref{sec:geometry} presents the geometric Block--HLP framework and its
AP specialization.
Section~\ref{sec:algorithm} describes the dual-envelope search algorithm.
Section~\ref{sec:experiments} reports empirical results.
Section~\ref{sec:discussion} discusses tractability boundaries, geometric
frontiers, and extensions of the framework.
Section~\ref{sec:conclusion} concludes and outlines promising future directions.

\section{Related Work}
\label{sec:related}

\paragraph{Computational manipulation in voting.}
A central theme in computational social choice is whether strategic behavior
can be discouraged by computational hardness. Following the impossibility
theorems of Gibbard and Satterthwaite \cite{gibbard1973,satterthwaite1975},
Bartholdi et al.~\cite{bartholdi1989} initiated the computational study of
coalitional manipulation, establishing NP-hardness results for prominent
positional rules such as Borda and STV. A rich subsequent literature has mapped
the boundary between tractable and intractable manipulation across a wide range
of voting rules, including positional, approval-based, and multiwinner systems
\ \cite{conitzer2006aaai,conitzer2007aaaiFewCandidates,faliszewski2010cacm,
zuckerman2009aij,betzler2009mfcs,davies2011aaaiBorda,
aziz2015aamas,bredereck2016largescale,erdelyi2015kapproval}.
These works typically cast manipulation as a search or decision problem—for
example, determining whether a designated candidate can be made to win or
whether a specific committee can be achieved—where the main obstacle is
worst-case computational complexity.

While this line of work provides important hardness guarantees, it offers
limited insight into the \emph{structure} of feasible score redistributions.
Complexity-theoretic reductions focus on adversarially constructed preference
profiles and do not characterize the geometry of coalition influence under
regular scoring rules such as Borda. Moreover, traditional formulations study
\emph{targeted} manipulation—optimizing toward a particular candidate or
committee—whereas many practical settings call for understanding the
\emph{extent} of manipulation achievable by a coalition, i.e., its maximum
achievable displacement $k^\ast$, independent of any specific outcome.

Our viewpoint differs fundamentally. Rather than optimizing toward a fixed
target, we characterize the \emph{entire feasible influence region}: the set of
all aggregate score vectors that a coalition of size $m$ can induce while
honest ballots remain fixed. Instead of asking whether some candidate can be
made to win, we determine precisely when a coalition can guarantee displacement
at level $k'$.
 This shift from target-specific search to
global structural characterization exposes the geometric and lattice structure
of feasible manipulations and yields sharp necessary and sufficient conditions
for guaranteed displacement. We are not aware of prior work that gives an explicit global description (in aggregate score space) of the full
feasible manipulation region for positional scoring rules in the multiwinner displacement setting studied here,
nor an exact characterization that supports computing $k^\ast$ via feasibility queries over $k'$.

\paragraph{Integer-programming formulations of manipulation.}
A parallel line of work develops \emph{integer programming} and
\emph{network-flow} formulations of manipulation, in which ballot-level
decisions are modeled explicitly through binary assignment variables. Early
approaches, such as those of Betzler and Dorn \cite{BetzlerDorn2009} and Xia and
Conitzer \cite{XiaConitzer2008}, reduce manipulation to the feasibility of a
structured integer program whose constraints encode candidate--position
assignments consistent with voter rankings. Subsequent refinements introduce
more sophisticated combinatorial constraints, flow-based formulations, or
optimization heuristics to improve practical performance on moderate-size
instances \cite{xiaConitzer2008PossibleNecessary}.  

While powerful as algorithmic tools, these IP-based methods fundamentally
operate at the level of \emph{ballot permutations}, requiring the solver to
encode or implicitly explore exponentially many candidate--position matchings.
In contrast, our approach bypasses ballot-level combinatorics entirely: the
Block--HLP decomposition and AP-lattice characterization describe the
coalition’s feasible influence region \emph{directly in aggregate score
space}. This yields sharp structural conditions and polynomial-time feasibility
tests without integer optimization. The geometric framework developed in this
paper is therefore orthogonal to IP-based approaches, capturing global
constraints on feasible score redistributions that permutation-level
formulations do not make explicit.

\paragraph{Parameterized complexity of manipulation.}
A complementary line of research studies whether the hardness of coalitional
manipulation can be alleviated by fixing structural parameters such as the
number of candidates, coalition size, or restrictions on preference profiles.
Early results show that manipulation under Borda and related positional rules
remains W[1]-hard when parameterized by coalition size
\cite{BetzlerSlinkoUhlmann2011}, while certain cases become fixed-parameter
tractable when parameterized by the number of candidates or score levels
\cite{XiaConitzer2008,FaliszewskiEtAl2011}. Subsequent work extends this
parameterized perspective to related problems such as bribery, control, and
multiwinner elections, identifying regimes that admit FPT algorithms and
others that remain intractable
\cite{bredereck2016largescale,dey2016frugal}.

Despite yielding a fine-grained complexity landscape, this body of work
continues to treat manipulation as a \emph{search problem over ballot
permutations}, with tractability achieved by bounding combinatorial parameters.
As a result, it offers limited insight into the global structure of feasible
score redistributions. In contrast, our geometric framework characterizes the
coalition’s influence \emph{without} parameter restriction: the Block--HLP
envelope and AP-lattice structure describe the entire feasible region of
aggregate score vectors directly. Tractability in our setting arises from
structural decomposition and majorization properties rather than from fixing
problem parameters.

\paragraph{Combinatorial interpretations of positional scoring transformations.}
A related direction studies the combinatorial structure of positional scoring
rules themselves, focusing on how properties of a scoring vector behave under
affine transformations, rescalings, or changes in score granularity. Classical
work by Fishburn \cite{Fishburn1978} and Young \cite{Young1977} examined
invariance and consistency properties of voting outcomes under affine
transformations of score vectors, clarifying when such transformations preserve
ordinal outcomes or social-choice axioms. This line of work sheds light on
equivalence classes of scoring rules and on how discrete score gaps influence
the behavior of positional mechanisms.

However, this literature focuses on \emph{transformations of scoring rules}
themselves, rather than on transformations of the \emph{feasible aggregate
score allocations} induced by a fixed rule. In contrast, our framework holds
the positional rule fixed and characterizes the coalition’s entire feasible
influence region. The Block--HLP envelope captures the continuous limits of
score redistribution, while the AP-lattice structure yields an exact discrete
description of feasible aggregate score vectors. Thus, while scoring-rule
analyses explain relationships between different positional mechanisms, our
geometric analysis explains the intrinsic structural constraints on coalition
power under a given rule.

\paragraph{Possible and necessary winners.}
The literature on possible and necessary winners under incomplete preferences studies which candidates can
win, or must win, across all linear extensions of a partially specified preference profile
\cite{xiaConitzer2011jaisPW}. Our guaranteed-displacement theorem has a related flavor, but the source of
uncertainty here is different: it comes from how the coalition completes the \emph{middle} positions of its
ballots and (potentially) from tie-breaking at score ties.

Once a coalition satisfies the appropriate boost/suppress feasibility constraints at a cutoff $B$, it can
guarantee the \emph{displacement level} $k'$: in every completion of the coalition's ballots (and under any
tie-breaking), the final Top-$k$ winner set contains at least $k'$ outsiders, and hence at least $k'$ honest
winners are excluded. When the scoring vector contains plateaus, the identities of the entering outsiders and
displaced winners may depend on tie-breaking; the robust guarantee is therefore about the displacement
\emph{count} rather than a fixed set of winners.

\paragraph{Positional scoring and structural decompositions.}
Prior work on manipulation for positional scoring rules typically relies on
explicit combinatorial formulations, including integer programs, network-flow
models, or matching-based encodings
(e.g.,~\cite{betzler2009mfcs,xia2008ec,xiaConitzer2008PossibleNecessary,dey2016frugal}).
These approaches reason directly about the placement of candidates at specific
positions on coalition ballots and often require substantial case analysis,
branching, or tailored dynamic programming to handle different scoring vectors
and constraints.

In contrast, our approach exposes a two-level geometric structure that governs
\emph{all} feasible manipulations under a fixed positional rule. At the
continuous level, the coalition’s aggregate score supply is the Minkowski sum
of $m$ permutahedra, whose convex hull admits a succinct description via the
Block--HLP family of blockwise prefix-sum inequalities. At the discrete level,
when per-ballot scores form an AP ladder, this convex
envelope sharpens to an \emph{exact} majorization lattice characterized by
prefix dominance together with a single congruence condition. As a result,
feasibility can be decided using only prefix capacities and modular arithmetic,
without enumerating ballot permutations or solving large combinatorial
subproblems.

\paragraph{Majorization and convex geometry.}
The relationship between permutahedra, convexity, and majorization is classical
\cite{hardy1934,marshall2011,rado1952,edmonds1970}. These foundations underlie a
broad range of results in inequalities, combinatorial optimization, and
polyhedral theory. However, existing applications of majorization do not
address coalitional manipulation in voting, the redistribution of score mass
across ballots, or feasibility questions in multiwinner settings under
positional rules.

Building on classical majorization and polymatroid geometry, we \emph{recast} the coalition’s continuous
feasible region---a Minkowski sum of ballot-wise permutahedra---into the specific blockwise prefix-capacity
form needed for displacement analysis (Theorem~\ref{thm:block-hlp}).
Our novelty is not the underlying permutahedral characterization itself, but how this envelope combines with
(i) the canonical boundary-target structure of displacement (boost vs.\ suppress), and (ii) an \emph{exact}
integer realizability refinement for common-step AP ladders via a single congruence condition
(Theorem~\ref{thm:ap-lattice-common}).

\paragraph{Why tractability emerges here.}
Classical NP-hardness results for manipulation typically arise from the need to
search over exponentially many subsets of candidates to promote or demote
\cite{bartholdi1989,conitzer2007aaaiFewCandidates,faliszewski2010cacm}.
In the displacement setting studied here, however, the “targets’’ are not chosen
by the algorithm; they are determined canonically by the order statistics of the
honest scores. The coalition always seeks to elevate the highest-scoring
outsiders and depress the lowest-scoring incumbents.

This canonical choice eliminates the combinatorial branching that underlies
hardness and collapses manipulation into a \emph{majorization-feasibility}
problem. For AP ladders with a common step size, the
resulting blockwise prefix inequalities are not only necessary but also
sufficient. This yields polynomial-time feasibility oracles and enables the
dual-envelope binary-search algorithm developed in this paper.

\paragraph{AP ladders and the tractability boundary.}
AP ladders form a broad and practically important class of positional scoring rules
for which score realizability exhibits strong arithmetic regularity.
When all colluding ballots share a \emph{common} step size~$g$, realizability is
governed by a clean prefix-and-congruence lattice structure.  
In contrast, prior hardness results for unrestricted positional rules
\cite{bartholdi1989,faliszewski2010cacm} indicate that when score gaps vary
across ballots, the feasible score vectors fragment into multiple incompatible
residue classes, preventing a unified majorization-based treatment.  
This contrast illustrates why common-step AP ladders yield a uniquely tractable
regime: they preserve the convex prefix structure while aligning all integer
points to a single modular congruence class.  
A full delineation of the tractability frontier—identifying which scoring
families admit such a prefix-and-congruence characterization—remains an
intriguing open problem.

\section{Model and Problem Setup}
\label{sec:model}

We study multiwinner elections conducted under positional scoring rules.
This section introduces the election model, notation, and the displacement
objective that underlies our structural and geometric analysis.

\paragraph{Candidates, voters, and positional scoring rules.}

The election contains a candidate set $\mathcal{C}$ with $|\mathcal{C}|=x$ candidates and a voter set $\mathcal{N}$ with $|\mathcal{N}|=n$ voters.
Each voter submits a complete ranking, and under a positional scoring rule specified by a nonincreasing vector
$p = (p_1,\ldots,p_x)$ where $p_1 \ge \cdots \ge p_x$, a candidate ranked in position $r$ receives $p_r$ points.
We assume $p$ is integer-valued (w.l.o.g.\ by scaling any rational scoring vector), so the one-unit separation
in the cutoff constraints (e.g., $B-1$) is well-defined.
The winners are the $k$ candidates with the largest total scores
(the standard Top-$k$ rule).

\paragraph{Tie-breaking for defining the honest baseline.}
If multiple candidates have equal total score, the Top-$k$ rule does not specify a unique committee.
To make the honest baseline committee and our boundary sets well-defined, 
we fix an arbitrary but exogenous strict priority order $\tau$ over candidates (e.g., by candidate
index), and break score ties in favor of the candidate with smaller priority value $\tau(\cdot)$.

Formally, given the honest score vector $(S_c)_c$, we define the tie-broken honest order $\succ_H$ by
\[
a \succ_H b \quad \Longleftrightarrow \quad
\bigl(S_a > S_b\bigr)\ \ \text{or}\ \ \bigl(S_a = S_b \ \text{and}\ \tau(a) < \tau(b)\bigr).
\]
Throughout, whenever we refer to the ``strongest''/``weakest'' candidates under the honest profile, we mean
with respect to $\succ_H$. This convention is used only to define $\mathcal{T},\mathcal{O},\mathcal{O}^\ast,\mathcal{T}_w$; our displacement guarantees
are stated to hold under \emph{any} tie-breaking after manipulation because we require a strict one-unit separation.

\paragraph{Fixed committee size.}
Throughout, we assume that the committee size $k$ is fixed and known
\emph{a~priori}.
This assumption is essential: it yields a well-defined post-manipulation
threshold score and makes the displacement objective meaningful.
Mechanisms with variable committee size (e.g., approval-based quota rules)
require a different notion of displacement and fall outside the scope of this
work.

\paragraph{Honest scores, winners, and outsiders.}
Let $S_c$ denote the total score of candidate $c$ from all honest voters.
When honest scores tie, we use the tie-broken honest order $\succ_H$ defined above.

Let
\[
\mathcal{T}=\{t_1,\ldots,t_k\}
\]
denote the honest winner set (the top-$k$ candidates under $\succ_H$), indexed in increasing order under $\succ_H$ (from weakest to strongest), so that $t_1$ is the weakest incumbent and $t_k$ is the strongest.

The remaining candidates form the outsider set
\[
\mathcal{O}=\{o_1,\ldots,o_{x-k}\},
\]
indexed in decreasing order under $\succ_H$ so that $o_1$ is the strongest outsider.

\paragraph{Coalition and displacement objective.}
A fully coordinated coalition of $m$ voters seeks to replace $k'$ of the $k$
current winners.  Here $k'$ denotes a \emph{displacement level under consideration}
(i.e., a feasibility query), and the coalition’s \emph{maximum achievable displacement} is
\[
k^\ast \;:=\; \max\{\,k' : \text{displacement at level $k'$ is achievable by a coalition of size $m$}\,\}.
\]
The displacement level $k'$ must satisfy
\[
0 \le k' \le \min\{k,\, x-k\},
\]
since at most $k$ incumbents can be removed and at most $x-k$ outsiders are
available to replace them.  
Within this feasible range, a displacement attempt at level $k'$
concerns only two boundary subsets (defined with respect to the tie-broken honest order $\succ_H$):
\[
\mathcal{O}^\ast = \{o_1,\dots,o_{k'}\} \quad\text{(strongest outsiders)}, \qquad
\mathcal{T}_w = \{t_1,\dots,t_{k'}\} \quad\text{(weakest winners)}.
\]
These sets are canonical: if the coalition cannot separate $\mathcal{O}^\ast$ from
$\mathcal{T}_w$—that is, if it cannot raise all candidates in $\mathcal{O}^\ast$
above a common threshold while pushing all candidates in $\mathcal{T}_w$
strictly below it—then no other selection of $k'$ outsiders and $k'$ winners
can achieve displacement.  
The remainder of the paper characterizes exactly when such a separation is
possible.

\paragraph{Cutoff score and strict separation.}
Let $B$ be a \emph{separating cutoff score}.  
After the coalition’s ballots have been applied, a displacement at level~$k'$ is
\emph{successful} if
\[
S_o+\Delta_o \ge B \quad\text{for all } o\in\mathcal{O}^\ast,
\qquad
S_t+\Delta_t \le B-1 \quad\text{for all } t\in\mathcal{T}_w.
\]
These inequalities ensure that every targeted outsider lies strictly above every
targeted weak winner in the final score order.  Because of the one-unit gap, all
$o\in\mathcal{O}^\ast$ must appear \emph{ahead of} all $t\in\mathcal{T}_w$ in the
final ranking under \emph{any} tie-breaking rule. We require a strict score gap so the displacement guarantee holds under \emph{any} tie-breaking (even adversarial); we do not count manipulations that rely on favorable tie-breaking.

Importantly, candidates outside the boundary sets $\mathcal{O}^\ast$ and $\mathcal{T}_w$ may also end up on either side of the
cutoff $B$. When the scoring vector contains plateaus, such candidates may tie boundary candidates at the
Top-$k$ cutoff, so the \emph{identity} of which outsiders enter and which incumbents leave can depend on
tie-breaking. This does not affect the certified guarantee at level $k'$: once $\mathcal{O}^\ast$ is strictly separated from
$\mathcal{T}_w$ by a one-unit gap, every Top-$k$ winner set (under any tie-breaking) contains \emph{at least $k'$}
outsiders and therefore excludes \emph{at least $k'$} honest winners (Theorem~\ref{thm:guaranteed-topk} in Section~\ref{sec:guaranteed}). Thus strict
separation certifies the displacement \emph{level} $k'$ (count), not necessarily a unique displaced set when
plateaus create ties.

Studying how specific tie-breaking rules affect the \emph{identities} of entering and displaced candidates is
an interesting direction for future work.

\paragraph{Boost and suppress requirements.}
Fix a candidate cutoff value $B$.
To satisfy the separation inequalities above, each targeted outsider $o \in
\mathcal{O}^\ast$ must receive at least
\[
b_o = \max\{0,\, B - S_o\}
\]
additional points from the coalition.  
Similarly, each weak winner $t \in \mathcal{T}_w$ may accrue \emph{at most}
\[
u_t = \max\{0,\, B - 1 - S_t\}
\]
additional points without violating the separation constraint.
Let $b(B)$ and $u(B)$ denote the resulting boost and
suppress vectors with respect to cutoff $B$.  We also refer to $u(B)$ as the \emph{suppression slack vector}, as it captures the remaining score slack each weak winner can absorb while staying below the cutoff.
At a high level, the coalition must find a way to distribute its available
scores so that every outsider meets its required boost while every weak winner
respects its suppression allowance.

\paragraph{Coalition score budget.}
Each manipulative ballot contributes one copy of every positional score
$p_1,\dots,p_x$.  
Across $m$ coordinated voters, the coalition therefore controls $m$ copies of
each score in the vector~$p$, which it may distribute arbitrarily
through its chosen ballot rankings.
How these resources can be arranged to satisfy the boost and suppress demands
described above—and how these two forms of influence interact—is the focus of
Section~\ref{sec:decomposition}, where we develop the structural decomposition
that underpins our geometric analysis.

\begin{table}[t]
\centering
\small
\renewcommand{\arraystretch}{1.2}
\begin{tabular}{cl}
\toprule
\textbf{Symbol} & \textbf{Meaning} \\
\midrule
$x$ & Number of candidates \\
$n$ & Number of voters \\
$m$ & Number of colluding voters \\
$k$ & Size of winner set (Top-$k$ rule) \\
$k'$ & Queried displacement level (feasibility target) \\
$k^\ast$ & Maximum achievable displacement for coalition size $m$ \\
$p$ & Positional scoring vector \\
$S_c$ & Honest score of candidate $c$ \\
$\tau$ & Fixed exogenous tie-breaking priority over candidates used to define $\succ_H$ \\
$\mathcal{T},\ \mathcal{O}$ & Honest winners and outsiders \\
$\mathcal{T}_w, \mathcal{O}^\ast$ & Weakest $k'$ winners / strongest $k'$ outsiders under the tie-broken honest order $\succ_H$ \\
$\Delta_c$ & Coalition’s score contribution to candidate $c$ \\
$B$ & Post-manipulation cutoff score \\
$b_o$ & Boost requirement for outsider $o$ \\
$u_t$ & Suppression allowance for weak winner $t$ \\
\bottomrule
\end{tabular}
\caption{Summary of notation introduced in Section~\ref{sec:model}.}
\label{tab:notation}
\end{table}

\paragraph{Notation.}
Table~\ref{tab:notation} summarizes the main symbols introduced in this section.
Throughout the paper, when referring to ``top'' or ``bottom'' scores of a ballot,
we mean the \emph{score copies associated with specific ballot positions}
(e.g., top-$k'$ or bottom-$k'$ positions), not distinct numerical values.
This distinction matters when the scoring vector contains ties.

\paragraph{Looking ahead.}
The boost and suppress demand vectors $(b,u)$ form
the inputs to our feasibility analysis.  
Section~\ref{sec:decomposition} shows that, once coalition ballots are placed
in canonical form, feasibility of these two demand systems completely
characterizes whether displacement at level~$k'$ is achievable.  
Section~\ref{sec:geometry} then develops the geometric machinery—Block--HLP
prefix inequalities and their AP-lattice refinement—that determines exactly
when these demand systems can be met.

\section{Structural Decomposition}
\label{sec:decomposition}

This section develops the structural principles that make the maximum
displacement problem tractable.  
We show that, for any positional scoring rule, the coalition’s influence
splits cleanly into two independent components: the ability to \emph{boost}
outsiders using the top positional scores, and the ability to \emph{suppress}
incumbents using the bottom positional scores.  
This separation is a direct consequence of the ordered score structure of
positional rules and underlies both the geometric feasibility tests of
Section~\ref{sec:geometry} and the dual binary-search algorithm of
Section~\ref{sec:algorithm}.  
Throughout, we use the notation and score requirements introduced in
Section~\ref{sec:model}.

\subsection{Conceptual Overview}
\label{sec:conceptual}

Under any positional scoring rule, each coalition ballot contributes one copy
of every positional score $p_1 \ge p_2 \ge \cdots \ge p_x$.  
Crucially, the \emph{top $k'$} positions and the \emph{bottom $k'$} positions form
two completely disjoint pools of influence:
\[
\underbrace{p_1, \ldots, p_{k'}}_{\text{boost resource}}
\quad\text{and}\quad
\underbrace{p_{x-k'+1}, \ldots, p_x}_{\text{suppress resource}}.
\]
Only these extremal score segments can meaningfully affect displacement at level~$k'$:
the top-$k'$ scores determine how much the coalition can raise the strongest outsiders $\mathcal{O}^\ast$,
and the bottom-$k'$ scores determine the minimum total coalition contribution that must be distributed
among the weakest winners $\mathcal{T}_w$ in any maximally suppressive (canonical) manipulation.

The remaining $(x - 2k')$ positions, lying strictly between these two segments,
play no direct role in satisfying the score inequalities $S_o+\Delta_o \ge B$ for
outsiders and $S_t+\Delta_t \le B-1$ for weak winners (though they may affect non-boundary rankings).   
Their precise assignments can influence tie patterns or the rankings of
non-boundary candidates, but they do not provide additional boost to outsiders
or additional suppression to weak winners.  
Their effect on feasibility is therefore mediated entirely through the
extremal segments, whose structure we analyze in the remainder of this section.

This structural separation has two immediate consequences.

\begin{itemize}
\item \textbf{Canonical form.}  
Any successful manipulation can be transformed, without loss of generality, so
that all top-$k'$ positions are assigned to $\mathcal{O}^\ast$ and all 
bottom-$k'$ positions to $\mathcal{T}_w$.  
All middle positions are irrelevant to feasibility and may be filled 
arbitrarily.

\item \textbf{Independent feasibility.}  
Because the boost and suppress resources come from disjoint subsets of the
scoring vector, the feasibility of boosting $\mathcal{O}^\ast$ and restricting the
scores of $\mathcal{T}_w$ can be analyzed as two completely independent 
subproblems.  
No choice in one score pool can affect feasibility in the other.
\end{itemize}

In addition, whenever the coalition achieves a strict score separation at a cutoff $B$ between $\mathcal{O}^\ast$ and
$\mathcal{T}_w$, it guarantees the displacement level $k'$: in every completion of the coalition's ballots and under any
tie-breaking, the final Top-$k$ winner set contains \emph{at least $k'$} outsiders (and hence excludes at least
$k'$ honest winners). When the scoring vector contains plateaus, the identities of which outsiders enter and
which incumbents leave can depend on tie-breaking. Section~\ref{sec:guaranteed} formalizes the robust displacement
guarantee we use throughout.

The remainder of this section develops these principles formally through the
canonical-manipulation lemma, the independent-feasibility lemma, and the
guaranteed-displacement theorem.

\subsection{Canonical Manipulation Strategies}
\label{sec:canonical}

The coalition controls $m$ copies of each positional score
$p_1,\dots,p_x$.
For a fixed displacement level~$k'$, only the \emph{top $k'$} and
\emph{bottom $k'$} positions on each coalition ballot affect the feasibility of
separating $\mathcal{O}^\ast$ and $\mathcal{T}_w$:
the top positions determine how much boost can be delivered to the strongest
outsiders, while the bottom positions determine the minimum total coalition contribution that must be distributed
among the weakest winners in a maximally suppressive (canonical) manipulation.
As noted in Section~\ref{sec:model}, these two score segments are disjoint for
all positional scoring rules.

Our first structural result shows that every successful manipulation can be
transformed into a canonical form in which all influential positions are
allocated exclusively to the targeted outsiders and the targeted weak winners.
In other words, only the extremal positions of the coalition ballots matter.

\begin{lemma}[Canonical Manipulation]
\label{lem:canonical}
If displacement at level~$k'$ is feasible under a positional scoring rule,
then there exists a successful manipulation in which
\begin{enumerate}[label=(\roman*)]
    \item every top-$k'$ position on every coalition ballot is assigned to a
          candidate in $\mathcal{O}^\ast$, and
    \item  every bottom-$k'$ position on every coalition ballot is assigned to a candidate in $\mathcal{T}_w$.
\end{enumerate}
\end{lemma}

\begin{proof}[Proof sketch]
Let $\pi$ be any successful manipulation achieving displacement~$k'$.
Throughout, the honest profile and its induced score ordering are held fixed;
only the coalition’s ballots are transformed.

Suppose some top-$k'$ position on some ballot $\pi_v$ is filled by a candidate
$c \notin \mathcal{O}^\ast$.  Let $i\le k'$ be the \emph{largest} index such that the
occupant $c:=\pi_v[i]$ satisfies $c\notin \mathcal{O}^\ast$.  By maximality of $i$,
every position $j\in\{i+1,\dots,k'\}$ is occupied by a candidate in $\mathcal{O}^\ast$.
Since $|\mathcal{O}^\ast|=k'$, there exists some $o\in\mathcal{O}^\ast$ that appears
below position $k'$ on the same ballot.  Using adjacent swaps, move $c$ down to
position $k'$ and then move $o$ up to position $k'$, which places $o$ into the
top-$k'$ block and pushes $c$ to position $k'+1$.

This operation
\begin{itemize}
    \item increases (or preserves) the coalition’s contribution to $o$,
    \item decreases (or preserves) the coalition’s contribution to $c$
      (which cannot hurt feasibility, and only helps if $c\in \mathcal{T}_w$), and
    \item does not increase the contribution to any $t \in \mathcal{T}_w$
          (in Stage~1 only candidates from $\mathcal{O}^\ast$ move upward past $c$,
          and in Stage~2 all candidates passed by $o$ move downward).
\end{itemize}
Thus all inequalities $S_o + \Delta_o \ge B$ and $S_t + \Delta_t \le B-1$
remain satisfied, while the number of misallocated top positions decreases.
A symmetric argument then fixes misallocated bottom-$k'$ positions, starting from a profile where (i) already holds,
so iterating these local fixes yields a manipulation satisfying (i)--(ii).
A full proof is given in Appendix~\ref{sec:ProofLemma1}.
\end{proof}

By Lemma~\ref{lem:canonical}, we may assume without loss of generality that all
meaningful influence is concentrated in the top-$k'$ and bottom-$k'$ positions
of each coalition ballot.
The top-$k'$ positions form the coalition’s \emph{boost resource} for
$\mathcal{O}^\ast$; the bottom-$k'$ positions form its \emph{suppress resource} for
$\mathcal{T}_w$.
Consequently, score levels in the middle positions play no role in determining
feasibility and may be filled arbitrarily.

Crucially, the boost and suppress resources draw from disjoint sets of ballot
positions and affect disjoint sets of candidates.
As a result, the feasibility of boosting the strongest outsiders and the
feasibility of suppressing the weakest winners are independent under positional
scoring rules, interacting only through the common displacement parameter~$k'$.

\subsection{Independent Boost and Suppress Feasibility}
\label{sec:independent}

Given a queried  displacement level~$k'$ and a proposed cutoff~$B$, the
coalition’s requirements decompose according to the two disjoint score pools
identified in Section~\ref{sec:canonical}:  
the $mk'$ score copies occupying the top-$k'$ positions (the \emph{boost resources}) and the $mk'$
score copies occupying the bottom-$k'$ (the \emph{suppress resources}).  
For each outsider $o \in \mathcal{O}^\ast$, the coalition must contribute at least
\[
\Delta_o \;\ge\; b_o := \max\{0,\, B - S_o\},
\]
while for each weak winner $t \in \mathcal{T}_w$, it must ensure
\[
\Delta_t \;\le\; u_t := \max\{0,\, B - 1 - S_t\}.
\]

Under the canonical form guaranteed by Lemma~\ref{lem:canonical}, all boost
contributions must come from the top-$k'$ positions of the coalition ballots,
and all suppress contributions must come from the bottom-$k'$ positions.  
These two positional segments correspond to disjoint subsets of the scoring
vector~$p$.  
Consequently, satisfying the boost requirements and satisfying the suppress
requirements draw on completely separate score resources.  
The two subproblems therefore impose no mutual constraints.

\begin{lemma}[Independent Feasibility]
\label{lem:independent}
Fix $k'$ and a candidate cutoff~$B$, and let $b = (b_o)_{o\in\mathcal{O}^\ast},$ and $u = (u_t)_{t\in\mathcal{T}_w}$ denote the
boost and suppress requirement vectors induced by~$B$.  
Suppose that:
\begin{enumerate}[label=(\arabic*)]
    \item the coalition can allocate its $mk'$ highest scores so that each
    outsider $o \in \mathcal{O}^\ast$ receives at least $b_o$
          \textup{(boost feasibility)}, and
    \item the coalition can allocate its $mk'$ lowest scores so that each weak
    winner $t \in \mathcal{T}_w$ receives at most $u_t$
          \textup{(suppress feasibility)}.
\end{enumerate}
Then the coalition can simultaneously satisfy all requirements induced by the
cutoff~$B$.
\end{lemma}

\begin{proof}[Proof sketch]
By Lemma~\ref{lem:canonical}, boost contributions to $\mathcal{O}^\ast$ must draw
exclusively from the $mk'$ score copies associated with the top-$k'$ ballot
positions, and suppress contributions to $\mathcal{T}_w$ must draw exclusively
from the $mk'$ score copies associated with the bottom-$k'$ ballot positions.  
Because these two pools correspond to disjoint ballot positions, any feasible
allocation for one pool cannot reduce or interfere with the resources available
to the other, even if some numerical score values coincide.  
Hence boost feasibility and suppress feasibility can always be achieved
together.  
A full proof appears in Appendix~\ref{sec:ProofLemma2}.
\end{proof}

Lemma~\ref{lem:independent} reduces feasibility at cutoff~$B$ to two independent
tests: whether the demand vector~$b$ is achievable using only the top-$k'$ score
pool, and whether the demand vector~$u$ is achievable using only the
bottom-$k'$ score pool.  
Section~\ref{sec:geometry} provides a geometric characterization of each
subproblem, showing that both reduce to checking whether a demand vector is
majorized by an appropriate prefix-capacity profile.

\subsection{Guaranteed Displacement}
\label{sec:guaranteed}

The previous lemmas establish that, for any proposed cutoff~$B$, the boost and
suppress requirements can be satisfied independently using the coalition’s two
disjoint score pools.  These results verify that the coalition can produce the
\emph{required} score changes, but they do not yet imply that \emph{every}
completion of the coalition’s ballots will yield the desired replacement of
$k'$ winners.  
The next theorem closes this gap: strict separation at the cutoff $B$ is sufficient to guarantee displacement
level $k'$---i.e., every completion of the coalition's ballots (and any tie-breaking) yields a final Top-$k$
winner set containing \emph{at least $k'$} outsiders, and hence excluding at least $k'$ honest winners.

\begin{theorem}[Guaranteed displacement at level $k'$]
\label{thm:guaranteed-topk}
Fix a displacement level $k'$ and a cutoff $B$. Let $\mathcal{O}^\ast$ be the $k'$ strongest outsiders and $\mathcal{T}_w$ the $k'$
weakest winners under the honest profile. Suppose the coalition can assign the top-$k'$ and bottom-$k'$
positions of its ballots so as to ensure final scores $F_c = S_c + \Delta_c$ satisfying
\[
F_o \ge B \quad \text{for all } o \in \mathcal{O}^\ast, \qquad
F_t \le B-1 \quad \text{for all } t \in \mathcal{T}_w.
\]
Then, for every completion of the remaining (middle) positions of the coalition's ballots and under any
tie-breaking rule, the final Top-$k$ winner set contains at least $k'$ outsiders. Equivalently, at least $k'$
members of the honest winner set $T$ are excluded from the final Top-$k$.
\end{theorem}

Consequently, the coalition deterministically displaces at least $k'$ honest winners in every completion.
When the scoring vector contains plateaus, the identities of the displaced incumbents (and the entering
outsiders) may depend on tie-breaking; our guarantee is therefore about the displacement \emph{level} $k'$,
not necessarily the exact displaced set.

\begin{proof}[Proof sketch]
Let $\mathcal{W}$ be the final Top-$k$ winner set under an arbitrary completion and tie-breaking. Suppose for
contradiction that $\mathcal{W}$ contains at most $k'-1$ outsiders. Then $\mathcal{W}$ contains at least $k-k'+1$ members of the
honest winner set $\mathcal{T}$. Since $|\mathcal{T} \setminus \mathcal{T}_w| = k-k'$, this forces $\mathcal{W}$ to contain some
$t \in \mathcal{T}_w$.

But every $o \in \mathcal{O}^\ast$ satisfies $F_o \ge B > B-1 \ge F_t$, hence $F_o > F_t$. Under the Top-$k$ rule, a
candidate with strictly higher score than a selected candidate must also be selected (tie-breaking only applies
among equal scores). Therefore all $k'$ candidates in $\mathcal{O}^\ast$ must belong to $\mathcal{W}$, contradicting that
$\mathcal{W}$ had at most $k'-1$ outsiders. Hence $\mathcal{W}$ contains at least $k'$ outsiders, and at least $k'$ honest winners are
displaced. A full proof appears in Appendix~\ref{sec:ProofGuranteed}.
\end{proof}

The next theorem establishes the converse direction: if either the boost or
suppress requirements fails at cutoff~$B$, then displacement at level~$k'$
cannot be achieved.

\begin{theorem}[Impossibility when independent feasibility fails]
\label{thm:impossibility}
Fix a displacement level $k'$ and a cutoff~$B$.
If at least one of the following fails:
\begin{enumerate}[label=(\alph*),leftmargin=1.6em]
    \item the boost requirements for $\mathcal{O}^\ast$ are feasible using the
          coalition’s top-$k'$ score positions, or
    \item the suppress requirements for $\mathcal{T}_w$ are feasible using the
          coalition’s bottom-$k'$ score positions,
\end{enumerate}
then no manipulation by the coalition can achieve displacement level~$k'$ at
cutoff~$B$.
\end{theorem}

\begin{proof}[Proof sketch]
By Lemma~\ref{lem:canonical}, we may without loss of generality assume that 
all boost contributions to $\mathcal{O}^\ast$ come from the top-$k'$ positional 
scores, and all suppress contributions to $\mathcal{T}_w$ come from the 
bottom-$k'$ scores; these score pools are disjoint and cannot substitute for 
each other.

If the boost requirements are not feasible using the top-$k'$ scores, then some
$o \in \mathcal{O}^\ast$ must satisfy $F_o < B$, making separation at~$B$
impossible.  
If the suppress requirements are not feasible using the bottom-$k'$ scores,
then some $t \in \mathcal{T}_w$ must satisfy $F_t \ge B$, again preventing
strict separation at~$B$.

Thus no manipulation can achieve displacement at level~$k'$ for this cutoff.
A full proof appears in Appendix~\ref{sec:ProofImpossibility}.
\end{proof}

\begin{corollary}[Feasibility characterization]
Displacement at level~$k'$ is achievable \textbf{if and only if} there exists a cutoff
$B$ such that both the boost and suppress requirements are feasible at $(k',B)$.
\end{corollary}

Consequently, feasibility of the boost and suppress systems completely
determines the coalition’s ability to achieve displacement at level~$k'$:
the coalition can successfully and robustly guarantee the replacement of at
least $k'$ winners at cutoff~$B$ if and only if both subproblems are feasible
at that cutoff.

\subsection{Feasible Cutoff Interval and Monotonicity}
\label{sec:interval}

The final structural ingredient is the monotone behavior of boost and suppress
feasibility with respect to the cutoff~$B$.  
As $B$ increases, outsiders require more additional score while weak winners
are permitted more score; as $B$ decreases, the reverse holds.  
This monotonic relationship ensures that feasible cutoffs form a single
contiguous interval.

\begin{lemma}[Monotonicity of boost and suppress feasibility]
\label{lem:monotonicity}
Fix a displacement level~$k'$. Then:
\begin{itemize}
    \item \textbf{Boost feasibility is monotone downward:}  
    if the boost requirements are feasible at some cutoff~$B$, then they remain
    feasible at every cutoff $B' \le B$.

    \item \textbf{Suppress feasibility is monotone upward:}  
    if the suppress requirements are feasible at some cutoff~$B$, then they 
    remain feasible at every cutoff $B' \ge B$.
\end{itemize}
\end{lemma}

\begin{proof}[Proof sketch]
Raising the cutoff~$B$ requires strictly more score for every outsider, but
permits strictly more score for every weak winner.  
Thus boost feasibility can only become harder as $B$ increases, while suppress
feasibility can only become easier.  
Lowering $B$ has the symmetric effect.  
A detailed proof appears in Appendix~\ref{sec:ProofMonotonicity}.
\end{proof}

Combining these two monotonicity directions yields a simple and useful
interval structure.

\begin{corollary}[Feasible cutoff interval]
\label{cor:interval}
For any displacement level $k'$, the set of cutoffs~$B$ that satisfy \emph{both}
boost and suppress feasibility forms a (possibly empty) integer interval
\[
[B_{\min}(k'),\, B_{\max}(k')].
\]
The displacement level $k'$ is feasible if and only if the integer interval
$[B_{\min}(k'),\,B_{\max}(k')]$ is nonempty.
\end{corollary}

This interval structure is a key component of the search procedures in
Section~\ref{sec:algorithm}.  
For each fixed $k'$, an inner binary search over~$B$ identifies the feasible
cutoff interval, and an outer search over~$k'$ determines the maximum achievable
displacement $k^\ast$ for which the feasible interval is nonempty.

\paragraph{Summary.}
At a fixed displacement level~$k'$, feasibility reduces to two independent
tasks—boosting the strongest outsiders and suppressing the weakest winners—
each of which is tested via ordered prefix-majorization inequalities
(Section~\ref{sec:geometry}).  
The monotone dependence on the cutoff~$B$ ensures that feasible cutoffs form a
single contiguous interval, enabling the efficient dual binary-search algorithm
developed in Section~\ref{sec:algorithm}.

\section{Geometric and Lattice Characterization}
\label{sec:geometry}

\subsection{Overview}
\label{sec:geometry-overview}

This section develops the geometric machinery that underlies the
boost and suppress feasibility tests from Section~\ref{sec:decomposition}.
Although boosting and suppressing are posed as combinatorial score-allocation
problems, the coalition’s aggregate capabilities have a clean geometric
structure: each ballot contributes a permutahedron, and the coalition’s total
feasible region is the Minkowski sum of these permutahedra.  Understanding this
region is the key to determining when, for a fixed queried displacement level $k'$, a $k'$-dimensional demand vector
$d$---standing for either $b(B)$ or
$u(B)$---can be realized.

Our analysis proceeds in two layers.

\paragraph{(1) Continuous layer: the Block--HLP envelope.}
We first derive the convex hull of the coalition’s feasible region by expressing
each ballot’s score-allocation possibilities as a permutahedron and taking their
Minkowski sum (Section~\ref{sec:geometry-convex}).  The Block--HLP theorem (Section~\ref{sec:block-hlp}) shows that
this continuous envelope is characterized exactly by a family of blockwise
Hardy--Littlewood--Pólya prefix inequalities.  These inequalities hold for
\emph{every} positional scoring rule, providing a universal necessary condition
for boosting and suppressing.  This continuous envelope serves as the geometric
backbone for the discrete characterization developed later.

\paragraph{(2) Discrete layer: exact AP-lattice structure.}
When all colluding ballots follow AP ladders with a
\emph{common step size}~$g$, the integer-feasible points inside the Block--HLP
envelope admit an exact description.  In this common-step regime, an aggregate
vector is realizable if and only if it satisfies (i) the
Block--HLP prefix inequalities and (ii) a single coordinatewise congruence
condition modulo~$g$.  Thus the realizable
region is a \emph{majorization-and-congruence lattice}, giving a clean prefix-and-lattice
structure for integer score redistributions.  A constructive packing argument
then yields a simple polynomial-time oracle for realizability.  By contrast,
when step sizes differ across ballots, integer-feasible points must lie in a
more intricate AP semigroup, and the prefix-and-congruence characterization is
no longer sufficient—a distinction discussed later in Section~\ref{sec:geometry-ap}.

\paragraph{Using the geometric characterization for boosting and suppressing.}
After characterizing the realizable region itself, we apply the geometry to the
boost and suppress subproblems. These are not simple membership tests.
Boosting requires checking whether there exists a realizable aggregate vector
that dominates the boost demand vector $b(B)$
componentwise—that is, whether the coalition has enough high-score capacity to
meet all required increases.  
Suppression is the exact symmetric condition: we must check whether there
exists a realizable aggregate vector that is componentwise no larger than the
suppression allowance vector $u(B)$—i.e., whether the
coalition can keep each weak winner’s score below its permitted threshold.
Although the dominance directions differ, both checks ultimately rely on the
same geometric description of the realizable region.

\medskip
The remainder of this section proceeds as follows.
We first characterize the coalition’s convex envelope via the
Minkowski sum of ballot-wise permutahedra
(Section~\ref{sec:geometry-convex}), and then derive its
facet description through the Block--HLP prefix inequalities
(Section~\ref{sec:block-hlp}).  
We next refine this continuous region to obtain an exact discrete
characterization for AP scoring rules
(Section~\ref{sec:geometry-ap}), and finally use these structural
tools to develop the boost and suppress feasibility oracles
(Section~\ref{subsec:AP-feasible-demands}).

\subsection{The Minkowski-Sum Convex Envelope}
\label{sec:geometry-convex}

To determine whether boosting or suppressing is feasible at a fixed cutoff~$B$,
we must understand the set of all $k'$-dimensional score vectors that the $m$
colluding voters can collectively assign to the $k'$ designated targets.  While
each ballot contributes scores in a discrete and rank-constrained way, the
continuous relaxation of these constraints has a clean geometric structure:
each ballot induces a permutahedron—the convex hull of all permutations of its
score vector—and the coalition’s continuous feasible region is the Minkowski
sum of these permutahedra.  This subsection characterizes that convex region
and shows that it is governed entirely by Hardy–Littlewood–Pólya (HLP)
prefix-sum inequalities.

\paragraph{Ballot-level structure.}
Fix a displacement level $k'$ and its corresponding set of $k'$ targets (the
strongest outsiders for boosting or the weakest winners for suppression).
For each colluding ballot $v \in [m]$, let
$r^{(v)} = (r^{(v)}_1, \dots, r^{(v)}_{k'})$
denote the multiset of positional scores that ballot~$v$ assigns to these
targets, written in nonincreasing order
$r^{(v)}_1 \ge \cdots \ge r^{(v)}_{k'}$
(each $r^{(v)}_i \in \{p_1,\dots,p_x\}$).
Any feasible redistribution of these scores corresponds to applying a 
permutation $\sigma \in S_{k'}$ to the coordinates of $r^{(v)}$, where 
$S_{k'}$ is the set of all permutations of $\{1,\dots,k'\}$.  
Each such $\sigma$ acts by
\[
\sigma(r^{(v)}) := (r^{(v)}_{\sigma(1)},\dots,r^{(v)}_{\sigma(k')}).
\]
For clarity, we distinguish between the discrete set of score permutations and
its convex hull.  Define
\[
  \mathcal{P}(r^{(v)}) := \{\, \sigma(r^{(v)}) : \sigma \in S_{k'} \,\}
  \qquad\text{and}\qquad
  \Pi(r^{(v)}) := \mathrm{conv}(\mathcal{P}(r^{(v)})),
\]
so $\Pi(r^{(v)})$ is the permutahedron generated by $r^{(v)}$.

This polytope $\Pi(r^{(v)})$ captures exactly the continuous 
score-allocation possibilities of ballot~$v$ to the $k'$ targets.

\paragraph{A running example ($k'=3$, $m=2$).}
To anchor the geometry, consider two colluding ballots.
Suppose the highest three scores on each ballot are
\[
\text{Ballot 1: } (9,5,1),
\qquad
\text{Ballot 2: } (10,6,2).
\]
Thus $r^{(1)} = (9,5,1)$ and $r^{(2)} = (10,6,2)$.

\emph{Permutahedron for Ballot~1.}
The six permutations of $(9,5,1)$,
\[
(9,5,1),\ (9,1,5),\ (5,9,1),\ (5,1,9),\ (1,9,5),\ (1,5,9),
\]
describe every way ballot~1 can distribute its three scores across the three
targets.  Their convex hull, $\Pi(r^{(1)})$, is a 2-dimensional permutahedron
(a hexagon) embedded in $\mathbb{R}^3$; the same holds for $\Pi(r^{(2)})$.

\medskip
The coalition’s continuous feasible region is the Minkowski sum
\[
\Pi(r^{(1)}) + \Pi(r^{(2)}),
\]
a centrally symmetric 2-dimensional polytope in $\mathbb{R}^3$ whose shape is
fully determined by the prefix sums of the sorted vectors $(9,5,1)$ and
$(10,6,2)$.  
The next subsection shows that this Minkowski sum is exactly captured by a
family of Block--HLP prefix-sum inequalities applied to these sorted vectors.

\paragraph{Block structure induced by the scoring vector.}
Under any positional scoring rule, each ballot contributes exactly one copy of
each score $p_1,\dots,p_x$.  For the $k'$ designated targets, ballot~$v$ can
assign them any $k'$ scores drawn from this multiset, subject only to the
ordering constraint imposed by the ranking.  When these $k'$ scores are written
in nonincreasing order as
$r^{(v)}_1 \ge \cdots \ge r^{(v)}_{k'}$,
the quantity
\[
R^{(v)}_t := \sum_{i=1}^t r^{(v)}_i, \qquad t = 1,\dots,k',
\]
captures the maximum total score ballot~$v$ can give to any chosen set of $t$
targets.  We refer to $R^{(v)}$ as the \emph{block of prefix capacities} for
ballot~$v$.

\paragraph{The coalition’s convex envelope.}
Across the $m$ colluding ballots, the coalition’s continuous feasible region is
the Minkowski sum
\[
\mathcal{P}_{\mathrm{conv}}
:= \sum_{v=1}^m \Pi(r^{(v)})
= \bigl\{\, y = \sum_{v=1}^m y^{(v)} : y^{(v)} \in \Pi(r^{(v)}) \,\bigr\}.
\]
Every vector in $\mathcal{P}_{\mathrm{conv}}$ is a continuous aggregate of
score-changes on the $k'$ targets; conversely, every continuous score-change
vector feasible for the coalition lies in this Minkowski sum.

\paragraph{Running example (continued).}
Consider again the two score ladders
\[
r^{(1)} = (9,5,1), \qquad r^{(2)} = (10,6,2),
\]
with prefix capacities
\[
R^{(1)} = (9,14,15), \qquad R^{(2)} = (10,16,18).
\]
The coalition’s aggregate prefix capacities are therefore
\[
R = R^{(1)} + R^{(2)} = (19,30,33),
\]
so the total available mass is
\[
T = R_3 = 33.
\]

For any ballot $v$, the continuous feasible region $\Pi(r^{(v)})$ is the
permutahedron
\[
  \Pi(r^{(v)}) = \operatorname{conv}(\mathcal{P}(r^{(v)})).
\]
This permutahedron admits the following majorization description:
\[
\Pi(r^{(v)}) =
\bigl\{
   y^{(v)} \in \mathbb{R}^3 :
   y^{(v)}_1 \ge y^{(v)}_2 \ge y^{(v)}_3,\;
   \sum_{i=1}^t y^{(v)}_i \le R^{(v)}_t\ (t=1,2),\;
   \sum_{i=1}^3 y^{(v)}_i = R^{(v)}_3
\bigr\}.
\]
In particular, every point in $\Pi(r^{(1)})$ has total sum $15$ and every point
in $\Pi(r^{(2)})$ has total sum $18$.

\medskip
The coalition’s continuous feasible region is the Minkowski sum
\[
\mathcal{P}_{\mathrm{conv}}
   = \Pi(r^{(1)}) + \Pi(r^{(2)})
   = \bigl\{\, y^{(1)} + y^{(2)} :
               y^{(1)} \in \Pi(r^{(1)}),\
               y^{(2)} \in \Pi(r^{(2)}) \,\bigr\}.
\]
Because Minkowski sums preserve majorization inequalities, every
$y \in \mathcal{P}_{\mathrm{conv}}$ is sorted and satisfies
\[
\sum_{i=1}^t y_i \le R_t \quad (t=1,2),
\qquad
\sum_{i=1}^3 y_i = R_3 = 33.
\]
Thus $\mathcal{P}_{\mathrm{conv}}$ is a two-dimensional polygon embedded in the
hyperplane
\[
\{\, y \in \mathbb{R}^3 : y_1 \ge y_2 \ge y_3,\ \sum_i y_i = 33 \,\},
\]
cut out by the prefix constraints
$\sum_{i=1}^t y_i \le R_t$ for $t=1,2$.

Geometrically, the coalition’s feasible region is therefore a nondegenerate
two-dimensional polytope, not a single line segment.  In general, the
Block--HLP conditions define a higher-dimensional polygonal envelope, and only
when additional structural constraints (such as common-step AP scoring,
introduced later) are imposed do the \emph{integer-feasible points} collapse
onto a thin lattice embedded within this envelope.

\paragraph{Prefix characterization.}
A classical fact in majorization theory is that permutahedra—and therefore any
Minkowski sum of permutahedra—are fully characterized by inequalities on the
prefix sums of their sorted coordinates.  In our setting, the prefix-capacity
vectors $R^{(v)}$ play the role of the fundamental building blocks: their sum
determines the supporting hyperplanes of $\mathcal{P}_{\mathrm{conv}}$.
The next theorem \emph{restates} this classical majorization characterization
in the blockwise prefix-capacity notation used by our displacement oracle.

\subsection{Block--HLP Prefix Inequalities}
\label{sec:block-hlp}

We now give a complete description of the coalition’s continuous feasible
region $\mathcal{P}_{\mathrm{conv}} = \sum_{v=1}^m \Pi(r^{(v)})$, the Minkowski
sum of the ballot-wise permutahedra introduced in
Section~\ref{sec:geometry-convex}.  
In the positional-scoring setting, each ballot’s permutahedron has a simple
block structure determined by its prefix capacities, and summing these blocks
yields the exact supporting hyperplanes of the coalition’s convex envelope.
The resulting \emph{Block--HLP prefix inequalities} give a sharp and complete
characterization of $\mathcal{P}_{\mathrm{conv}}$.

The coalition’s achievable aggregate score vectors form a Minkowski sum of
permutahedra, yielding a symmetric polymatroid base polytope; because this
region is invariant under permutation of coordinates, all subset constraints
reduce to cardinality-based prefix constraints.

\paragraph{Notation.}
For any vector $y \in \mathbb{R}^{k'}$, let $y^\downarrow$ denote its
nonincreasing rearrangement, and let
\[
R_t := \sum_{v=1}^m R^{(v)}_t
\qquad (t = 1,\dots,k')
\]
denote the coalition’s total prefix capacities, where
$R^{(v)}_t = r^{(v)}_1 + \cdots + r^{(v)}_t$ was defined in
Section~\ref{sec:geometry-convex}.

\begin{theorem}[Block--HLP Convex Envelope (classical majorization form)]
\label{thm:block-hlp}
Fix $m,k' \in \mathbb{N}$ and let 
$R_t$ $(t=1,\dots,k')$ denote the coalition’s aggregate prefix capacities,
with total capacity $R_{k'}$.

A vector $y \in \mathbb{R}^{k'}$ lies in the coalition’s convex envelope
$\mathcal{P}_{\mathrm{conv}}$ if and only if it satisfies:
\begin{enumerate}[label=(\Alph*)]
\item \textbf{Prefix-sum bounds:}  
      \[\sum_{i=1}^t y_i^\downarrow \le R_t 
         \quad \text{for all } t=1,\dots,k'-1.\]

\item \textbf{Total sum:}  
      \[
\sum_{i=1}^{k'} y_i \;=\; R_{k'}.
\]
\end{enumerate}
\end{theorem}

\paragraph{Provenance and role in this paper.}
Theorem~\ref{thm:block-hlp} is a direct corollary of classical results on
majorization and permutahedra (e.g., Hardy--Littlewood--P\'olya and Rado),
together with standard polymatroid-base descriptions of Minkowski sums of
permutahedra (e.g., Edmonds).
We include it in this explicit ``blockwise prefix capacity'' form because it
yields an $O(k'\log k')$ membership test for $\mathcal{P}_{\mathrm{conv}}$
(after sorting), which is the continuous feasibility oracle used throughout our
boost/suppress envelope framework.
Our main new contributions begin after this envelope: we show how displacement
feasibility reduces to two independent blockwise subproblems and we give an
\emph{exact} integer realizability characterization for common-step AP ladders
via a single congruence condition (Theorem~\ref{thm:ap-lattice-common}).

\paragraph{Interpretation.}
Condition~(A) expresses that no $t$ targets can
collectively receive more score than the top $t$ scores available across all
ballots, regardless of how those scores are permuted among the targets.
Condition~(B) enforces conservation of the total score
mass contributed by the coalition.  
Together, these inequalities cut out a
$(k'-1)$-dimensional polytope whose extreme points are exactly the aggregates
obtainable by choosing specific permutations on each ballot.

\begin{proof}[Proof sketch]
For each ballot $v$, its permutahedron
$\Pi(r^{(v)}) = \operatorname{conv}(\mathcal{P}(r^{(v)}))$
is exactly the majorization polytope associated with $r^{(v)}$.    
By the Hardy--Littlewood--Pólya theorem~\cite[Thm.~A.1]{hardy1934}, a vector $y^{(v)}$ lies in 
$\Pi(r^{(v)})$ if and only if it satisfies
\[
\sum_{i=1}^t (y^{(v)})_i^\downarrow \le R^{(v)}_t
\quad\text{for all } t=1,\dots,k'-1,
\qquad\text{and}\qquad
\sum_{i=1}^{k'} (y^{(v)})_i = R^{(v)}_{k'}.
\]
Summing these inequalities over all ballots $v=1,\dots,m$ yields the necessity
of Conditions~(A) and~(B) for any aggregate 
$y = \sum_v y^{(v)}$.

For sufficiency, assume $y$ satisfies Conditions~(A) and~(B).
Greene--Kleitman rank inequalities (equivalently, Rado’s theorem) imply that a
vector belongs to the Minkowski sum $\sum_v \Pi(r^{(v)})$ precisely when its
sorted prefix sums are dominated by the aggregated prefix bounds $R_t$ and its
total mass matches $R_{k'}$.  
Since these two requirements are exactly Conditions~(A) and~(B), any such $y$
can be decomposed into a sum of vectors $y^{(v)} \in \Pi(r^{(v)})$, establishing
the claim.
\end{proof}

\paragraph{Running example (continued).}
Recall the two ballots with
\[
r^{(1)} = (9,5,1), \qquad r^{(2)} = (10,6,2),
\]
whose prefix capacities are
\[
R^{(1)} = (9,14,15), 
\qquad
R^{(2)} = (10,16,18),
\]
so the coalition’s aggregate capacities are
\[
R = (19,\; 30,\; 33).
\]

Consider the aggregate vector
\[
y = (18, 9, 6).
\]
Its sorted form is $y^\downarrow = (18,9,6)$, and its prefix sums are
\[
(18,\; 27,\; 33).
\]
Comparing with $R$,
\[
18 \le 19, \qquad 27 \le 30, \qquad 33 = 33,
\]
we see that $y$ satisfies all Block--HLP inequalities and therefore lies within
the continuous feasible envelope $\mathcal{P}_{\mathrm{conv}}$.

By contrast, the vector
\[
y' = (20, 8, 5)
\]
has prefix sums $(20,28,33)$, violating the first inequality
\[
20 \le 19.
\]
Hence $y' \notin \mathcal{P}_{\mathrm{conv}}$.
This illustrates how the Block--HLP prefix inequalities precisely determine
which continuous aggregate allocations are feasible for the coalition.

\medskip
The Block--HLP theorem gives a tight and efficiently checkable description of
the coalition's \emph{continuous} capabilities: it characterizes the entire
Minkowski-sum convex envelope via a minimal set of prefix inequalities.

\paragraph{Remark.}
The convex envelope $\mathcal{P}_{\mathrm{conv}}$ is a cardinality-based
polymatroid: its facets depend only on prefix sizes, and the Block--HLP
inequalities give an irredundant, efficiently checkable description.  
Feasibility in the continuous relaxation therefore reduces to a single pass over
the sorted vector $y^\downarrow$.

\paragraph{Integrality considerations.}
While Theorem~\ref{thm:block-hlp} fully captures the \emph{continuous} region,
actual ballot contributions are integer-valued and must respect the ladder
structure of the scoring rule.  As a result, not every integer point in
$\mathcal{P}_{\mathrm{conv}}$ is realizable: discrete score gaps induce residue
classes that the convex relaxation cannot express.

\medskip
We now refine this continuous picture for AP scoring
ladders.  For these rules, realizability is characterized exactly by the
Block--HLP prefix inequalities together with a single global congruence
constraint, yielding a precise lattice description of feasible integer
aggregates.

\subsection{Discrete Geometry for Arithmetic–Progression Ladders}
\label{sec:geometry-ap}

The convex Block--HLP envelope described above characterizes all aggregate score
vectors achievable under a \emph{continuous relaxation} of manipulation, in
which each ballot’s score vector may be fractionally redistributed among
positions. Equivalently, we replace each discrete permutation set 
$\mathcal{P}(r^{(v)})$ with its convex hull 
$\Pi(r^{(v)}) = \operatorname{conv}(\mathcal{P}(r^{(v)}))$,
treating each ballot as a divisible resource.  The Minkowski sum of these convexified ladders
is a polymatroid base polytope whose integer points contain---but generally
strictly overapproximate---the truly realizable aggregates.

In practice, ballot scores come from fixed integer ladders; the coalition can only choose elements of $\mathcal{P}(r^{(v)})$ rather than any
point of $\Pi(r^{(v)})$, so only a subset of the integer points in the
continuous Block--HLP polytope are realizable.  To obtain an \emph{exact} integer-feasibility
test for positional scoring rules, we specialize to the broad and practically
important class of \emph{arithmetic–progression (AP) ladders}.  This family includes many common positional rules such as Borda, truncated Borda, plurality, and $3$--$2$--$1$
scoring (among others with a small number of score levels and regular gaps).

Under AP structure, the continuous Block--HLP envelope refines to a discrete
feasibility region governed by two ingredients:
\begin{itemize}
\item the usual \emph{prefix-sum inequalities}, and
\item a finite set of \emph{allowable coordinate values} induced by the ladder
      step sizes.
\end{itemize}
When all colluding ballots share the same step size, these discrete restrictions
collapse to a particularly clean \emph{prefix–and–congruence lattice} that
exactly characterizes realizability.  For heterogeneous step sizes, the
situation is more intricate and depends delicately on both the step sizes and
the ladder length~$k'$; we provide the appropriate necessary condition below.

\paragraph{Notation.}
We use $F(t)$ to denote the coalition’s aggregate prefix capacities for
$t=1,\dots,k'$ throughout.
(Sections~\ref{sec:geometry-convex}--\ref{sec:block-hlp} previously used the
notation $R_t$ for the same quantity; thus $F(t)\equiv R_t$.)

\paragraph{Discrete realizability vs.\ convex envelope.}
Recall from Section~\ref{sec:geometry-convex} that
$P(r^{(v)}) := \{\sigma(r^{(v)}) : \sigma \in S_{k'}\}$ is the \emph{discrete} set of ballot-wise permutations,
and $\Pi(r^{(v)}) :=\operatorname{conv}(P(r^{(v)}))$ is its permutahedron.
We will use the shorthand
\[
\mathcal{P}^{\mathrm{disc}} := \sum_{v=1}^m P(r^{(v)}),
\qquad
\mathcal{P}^{\mathrm{conv}} := \sum_{v=1}^m \Pi(r^{(v)}),
\]
where $\sum$ denotes Minkowski sum.
An integer vector $y \in \mathbb{Z}^{k'}$ is \emph{(ballot-)realizable} if $y \in \mathcal{P}^{\mathrm{disc}}$, i.e.,
$y=\sum_{v=1}^m \sigma_v(r^{(v)})$ for some permutations $\sigma_v \in S_{k'}$.
By construction, $\mathcal{P}^{\mathrm{disc}} \subseteq \mathcal{P}^{\mathrm{conv}}$.

\paragraph{AP ladders.}
An AP ladder is defined by sampling scores from a single arithmetic progression.

\begin{definition}[Arithmetic–Progression (AP) Ladder]
\label{def:ap-ladder}
An AP ladder of length $k'$ is a nonincreasing integer vector
$r\in\mathbb{Z}^{k'}$ for which there exist integers $L\in\mathbb{Z}$, $g\ge1$,
and integers $\ell_1\ge\ell_2\ge\cdots\ge\ell_{k'}\ge0$ such that
\[
  r_j = L + g\,\ell_j, \qquad j=1,\dots,k'.
\]
Equivalently, the distinct values of $r$ form the set
$\{L,\,L+g,\,L+2g,\dots,L+Hg\}$ for some $H\le k'-1$, with arbitrary multiplicity
subject to the nonincreasing order.
\end{definition}

For a coalition of $m$ voters, the $v$th ballot has baseline $L_v$, step
size $g_v$, and level multiset
$\Lambda_v := \{\ell^{(v)}_1,\dots,\ell^{(v)}_{k'}\}$.
Each feasible permutation of that ballot is drawn from $\mathcal{P}(r^{(v)})$.

\paragraph{Example.}
Almost all standard positional rules satisfy Definition~\ref{def:ap-ladder}.
For example, with $k'=10$:
\begin{itemize}
\item Borda $(9,8,\ldots,0)$: baseline $L=0$, step $g=1$.
\item Truncated Borda $(4,3,2,1,0,0,\ldots)$: repeated levels allowed.
\item The 3--2--1 rule $(3,3,3,2,2,2,1,1,1,0)$: step size $1$.
\item Plurality $(1,0,\ldots,0)$: two levels, still AP.
\item A large-gap rule $(8,2,2,\ldots,2)$: baseline $L=2$, step $g=6$.
\end{itemize}
A vector such as $(8,5,4,\ldots)$ is \emph{not} AP, as no single step size
generates all pairwise differences.

\paragraph{Continuous envelope (always exact).}
For \emph{any} ladders---AP or otherwise—the Minkowski sum of ballot-wise
permutahedra is characterized by the Block--HLP inequalities:
\[
y\in\sum_v \Pi(r^{(v)}) 
\iff
\begin{cases}
 \sum_{i=1}^t y^\downarrow_i \le F(t), & t=1,\dots,k',\\[4pt]
 \sum_{i=1}^{k'} y_i = F(k').
\end{cases}
\tag{$\star$}
\label{eq:block-hlp-discrete}
\]
The remaining question is: \emph{which integer points of this polytope are
actually realizable?}

\paragraph{Discrete constraints for heterogeneous step sizes.}
For AP ladders, each coordinate of a realizable aggregate is a sum of one
selected score from each ballot.  
Define the per-ballot score sets
\[
S_v := \{\,L_v + g_v\ell : \ell \in \Lambda_v\,\},
\qquad
S := S_1 + \cdots + S_m
     := \{\,s_1+\cdots+s_m : s_v \in S_v\,\}.
\]
Thus $S$ is the (finite) set of all allowable coordinate values for realizable
aggregates.

Let
\[
g_* := \gcd(g_1,\dots,g_m)
\]
denote the common period of the step sizes.  
Every realizable coordinate must lie in the residue class
$L_{\mathrm{tot}} \pmod{g_*}$, but the attainable values form only the strict
subset $S \subseteq \{L_{\mathrm{tot}} + g_* z : z \in \mathbb{Z}\}$.  
Hence the coarse modular obstruction is governed by $g_*$, while the true
arithmetic granularity of the mixed-step setting is determined by the set $S$.

\begin{proposition}[Discrete necessary condition for mixed-step AP ladders]
If $y \in \mathcal{P}^{\mathrm{disc}}=\sum_v P(r^{(v)})$, then $y$ satisfies $(\star)$ and $y_i \in S$ for all
$i=1,\ldots,k'$.
\end{proposition}

The set $S$ captures the true arithmetic granularity induced by the step sizes
$g_1,\dots,g_m$ and the ladder length~$k'$.  In general $S$ may contain only a
strict subset of the residues modulo $\gcd(g_1,\dots,g_m)$, and characterizing
the full integer hull in this mixed-step regime is subtle; a complete treatment
remains open.

\paragraph{Exact lattice structure for common step size.}
A complete solution emerges when all colluding ballots share a \emph{common}
step size. Intuitively, when all ladders share a common step size $g$, the discrete feasibility obstruction collapses to a single residue class constraint: after factoring out 
$g$, the remaining structure behaves as a unit-step lattice, for which prefix dominance is sufficient

\begin{theorem}[AP--Ladder Lattice Theorem (Common Step Size)]
\label{thm:ap-lattice-common}
Suppose each ballot has step size $g\ge1$, i.e.,
$r^{(v)}_j = L_v + g\,\ell^{(v)}_j$. Let $L_{\mathrm{tot}} := \sum_v L_v$, and
let $F(t)$ denote the aggregate prefix capacities.  
Then an integer vector $y \in \mathbb{Z}^{k'}$ is (ballot-)realizable, i.e.\ $y \in
\mathcal{P}^{\mathrm{disc}}=\sum_v P(r^{(v)})$, if and only if:
\begin{enumerate}[label=(\Alph*),itemsep=2pt]
  \item \textbf{Prefix-sum bounds:} 
        $\displaystyle \sum_{i=1}^t y^\downarrow_i \le F(t)$  \text{for all } t=1,\dots,k'-1;
  \item \textbf{Total sum:} 
        $\displaystyle \sum_{i=1}^{k'} y_i = F(k')$;
  \item \textbf{Coordinatewise congruence:}
        $\displaystyle y_i \equiv L_{\mathrm{tot}} \pmod g$ for all $i=1,\dots,k'$.
\end{enumerate}
These conditions describe exactly the integer points of a polymatroid base
polytope intersected with a single residue class modulo~$g$.
\end{theorem}

\noindent\textbf{Intuition.}
Conditions (A)--(B) describe the continuous Block--HLP envelope; condition (C) selects the single lattice coset
compatible with the common AP step $g$. The proof reduces realizability to a unit-step decomposition on the AP levels.

\begin{proof}[Proof sketch]
\emph{Necessity.}
Conditions (A)--(B) are exactly the Block--HLP prefix constraints and total-mass equality
for the Minkowski-sum envelope (Theorem~\ref{thm:block-hlp}).
For (C), note that every ballot has the form
$r^{(v)}_j = L_v + g\,\ell^{(v)}_j$ with $\ell^{(v)}_j \in \mathbb{Z}$, so every coordinate contributed by ballot $v$
is congruent to $L_v \pmod g$. Summing over $v=1,\dots,m$ forces
$y_i \equiv L_{\mathrm{tot}} := \sum_{v=1}^m L_v \pmod g$ for every coordinate $i$.

\emph{Sufficiency.}
Assume $y \in \mathbb{Z}^{k'}$ satisfies (A)--(C). By (C) we can write
\[
y \;=\; L_{\mathrm{tot}}\mathbf{1} + g z
\qquad\text{for some } z \in \mathbb{Z}^{k'}.
\]
Let $\ell^{(v)} = (\ell^{(v)}_1,\dots,\ell^{(v)}_{k'})$ be ballot $v$'s (sorted) level vector, so that
$r^{(v)} = L_v\mathbf{1} + g\,\ell^{(v)}$.
Define the induced \emph{level} prefix capacities
\[
F^{\mathrm{lev}}(t)
\;:=\;
\frac{F(t)-tL_{\mathrm{tot}}}{g}
\;=\;
\sum_{v=1}^m \sum_{i=1}^t \ell^{(v)}_i,
\qquad t=1,\dots,k'.
\]
Dividing (A)--(B) by $g$ and subtracting baselines yields the unit-step Block--HLP system
\[
\sum_{i=1}^t z_i^\downarrow \le F^{\mathrm{lev}}(t)\quad (t=1,\dots,k'-1),
\qquad
\sum_{i=1}^{k'} z_i = F^{\mathrm{lev}}(k').
\]
At this point we are in a purely \emph{integer level-allocation} problem: each ballot $v$ contributes a
permutation of $\ell^{(v)}$, and we want their coordinatewise sum to equal $z$.
A standard multilevel majorization theorem (equivalently, the Brualdi--Hwang / multilevel Gale--Ryser
construction via nested 0--1 ``layers'') says that these prefix inequalities are not only necessary but also
\emph{sufficient} for such an integer decomposition: one can assign levels to targets across ballots so that
each ballot uses exactly the multiset $\ell^{(v)}$ and the total received by target $i$ equals $z_i$.

Finally, mapping levels back to scores by $L_v\mathbf{1} + g(\cdot)$ gives permutations of each $r^{(v)}$ whose
Minkowski sum equals $y$. Appendix~\ref{sec:ProofAP}. provides the full constructive decomposition.
\end{proof}

When step sizes are equal, the semigroup $S$ collapses to a single congruence
class modulo~$g$, and the Block--HLP inequalities fully determine the discrete
feasible region.  This yields a clean, efficiently testable lattice structure
that we exploit in our algorithmic framework.  
With heterogeneous step sizes, the integer-feasible region is instead governed
by the finer AP--semigroup condition $y_i\in S$, which is strictly stronger than
mere congruence modulo $\gcd(g_v)$ and reflects the true arithmetic
granularity of the mixed-step setting.

\begin{figure}[t]
    \centering
    \begin{subfigure}[b]{0.48\textwidth}
        \centering
        \includegraphics[width=\textwidth]{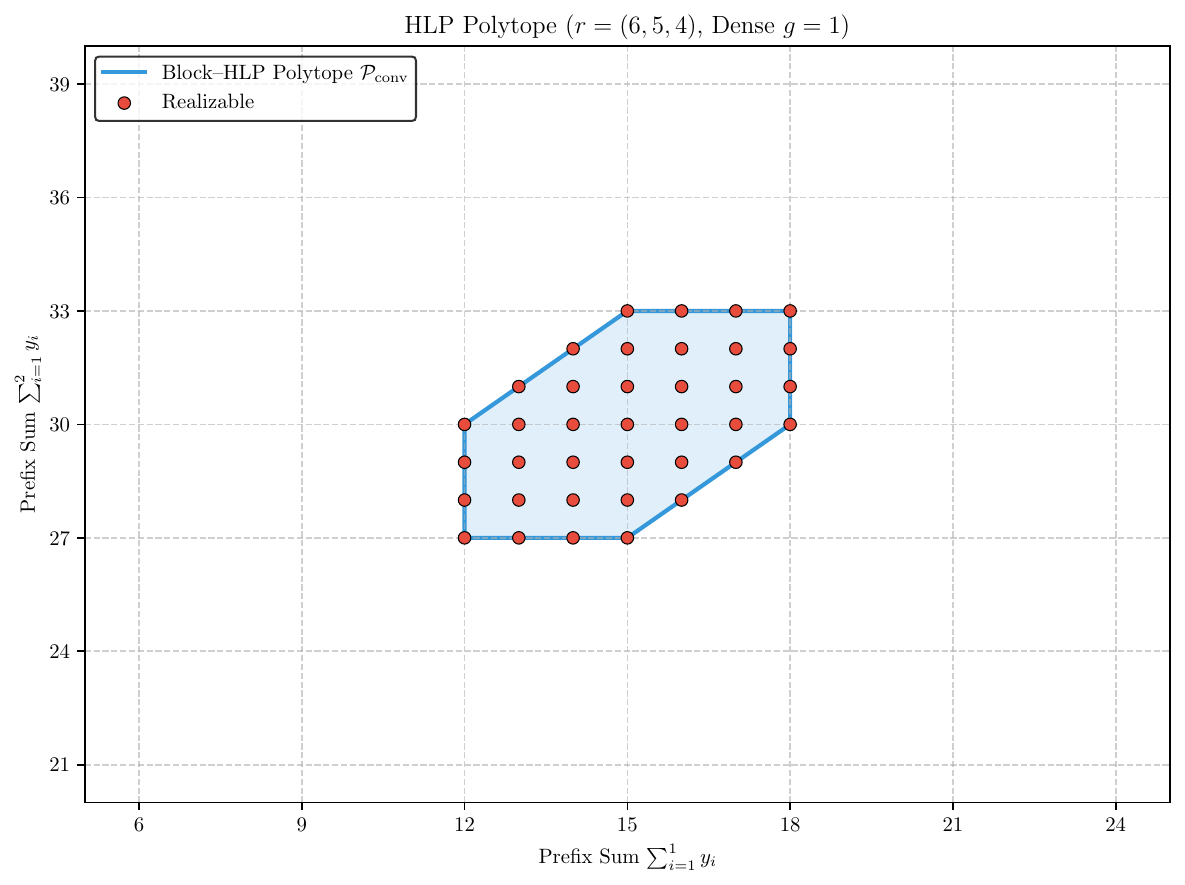}
        \caption{Dense Lattice ($g=1$)}
        \label{fig:dense_lattice}
    \end{subfigure}
    \hfill
    \begin{subfigure}[b]{0.48\textwidth}
        \centering
        \includegraphics[width=\textwidth]{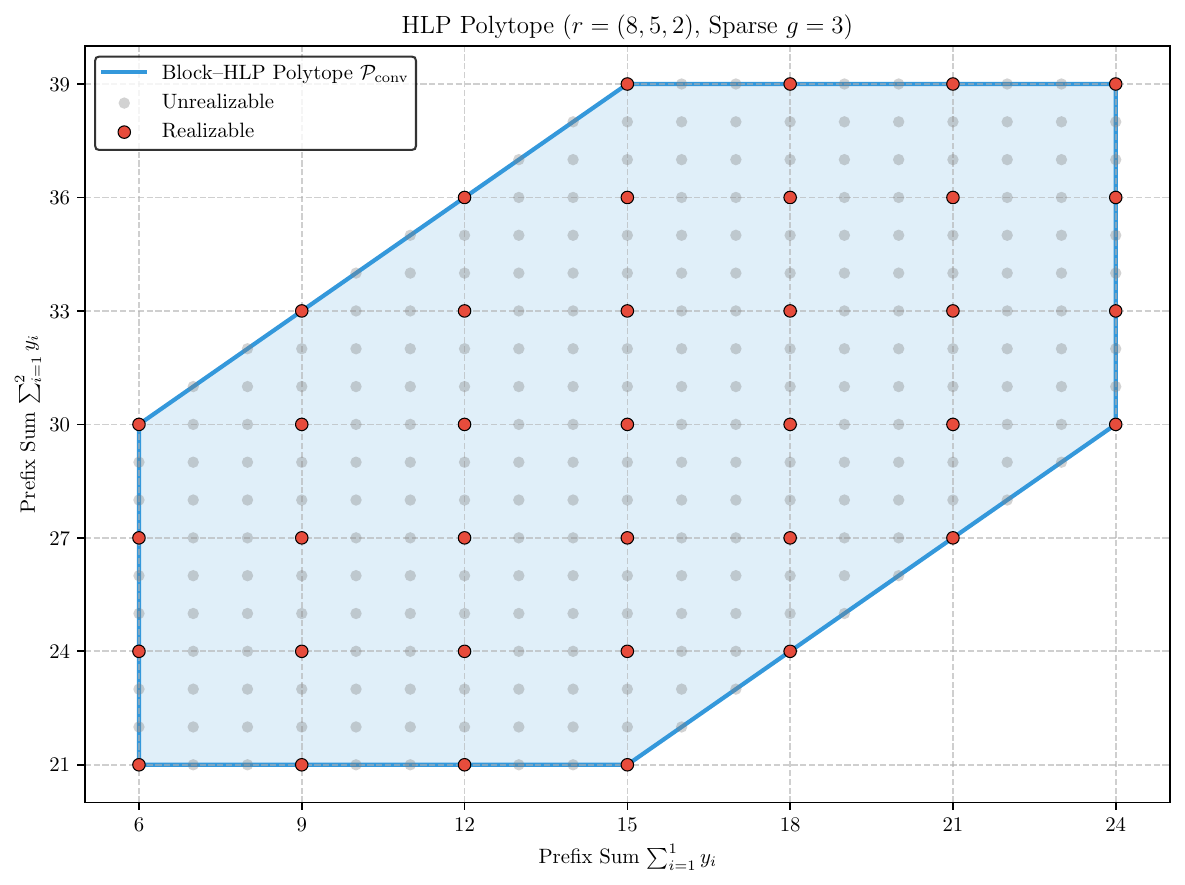}
        \caption{Sparse Lattice ($g=3$)}
        \label{fig:sparse_lattice}
    \end{subfigure}
    
   \caption{
\textbf{Majorization vs.\ Lattice Realizability.}
Each panel shows a two-dimensional projection of the Block--HLP polytope
$\mathcal{P}_{\mathrm{conv}}$ (blue) for $m=3$ and $k'=3$ onto the prefix-sum
coordinates $(y_1,\; y_1+y_2)$, together with integer points in the projected
region.
Both figures are plotted using the same coordinate range to facilitate direct
visual comparison.
(a) \emph{Unit-step AP ladder ($g=1$).}
For the unit-step ladder $r=(6,5,4)$, the congruence constraint of the
AP--Ladder Lattice Theorem is vacuous.
Consequently, every integer point satisfying the Block--HLP prefix inequalities
is realizable, and the discrete feasible set (red) fills the convex envelope.
(b) \emph{Sparse-step AP ladder ($g=3$).}
For the common-step ladder $r=(8,5,2)$ with step size $g=3$, realizable aggregates
are restricted to a single residue class modulo~3.
Accordingly, only a strict sublattice of integer points (red) is realizable,
while other interior points (grey) satisfy all Block--HLP prefix inequalities
yet violate the congruence condition.
Together, the panels illustrate that Block--HLP majorization determines the
continuous feasibility envelope, while arithmetic step size imposes an
additional modular constraint on integer realizability.
}
    \label{fig:hlp_lattice_comparison}
\end{figure}

\paragraph{Geometric validation.}
Figure~\ref{fig:hlp_lattice_comparison} illustrates how 
AP structure refines the continuous Block--HLP feasibility envelope into a
discrete prefix-and-congruence lattice.
Each panel shows a two-dimensional projection of the Block--HLP polytope
$\mathcal{P}_{\mathrm{conv}}$ onto the plane of the first two prefix sums
$(y_1,\; y_1+y_2)$, together with all integer points in the projected region.
To facilitate visual comparison, both panels are plotted using the same
coordinate range.

In the \emph{unit-step} regime (Figure~\ref{fig:dense_lattice}), we use the
unit-step AP ladder $r=(6,5,4)$ with step size $g=1$.
By the unit-step case of Theorem~\ref{thm:ap-lattice-common}, the congruence
constraint is vacuous, and every integer point satisfying the Block--HLP prefix
inequalities is realizable.
Accordingly, the discrete feasible region (red points) fills the entire convex
envelope predicted by continuous majorization.

In the \emph{sparser-step} regime (Figure~\ref{fig:sparse_lattice}), we use the
common-step AP ladder $r=(8,5,2)$ with step size $g=3$.
In this case, realizable aggregates are restricted to a single residue class
modulo~3, as prescribed by the AP--Ladder Lattice Theorem.
As a result, only a strict sublattice of integer points (red) is realizable,
while other interior points (grey) satisfy all Block--HLP prefix inequalities yet
violate the congruence condition and cannot arise from ballot permutations.

Together, these panels confirm that Block--HLP majorization determines the
continuous feasibility envelope, while AP step size imposes an additional
modular constraint on integer realizability, producing a strict refinement of
the continuous envelope whenever $g>1$.

\paragraph{Running example (continued).}
In our running example, both score ladders have the same step size $g = 4$ and
baselines $L_1 = 1$ and $L_2 = 2$, so the total baseline is
$L_{\mathrm{tot}} = 3$.  
Hence every realizable integer aggregate $y$ must satisfy the coordinatewise
congruence condition
\[
y_i \equiv 3 \pmod 4
\qquad\text{for all } i = 1,2,3.
\]

The continuous Block--HLP envelope is determined by the aggregate prefix
capacities
\[
F(1)=19,\qquad F(2)=30,\qquad F(3)=33.
\]
Thus the realizable integer points are exactly those
$y \in \mathbb{Z}^3$ that satisfy these prefix inequalities and whose
coordinates all lie in the residue class $3 \bmod 4$.

For example,
\[
y = (19, 11, 3)
\]
satisfies both the prefix bounds
\[
19 \le 19,\qquad 19 + 11 = 30 \le 30,\qquad 33 = 33,
\]
and the congruence condition ($19 \equiv 11 \equiv 3 \equiv 3 \pmod 4$),
so it is realizable by some pair of permutations of $(9,5,1)$ and $(10,6,2)$.

In contrast,
\[
y' = (18, 11, 4)
\]
lies strictly inside the Block--HLP polytope (prefix sums $18 \le 19$,
$29 \le 30$, total $33$), but violates the congruence requirement because
neither $18$ nor $4$ is congruent to $3 \pmod 4$.  
Thus $y'$ cannot arise from any pair of ballot permutations.

Geometrically, the realizable region is the intersection of the continuous
Block--HLP polytope with the lattice slice
\[
\{\, y \in \mathbb{Z}^3 : y_i \equiv 3 \pmod 4 \text{ for all } i = 1,2,3 \,\},
\]
producing a discrete pattern of feasible integer points inside the convex
envelope.

\begin{corollary}[Unit-step ladders: no congruence restriction]
\label{cor:unit-step-structural}
If all AP ladders have a common step size $g = 1$, then the
congruence condition~\textup{(C)} in
Theorem~\ref{thm:ap-lattice-common} is vacuous.
An integer vector $y$ is therefore realizable as
$y \in \sum_v \Pi(r^{(v)})$ if and only if it satisfies the
Block--HLP prefix inequalities~\textup{(A)} and the total-sum
condition~\textup{(B)}.
\end{corollary}

\paragraph{Geometric intuition.}
Each AP ladder restricts the integer points of its ballot permutahedron to a
one-dimensional arithmetic lattice, whose Minkowski sum is a lattice polytope
thinly embedded within the continuous Block--HLP envelope.  
Conditions~\textup{(A)}--\textup{(B)} describe this outer envelope, while
condition~\textup{(C)} selects the unique lattice coset containing all realizable
integer aggregates.  
When $g=1$, this lattice coincides with the full integer grid, so the continuous
and discrete feasible regions agree.

\paragraph{The Borda rule as the cleanest special case.}
The pure Borda rule is the canonical instance of the unit-step regime:
every ballot’s ladder has a common step size $g = 1$.  
Under Borda, realizability becomes a \emph{pure} majorization test governed solely
by the Block--HLP inequalities, with no additional arithmetic structure.  
This makes the most widely used positional scoring rule also the simplest and
most transparent example of our geometric characterization.

\paragraph{Scope and complexity considerations.}
The AP--Ladder Lattice Theorem shows that whenever all ballots follow
arithmetic–progression ladders with a common step size, the coalition’s feasible
region is an exact prefix–and–congruence lattice.  
This yields both a complete structural characterization and a polynomial-time
realizability oracle, which our boosting and suppressing procedures
(Section~\ref{subsec:AP-feasible-demands}) rely on directly.

Outside the AP regime, however, this lattice structure typically breaks down.
Prefix majorization alone is no longer sufficient for integer realizability, as
the feasible region fragments according to the irregular score gaps imposed by
non-AP ladders.  
Understanding these non-AP settings—how feasibility degrades as ladders depart
from arithmetic progressions, when approximate tests remain tractable, and
which generalized scoring families preserve partial lattice structure—remains an
open and compelling direction.  
Section~\ref{sec:discussion} revisits this broader landscape and proposes a
three-layer structural taxonomy (AP, bounded deviation, unbounded deviation),
together with a conjectured complexity frontier separating tractable and
intractable realizability regimes.

\subsection{Feasibility Tests for Boosting and Suppression}
\label{subsec:AP-feasible-demands}

The AP--Ladder Lattice Theorem
(Theorem~\ref{thm:ap-lattice-common}) gives a complete description of all
realizable $k'$-dimensional aggregate score-change vectors 
$y \in \sum_v \Pi(r^{(v)})$ when all colluding ballots share a common step
size~$g$.  
In the boost and suppress subproblems of
Section~\ref{sec:independent}, the coalition must decide whether some
realizable aggregate can (i) \emph{raise} every targeted outsider above a
specified threshold, or (ii) \emph{keep} every targeted weak winner below a
specified limit.

Under common-step AP ladders, these feasibility questions reduce to a clean
prefix-majorization test once the demand vector is mapped into the unique
congruence class $L_{\mathrm{tot}} + g\mathbb{Z}$ in which all realizable
aggregates lie.  
This yields fast and exact oracles that form the backbone of the dual
binary-search framework in Section~\ref{sec:algorithm}.

\paragraph{Boosting feasibility.}
Given a nonincreasing boost-demand vector 
$b \in \mathbb{Z}_{\ge0}^{k'}$, the coalition asks whether there exists a
realizable aggregate $y$ such that $y_i \ge b_i$ for all~$i$. 
Because realizable coordinates must lie in the congruence class
$L_{\mathrm{tot}} + g\mathbb{Z}$, we first raise each $b_i$ to the smallest
value in this residue class that is at least $b_i$, thereby producing a
lattice-adjusted demand vector that any feasible aggregate must dominate.

\begin{lemma}[Boost feasibility under common-step AP ladders]
\label{lem:ap-boost}
Assume all colluding ballots follow AP ladders with a common step size $g$ and
total baseline $L_{\mathrm{tot}}$.
Let $b \in \mathbb{Z}_{\ge0}^{k'}$ be a nonincreasing boost-demand vector, and
define the lattice-adjusted lower-demand vector
\[
\hat b_i 
   := b_i + \bigl( (L_{\mathrm{tot}} - b_i) \bmod g \bigr),
   \qquad i = 1,\dots,k'.
\]
Thus $\hat b_i$ is the smallest integer $\hat b_i \ge b_i$ satisfying
$\hat b_i \equiv L_{\mathrm{tot}} \pmod g$.

There exists a realizable aggregate $y$ with $y_i \ge b_i$ for all~$i$ if and
only if the sorted vector $\hat b^\downarrow$ satisfies
\[
\sum_{i=1}^t \hat b^\downarrow_i \;\le\; F(t)
\quad\text{for all } t=1,\dots,k'-1,
\qquad
\sum_{i=1}^{k'} \hat b^\downarrow_i \;\le\; F(k').
\]
Equivalently, $\hat b^{\downarrow}$ lies in the polymatroid defined by the Block--HLP capacities $(F(t))$ and
can be extended to an \emph{integer} vector $y$ that saturates the total-sum equality; by Theorem~\ref{thm:ap-lattice-common} (and the
congruence ensured by the lattice adjustment), such a $y$ is (ballot-)realizable, i.e.\ $y \in \sum_v P(r^{(v)})$.

\end{lemma}

\begin{proof}[Proof sketch]
\textbf{Necessity.}
Any realizable aggregate $y$ satisfies the Block--HLP prefix bounds and
$y_i \equiv L_{\mathrm{tot}} \pmod g$ for all $i=1,\dots,k'$.
Thus $y_i \ge b_i$ for all $i$ implies $y_i \ge \hat b_i$ for all $i$, and hence
$\hat b^\downarrow$ satisfies the required prefix inequalities.

\medskip
\textbf{Sufficiency.}
If $\hat b^\downarrow$ satisfies the Block--HLP prefix bounds and has total sum
at most $F(k')$, the polymatroid-extension property yields a vector $y$ with
$y_i \ge \hat b_i$ for all $i=1,\dots,k'$, total sum $F(k')$, and the correct
congruence modulo $g$.
By Theorem~\ref{thm:ap-lattice-common}, such a $y$ is realizable and satisfies
$y_i \ge b_i$ for all $i=1,\dots,k'$.
A full proof appears in Appendix~\ref{sec:ProofBoostFeasibility}.
\end{proof}

\paragraph{Suppression feasibility.}
For suppression, the coalition must determine whether there exists a realizable
aggregate satisfying $y_i \le u_i$ for all~$i$.  
Since realizable points must lie in the congruence class
$L_{\mathrm{tot}} + g\mathbb{Z}$, we first push each upper bound
\emph{downward} to the nearest feasible value in this class.

\begin{lemma}[Suppression feasibility under common-step AP ladders]
\label{lem:ap-suppress}
Assume all colluding ballots follow AP ladders with a common step size $g$ and
total baseline $L_{\mathrm{tot}}$.
Let $u \in \mathbb{Z}_{\ge0}^{k'}$ be a nonincreasing suppression-demand vector,
and define the lattice-adjusted upper bound
\[
\hat u_i 
   := u_i - \bigl( (u_i - L_{\mathrm{tot}}) \bmod g \bigr),
   \qquad i = 1,\dots,k'.
\]
Thus $\hat u_i$ is the largest integer $\hat u_i \le u_i$ satisfying
$\hat u_i \equiv L_{\mathrm{tot}} \pmod g$.

Define $d := F(k')\mathbf{1} - \hat u$.  
There exists a realizable aggregate $y$ satisfying
$y_i \le u_i$ for all $i=1,\dots,k'$
if and only if the sorted vector $d^\downarrow$ satisfies
\[
\sum_{i=1}^t d^\downarrow_i \;\le\; F(t)
\quad\text{for all } t=1,\dots,k'-1,
\qquad
\sum_{i=1}^{k'} d^\downarrow_i \;\le\; F(k').
\]
Equivalently, the slack vector $d$ satisfies the Block--HLP prefix
inequalities and can be extended upward to a realizable aggregate whose
coordinates are at least $d_i$ for all $i=1,\dots,k'$; negating the
transformation yields a realizable aggregate $y$ with
$y_i \le u_i$ for all $i=1,\dots,k'$.
\end{lemma}

\begin{proof}[Proof sketch]
If a realizable aggregate $y$ satisfies
$y_i \le u_i$ for all $i=1,\dots,k'$, then by congruence we may assume
$y_i \le \hat u_i$ for all $i=1,\dots,k'$.
Define the corresponding slack vector
$d := F(k')\mathbf{1} - \hat u$.
Then $d$ is coordinatewise dominated by
$z := F(k')\mathbf{1} - y$, i.e.,
$z_i \ge d_i$ for all $i=1,\dots,k'$.
Moreover, $z$ satisfies the Block--HLP prefix inequalities
\[
\sum_{i=1}^t z_i^\downarrow \le F(t)
\quad\text{for all } t=1,\dots,k'.
\]
This implies the stated prefix bounds for $d^\downarrow$.

\medskip
Conversely, suppose that $d^\downarrow$ satisfies the stated prefix bounds.
Applying the same polymatroid-extension argument as in the boost-feasibility
case, there exists a vector $z$ such that
\[
z_i \ge d_i \quad\text{for all } i=1,\dots,k',
\qquad
\sum_{i=1}^{k'} z_i = F(k'),
\]
and
$z_i \equiv L_{\mathrm{tot}} \pmod g$ for all $i=1,\dots,k'$.
Applying Theorem~\ref{thm:ap-lattice-common} to $z$ yields a realizable
aggregate
\[
y := F(k')\mathbf{1} - z
\]
satisfying
$y_i \le u_i$ for all $i=1,\dots,k'$.
A full proof appears in Appendix~\ref{sec:ProofSupressFeasibility}.
\end{proof}

Thus suppression feasibility mirrors the boost case: a lattice adjustment
followed by a single majorization (Block--HLP prefix) test.

\paragraph{Interpretation.}
Boosting and suppression differ only in the direction of the lattice
adjustment: boosting lifts demands \emph{up} to the nearest feasible value in
the lattice congruence class $L_{\mathrm{tot}} + g\mathbb{Z}$, while
suppression pushes them \emph{down} to the nearest such value.  
Once adjusted, both problems become identical: each reduces to checking whether
the adjusted demand vector lies inside the Block--HLP prefix polymatroid.
These oracles are the key components enabling the dual binary-search procedure
of Section~\ref{sec:algorithm}, which computes the maximum displacement $k^\ast$
and the corresponding feasible cutoff interval.

\paragraph{Running example (continued).}
In our running example, both AP ladders have step size $g=4$ and baselines
$L_1=1$, $L_2=2$, so every realizable aggregate must satisfy
\[
y_i \equiv L_{\mathrm{tot}} = 3 \pmod 4.
\]
The coalition's prefix capacities are
\[
F(1)=19,\qquad F(2)=30,\qquad F(3)=33.
\]

\medskip
\noindent\emph{A feasible boost demand.}
Consider the boost-demand vector
\[
b = (17, 8, 3).
\]
All entries already lie in the congruence class $3 \pmod 4$, so
$\hat b = b$.  
Sorting gives $\hat b^\downarrow = (17,8,3)$, whose prefix sums are
\[
17 \le 19,\qquad 17+8 = 25 \le 30,\qquad 25+3 = 28 \le 33.
\]
Thus $\hat b^\downarrow$ lies inside the Block--HLP polymatroid, and because
$28 < 33$, the polymatroid-extension property ensures it can be extended to a
realizable aggregate $y \ge b$.  
Hence $b$ is a \emph{feasible} boost-demand vector.

\medskip
\noindent\emph{An infeasible boost demand.}
Now consider
\[
b' = (18, 9, 4).
\]
Lattice adjustment raises $b'$ to the nearest point congruent to $3 \pmod 4$:
\[
\hat b' = (19,11,7), \qquad \hat b'^\downarrow = (19,11,7).
\]
We check:
\[
19 \le 19,\qquad 19+11 = 30 \le 30,\qquad 19+11+7 = 37 > 33.
\]
Because the total mass $37$ exceeds the coalition’s capacity $F(3)=33$,
no realizable aggregate $y \ge b'$ exists.  
Thus $b'$ is \emph{infeasible}.

\medskip
These two cases illustrate how the lattice-adjusted majorization test
cleanly distinguishes feasible from infeasible demands:
first ensure the demand lies in the correct congruence class, then apply the
prefix-sum and total-sum comparison with the Block--HLP capacities.

\medskip
When all AP ladders share a common step size $g=1$, the congruence condition in
Theorem~\ref{thm:ap-lattice-common} disappears and every realizable integer
vector is governed solely by the Block--HLP prefix inequalities.  
In this unit-step regime—covering Borda, truncated Borda, plurality, and other
affine rules—boosting and suppression reduce directly to checking whether the
demand vector lies in the Block--HLP polymatroid.

\begin{corollary}[Dense lattice under unit-step ladders]
\label{cor:unit-step-dense}
If all AP ladders have common step size $g=1$, then every integer point in the
Block--HLP polytope is realizable.  
Equivalently, the discrete feasible region coincides exactly with the
continuous Block--HLP envelope.
\end{corollary}

\begin{corollary}[Boosting and suppression under unit-step ladders]
\label{cor:unit-step-demand}
Let $b$ (boost) and $u$ (suppress) be nonincreasing demand vectors.
If all AP ladders have common step size $g = 1$, then:

\begin{enumerate}[label=(\roman*)]
\item There exists a realizable aggregate $y$ with $y_i \ge b_i$ for all $i$  if and
only if
\[
\sum_{i=1}^{t} b_i \le F(t) 
\quad\text{for all } t=1,\dots,k'-1,
\qquad
\sum_{i=1}^{k'} b_i \le F(k').
\]

\item There exists a realizable aggregate $y$ with $y_i \le u_i$ for all $i$ if and
only if the vector $(-u)$ satisfies the inequalities in~(i) when applied with
respect to the negated ladders $\{-r^{(v)}\}$.  
Equivalently, suppression feasibility holds exactly when the ``slack'' vector
$F(k')\mathbf{1} - u$ lies in the Block--HLP polymatroid.
\end{enumerate}
\end{corollary}

\paragraph{Boosting vs.\ suppression in practice.}
For boosting, the AP ladders $r^{(v)}$ consist of the top $k'$ scores on each
ballot; for suppression, they consist of the bottom $k'$ scores.  
In both cases, these score vectors are valid AP ladders with baselines $L_v$ and
step sizes $g_v$, so the AP--Ladder Lattice Theorem and the feasibility oracles
above apply directly.  
These oracles are invoked repeatedly inside the dual binary search over the
cutoff~$B$ and the displacement level~$k'$ in
Section~\ref{sec:algorithm}, where they determine whether a candidate cutoff or
displacement level is feasible.

\section{Algorithmic Framework}
\label{sec:algorithm}

This section develops the algorithm for testing displacement feasibility at a
queried level~$k'$ and for computing the maximum achievable displacement
$k^\ast$ under AP-ladder scoring rules with a \emph{common step size}~$g$.

\subsection{Overview}

A displacement at level~$k'$ is achievable if and only if there exists an
integer cutoff~$B$ satisfying two conditions:
(i) each of the $k'$ strongest outsiders $\mathcal{O}^\ast$ can be boosted to
score at least~$B$, and
(ii) each of the $k'$ weakest winners $\mathcal{T}_w$ can be held strictly
below~$B$.
As $B$ increases, condition~(i) becomes strictly harder while condition~(ii)
becomes strictly easier.
Thus the feasible cutoffs of the two subproblems form monotone intervals
\[
(-\infty, B_{\max}]
\qquad\text{and}\qquad
[B_{\min}, \infty),
\]
and displacement is feasible exactly when these intervals overlap, i.e.,
$B_{\min} \le B_{\max}$.

\paragraph{Cutoff search domain.}
Any feasible cutoff must lie between the smallest final score a weak winner
can attain (after receiving only the coalition’s lowest ranks) and the largest
final score an outsider can attain (after receiving the highest ranks):
\[
B_{\mathrm{low}}
  = \min_{t\in\mathcal{T}_w}(S_t + m p_x),
\qquad
B_{\mathrm{high}}
  = \max_{o\in\mathcal{O}^\ast}(S_o + m p_1).
\]
Binary searches for $B_{\min}$ and $B_{\max}$ are therefore restricted to the
tight interval $[B_{\mathrm{low}}, B_{\mathrm{high}}]$.

\paragraph{Strict separation.}
Our feasibility notion requires the coalition to enforce a strict gap between
outsiders and weak winners.  
Manipulations that succeed only via favorable tie-breaking at~$B$ do not count
as successful displacement.  
This ensures robustness to arbitrary tie-resolution rules.

\paragraph{Feasibility oracles.}
Under common-step AP ladders, realizable score vectors lie in the
prefix-and-congruence lattice characterized in
Theorem~\ref{thm:ap-lattice-common}.  
Given a cutoff~$B$, we form the demand vectors
\[
b_i = \max\{0,\; B - S_{o_i}\}, \qquad
u_i = \max\{0,\; B - 1 - S_{t_i}\},
\]
representing the minimum coalition mass needed to boost each outsider and the
maximum mass each weak winner may receive while remaining below~$B$.
Each coordinate is then adjusted into the correct congruence class modulo~$g$:
\[
\hat b_i = b_i + \bigl((L_{\mathrm{tot}} - b_i) \bmod g\bigr), \qquad
\hat u_i = u_i - \bigl((u_i - L_{\mathrm{tot}}) \bmod g\bigr).
\]
After sorting these vectors in nonincreasing order, feasibility reduces to
checking the Block--HLP prefix inequalities
\[
\sum_{i=1}^t y_i^\downarrow \le F(t),
\qquad
\sum_{i=1}^{k'} y_i = T,
\]
together with the AP-congruence condition, which is automatically satisfied by
construction.  
In the unit-step case (i.e., when the \emph{common} step size is $g=1$), the
congruence condition vanishes and the test reduces purely to Block--HLP.

\paragraph{Pipeline for testing displacement at level~$k'$.}
The structural results above imply a clean algorithmic flow:
\begin{enumerate}[leftmargin=1.5em]
    \item Identify the $k'$ outsiders $\mathcal{O}^\ast$ and $k'$ weak winners
          $\mathcal{T}_w$.
    \item Compute the search interval $[B_{\mathrm{low}}, B_{\mathrm{high}}]$.
    \item Use the exact Block--HLP feasibility oracle on $b(B)$ to find
          $B_{\max}$ via binary search.
    \item Use the same oracle on $u(B)$ to find
          $B_{\min}$ via binary search.
    \item Conclude that displacement at level~$k'$ is feasible if and only if
          $B_{\min} \le B_{\max}$.
\end{enumerate}

Because each feasibility-oracle call sorts a length-$k'$ vector, it runs in
$O(k'\log k')$ time. Each cutoff binary search uses $O(\log(mx))$ probes, and we run
two such searches (for $B_{\min}$ and $B_{\max}$), so envelope evaluation at a fixed
level $k'$ takes
\[
O\!\left(k'\log k' \cdot \log(mx)\right).
\]
The outer binary search over $k' \in \{0,\dots,k\}$ contributes a factor $\log k$.
In the worst case $k' \le k$, hence the overall running time is
\[
O\!\left(k\log k \cdot \log(mx)\cdot \log k\right)
= O\!\left(k(\log k)^2\log(mx)\right),
\]
which is polynomial in all relevant election parameters.

\subsection{Constructing Ordered Demand Vectors}
\label{sec:ConstructDemand}

Given a cutoff $B$, the coalition must ensure that  
(i) every outsider $o \in \mathcal{O}^\ast$ reaches a final score of at least~$B$, and  
(ii) every weak winner $t \in \mathcal{T}_w$ remains strictly below~$B$.  
These requirements induce a \emph{boost-demand} vector for the outsiders and a
\emph{suppression-tolerance} vector for the weak winners.  
For the prefix-sum feasibility tests of Section~\ref{sec:geometry}, only the
\emph{sorted} forms of these vectors matter.  The following routine produces
these canonical ordered vectors.

\begin{algorithm}[t]
\footnotesize
\DontPrintSemicolon
\caption{\textsc{ComputeDemandVectors}$(S,\mathcal{O}^\ast,\mathcal{T}_w,B)$}
\label{alg:demand2e}
\KwIn{honest scores $(S_c)_{c}$; outsider set $\mathcal{O}^\ast$; weak-winner set $\mathcal{T}_w$; cutoff $B$}
\KwOut{$b(B)^{\downarrow}$ (boost demands, nonincreasing), $u(B)^{\uparrow}$ (suppression tolerances, nondecreasing).}
\BlankLine

$b \gets [\,\max\{0,\; B - S_o\} \mid o \in \mathcal{O}^\ast\,]$\;
$u \gets [\,\max\{0,\; B - 1 - S_t\} \mid t \in \mathcal{T}_w\,]$\;

$b(B)^{\downarrow} \gets \textsc{SortNonIncreasing}(b)$\;
$u(B)^{\uparrow} \gets \textsc{SortNonDecreasing}(u)$\;

\Return $(b(B)^{\downarrow},\, u(B)^{\uparrow})$\;
\end{algorithm}

\noindent
This $O(k' \log k')$ preprocessing step isolates exactly the quantities that
matter for feasibility: how much additional score each targeted outsider must
gain, and how much score each weak winner can still receive while remaining
below~$B$.  
Because middle-ranked candidates play no role in displacement feasibility, the
ordered pair $(b(B)^{\downarrow}, u(B)^{\uparrow})$
fully captures the coalition’s demand profile at cutoff~$B$.

Algorithm~\ref{alg:demand2e} serves as the preprocessing step for all subsequent
feasibility tests: it maps a cutoff $B$ to the ordered boost-demand and
suppression-tolerance vectors that are the sole inputs to the prefix-based
feasibility oracles.

\subsection{Prefix-Sum Feasibility Under Common-Step AP Ladders}
\label{sec:PrefixFeasibilityAlg}

Under arithmetic-progression (AP) ladders with a \emph{common step size}~$g$,
Lemmas~\ref{lem:ap-boost} and~\ref{lem:ap-suppress} show that feasibility
of either subproblem (boost or suppress) reduces to checking whether the
corresponding demand vector $q$—equal to the lattice-adjusted boost vector
$\hat b$ in the boosting case, or to the slack vector 
$d = F(k')\mathbf{1} - \hat u$ in the suppression case—can be rounded onto the
\emph{common-step AP lattice} and then satisfies the Block--HLP prefix
constraints together with the appropriate total-mass inequality
$\sum_{i=1}^{k'} y_i \le T$.
The routine below implements exactly these demand-feasibility tests: it checks
whether there exists a lattice vector $y \ge q$ that obeys the Block--HLP
prefix capacities and the required total-sum bound. Algorithm~\ref{alg:ap-feasible} implements the core prefix-and-congruence
feasibility oracle for a single demand vector under common-step AP ladders.
Algorithm~\ref{alg:envelope} builds on this oracle to compute, via two monotone
binary searches, the feasible cutoff interval $(B_{\min},B_{\max})$ at a fixed
displacement level~$k'$.

\begin{algorithm}[t]
\footnotesize
\DontPrintSemicolon
\caption{\textsc{APDemandFeasible}$(q, F(\cdot), T, g, \alpha)$}
\label{alg:ap-feasible}
\KwIn{demand vector $q \in \mathbb{Z}^{k'}$; prefix capacities $F(1),\ldots,F(k')$; total capacity $T$;
common step size $g$; residue $\alpha \in \{0,1,\ldots,g-1\}$ (representative of the feasible congruence class
$\alpha + g\mathbb{Z}$, i.e., feasible coordinates satisfy $y_i \equiv \alpha \pmod g$)}
\KwOut{\textbf{true} iff there exists realizable $y \ge q$}

\BlankLine
\For{$i \gets 1$ \KwTo $k'$}{
  $\hat q_i \gets \min\{ z \in \mathbb{Z} \mid z \ge q_i \ \wedge\ z \equiv \alpha \ (\mathrm{mod}\ g)\}$\tcp*[r]{round $q_i$ up to the next value in $\alpha + g\mathbb{Z}$}
}

$S \gets \sum_{i=1}^{k'} \hat q_i$\;
\If{$S > T$}{\Return \textbf{false}}

$\hat q \gets (\hat q_1,\dots,\hat q_{k'})$\;
$\hat q^{\downarrow} \gets \textsc{SortNonIncreasing}(\hat q)$\;
$pref \gets 0$\;
\For{$t \gets 1$ \KwTo $k'$}{
  $pref \gets pref + \hat q^{\downarrow}_t$\;
  \If{$pref > F(t)$}{\Return \textbf{false}}
}

\Return \textbf{true}\;
\end{algorithm}

\noindent
Here $F(t)$ and $T = F(k')$ are the coalition prefix capacities defined in
Section~\ref{sec:geometry}.
The parameter $g$ is the common step size of the AP ladders.
Write each ballot's relevant ladder values as $L_v + g\ell$ for integers $\ell$, and let
$L_{\mathrm{tot}} := \sum_{v=1}^m L_v$ be the sum of baselines across the $m$ colluders.
Then every realizable aggregate coordinate must lie in the same congruence class modulo $g$:
\[
y_i \equiv L_{\mathrm{tot}} \pmod g \qquad \text{for all } i=1,\dots,k'.
\]
We encode this by letting
\[
\alpha := L_{\mathrm{tot}} \bmod g \in \{0,1,\dots,g-1\}.
\]
Thus, the phrase ``residue class $\alpha$'' means the set $\alpha + g\mathbb{Z}$ of all integers
congruent to $\alpha$ modulo $g$.
Accordingly, in line~2 of Algorithm~\ref{alg:ap-feasible} we round each demand $q_i$ up to the smallest integer
$\hat q_i \ge q_i$ that lies in this residue class (i.e., $\hat q_i \equiv \alpha \pmod g$).

\paragraph{AP parameters for high and low ladders.}
The oracle \textsc{APDemandFeasible} is parametrized by the aggregate AP-ladder
structure of the relevant score pool.

For the \emph{boost} subproblem, we use only the coalition’s high-score
positions.  
Each ballot $v$ contributes a $k'$-length AP ladder
$S^{(v,\mathrm{high})}_{\mathrm{AP}}$ with a \emph{common step size} $g$ and
baseline $L^{\mathrm{high}}_v$.  
These determine the aggregate prefix capacities
\[
F^{\mathrm{high}}(t)
  = \sum_{v=1}^m \sum_{i=1}^t r^{(v,\mathrm{high})}_i,
\qquad
T^{\mathrm{high}} = F^{\mathrm{high}}(k'),
\]
and the total baseline
\[
L^{\mathrm{high}}_{\mathrm{tot}}
   = \sum_{v=1}^m L^{\mathrm{high}}_v.
\]
All realizable aggregates must lie in the residue class
$\alpha^{\mathrm{high}} \equiv L^{\mathrm{high}}_{\mathrm{tot}} \pmod g$,
which is the value used in the lattice-rounding step of the oracle.

For the \emph{suppression} subproblem, we analogously use the low-score AP
ladders $S^{(v,\mathrm{low})}_{\mathrm{AP}}$, which also share the same step
size~$g$, yielding
\[
F^{\mathrm{low}}(t), \qquad
T^{\mathrm{low}} = F^{\mathrm{low}}(k'),
\qquad
L^{\mathrm{low}}_{\mathrm{tot}}
  = \sum_{v=1}^m L^{\mathrm{low}}_v,
\qquad
\alpha^{\mathrm{low}} \equiv L^{\mathrm{low}}_{\mathrm{tot}} \pmod g.
\]

Thus, calls to \textsc{APDemandFeasible} for boosting and suppression simply
instantiate $(F(\cdot), T, g, \alpha)$ with the corresponding high- or
low-ladder parameters.

\paragraph{Boost and suppression as thin wrappers.}
Using Algorithm~\ref{alg:demand2e}, we first compute the ordered demand vectors
$(b(B)^{\downarrow},\, u(B)^{\uparrow})$ for any cutoff~$B$.
Under common-step AP ladders, both feasibility tests are handled by the same
prefix-and-congruence oracle \textsc{APDemandFeasible} after the appropriate
lattice adjustment.

\begin{algorithm}[t]
\footnotesize
\caption{\textsc{FeasibleEnvelopeAtLevel}$(k')$ (high-level sketch)}
\label{alg:envelope}
\KwIn{Honest scores $S$; sets $\mathcal{O}^\ast,\mathcal{T}_w$;
      AP parameters $(F^{\mathrm{high}},T^{\mathrm{high}},g,\alpha^{\mathrm{high}})$
      and $(F^{\mathrm{low}},T^{\mathrm{low}},g,\alpha^{\mathrm{low}})$;
      search bounds $B_{\mathrm{low}},B_{\mathrm{high}}$.}
\KwOut{(\textbf{true}, $B_{\min},B_{\max}$) or \textbf{false}}
\BlankLine

\tcp{Boost side: largest cutoff $B_{\max}$ with feasible boost demand}
Compute $b^{\downarrow}(B_{\mathrm{low}})$ and its lattice adjustment
$\hat b^{\downarrow}(B_{\mathrm{low}})$.\;
\If{\textsc{APDemandFeasible}$\bigl(
        \hat b^{\downarrow}(B_{\mathrm{low}}),
        F^{\mathrm{high}},T^{\mathrm{high}},
        g,\alpha^{\mathrm{high}}
     \bigr)$ is \textbf{false}}{
    \Return \textbf{false}\tcp*[r]{no feasible boost at any cutoff}
}
Binary search over $B \in [B_{\mathrm{low}},B_{\mathrm{high}}]$ to find
the largest $B_{\max}$ such that\;
\textsc{APDemandFeasible}$\bigl(
        \hat b^{\downarrow}(B),
        F^{\mathrm{high}},T^{\mathrm{high}},
        g,\alpha^{\mathrm{high}}
     \bigr)$ is \textbf{true}.\;

\medskip
\tcp{Suppress side: smallest cutoff $B_{\min}$ with feasible suppression demand}
Compute $u^{\uparrow}(B_{\mathrm{high}})$ and its lattice-adjusted
slack vector $q^{\mathrm{sup}}(B_{\mathrm{high}})$ (as in
Lemma~\ref{lem:ap-suppress}).\;
\If{\textsc{APDemandFeasible}$\bigl(
        q^{\mathrm{sup}}(B_{\mathrm{high}}),
        F^{\mathrm{low}},T^{\mathrm{low}},
        g,\alpha^{\mathrm{low}}
     \bigr)$ is \textbf{false}}{
    \Return \textbf{false}\tcp*[r]{no feasible suppression at any cutoff}
}
Binary search over $B \in [B_{\mathrm{low}},B_{\mathrm{high}}]$ to find
the smallest $B_{\min}$ such that\;
\textsc{APDemandFeasible}$\bigl(
        q^{\mathrm{sup}}(B),
        F^{\mathrm{low}},T^{\mathrm{low}},
        g,\alpha^{\mathrm{low}}
     \bigr)$ is \textbf{true}.\;

\medskip
\tcp{Final envelope check}
\If{$B_{\min} \le B_{\max}$}{
    \Return (\textbf{true}, $B_{\min},B_{\max}$)\;
}\Else{
    \Return \textbf{false}\;
}
\end{algorithm}

\begin{itemize}[leftmargin=1.4em]

\item \emph{Boost feasibility at $B$.}
Form the boost-demand vector $b(B)^{\downarrow}$.
Lattice-adjusting each coordinate upward produces
\[
\hat b_i(B)
  = b_i(B)^{\downarrow}
    + \bigl((L_{\mathrm{tot}} - b_i(B)^{\downarrow}) \bmod g\bigr).
\]
Calling
\[
   \textsc{APDemandFeasible}\!
     \bigl(\hat b(B),\, F(\cdot),\, T,\, g,\, \alpha\bigr)
\]
returns \textbf{true} exactly when a realizable aggregate $y$ with
$y_i \ge b_i(B)$ exists.

\item \emph{Suppression feasibility at $B$.}
Form the suppression upper-bound vector $u(B)^{\uparrow}$.
Lattice-adjusting each coordinate downward gives
\[
\hat u_i(B)
   = u_i(B)^{\uparrow}
     - \bigl((u_i(B)^{\uparrow} - L_{\mathrm{tot}}) \bmod g\bigr).
\]
Define the slack vector
\[
d(B) := F(k')\mathbf{1} - \hat u(B),
\]
and sort it in nonincreasing order: $d(B)^{\downarrow}$.
Lemma~\ref{lem:ap-suppress} shows that suppression at $B$ is feasible
iff
\[
   \textsc{APDemandFeasible}\!
     \bigl(d(B)^{\downarrow},\, F(\cdot),\, T,\, g,\, \alpha\bigr)
\]
returns \textbf{true}, in which case a realizable aggregate
$y$ exists with $y_i \le u_i(B)$.
\end{itemize}

Together with the monotone envelope structure from
Section~\ref{sec:interval}, these two thin-wrapper calls to
\textsc{APDemandFeasible} support the binary searches for
$B_{\max}$ and $B_{\min}$.

\paragraph{Feasible envelopes via binary search.}
By monotonicity (Lemma~\ref{lem:monotonicity}), boost feasibility is
nonincreasing in~$B$ while suppress feasibility is nondecreasing in~$B$.
Thus the endpoints $B_{\max}$ and $B_{\min}$ of the feasible envelopes
can be computed by binary search over the integer interval
$[B_{\mathrm{low}}, B_{\mathrm{high}}]$, using the common-step
AP-lattice oracle \textsc{APDemandFeasible} to evaluate feasibility at
each probe value~$B$.  
Displacement at level~$k'$ is achievable exactly when the envelopes
overlap, i.e., when $B_{\min} \le B_{\max}$.

\noindent
This envelope routine fully characterizes feasibility at level~$k'$.
Whenever the intervals overlap,
any cutoff $B \in [B_{\min}, B_{\max}]$ yields a valid feasible point,
from which explicit manipulative ballots can be constructed
(Section~\ref{sec:construct-ballots}).

\subsection{Maximizing Feasible Displacement}

For a fixed displacement level $k'$, feasibility is characterized by two
monotone envelopes: the boost-feasible interval $(-\infty,B_{\max}]$ and the
suppress-feasible interval $[B_{\min},\infty)$.  
Level $k'$ is manipulable exactly when these intervals overlap, i.e., when
$B_{\min} \le B_{\max}$.  
The routine \textsc{FeasibleEnvelopeAtLevel} computes $(B_{\min},B_{\max})$
for the $k'$ strongest outsiders and the $k'$ weakest winners using the
common-step AP-lattice feasibility oracles from
Lemmas~\ref{lem:ap-boost}--\ref{lem:ap-suppress}.

\begin{algorithm}[t]
\footnotesize
\DontPrintSemicolon
\caption{\textsc{MaximizeDisplacement}$(S,\mathcal{O},\mathcal{T},k,\Theta^{\mathrm{high}},\Theta^{\mathrm{low}})$}
\label{alg:MaximizeDisplacement}
\KwIn{
  honest scores $S$; outsiders $\mathcal{O}$ sorted by decreasing $S$;
  winners $\mathcal{T}$ sorted by increasing $S$; number of winners $k$;
  common-step AP parameters for high and low ladders
  $\Theta^{\mathrm{high}},\Theta^{\mathrm{low}}$;
  search range $k' \in \{0,\dots,k\}$.
}
\KwOut{maximum achievable displacement $k^\ast$ and its cutoff interval
$[B_{\min}^\star,B_{\max}^\star]$}

\BlankLine
\tcp{Monotonicity: if level $k'$ is feasible then all smaller levels are feasible}
Binary search over $k' \in \{0,\dots,k\}$ to find the largest feasible level $k^\ast$.\;

\medskip
\tcp{Feasibility test at a candidate level $k'$}
For a queried $k'$, set
$\mathcal{O}^\ast \gets$ first $k'$ elements of $\mathcal{O}$ and
$\mathcal{T}_w \gets$ first $k'$ elements of $\mathcal{T}$.\;
Invoke
$(ok,B_{\min},B_{\max}) \gets
\textsc{FeasibleEnvelopeAtLevel}(S,\mathcal{O}^\ast,\mathcal{T}_w,
\Theta^{\mathrm{high}},\Theta^{\mathrm{low}})$.\;

\medskip
\tcp{Update rule (implemented by the binary search)}
If $ok$ is \textbf{true}, record $(k',B_{\min},B_{\max})$ as the current best
and continue the search on larger $k'$; otherwise continue on smaller $k'$.\;

\medskip
\Return $(k^\ast,B_{\min}^\star,B_{\max}^\star)$\;
\end{algorithm}

To find the overall maximum achievable displacement $k^\ast$, we perform a binary search over
$k' \in \{0,\dots,k\}$, checking feasibility at each midpoint using
\textsc{FeasibleEnvelopeAtLevel}. The procedure performs $O(\log k)$
feasibility-envelope evaluations.
Each envelope evaluation conducts two binary searches over the cutoff~$B$,
and each such search calls the common-step AP-lattice feasibility oracle, which
runs in
\[
O\!\left(k' \log k' \right)
\]
time per probe.  
Since every relevant cutoff lies in an interval of length $O(mp_1)$ with
$p_1 = O(x)$, each binary search over $B$ uses $O(\log(mx))$ probes.  
Consequently, \textsc{MaximizeDisplacement} runs in
\[
O\!\left(k\log k \cdot \log(mx)\cdot \log k\right)
= O\!\left(k(\log k)^2\log(mx)\right)
\] 
time overall
(since every queried level satisfies $k' \le k$),
and returns the maximum displacement $k^\ast$ together with its feasible cutoff interval
$[B^\star_{\min},B^\star_{\max}]$, which fully
characterizes all realizable manipulative outcomes at the optimal level.

\subsection{Constructing Manipulative Ballots}
\label{sec:construct-ballots}

Once the maximum achievable displacement $k^\ast$ and a feasible cutoff $B$
have been identified, it is straightforward to construct explicit coalition
ballots that realize the certified boost and suppression contributions.
We construct ballots at this optimal level $k^\ast$.
Although the feasibility tests already guarantee existence, the procedure
below provides concrete rankings for the colluding voters.

\begin{algorithm}[t]
\footnotesize
\DontPrintSemicolon
\caption{\textsc{ConstructManipulativeBallots}$(S,\mathcal{O}^\ast,\mathcal{T}_w,m,x,k^\ast,B,\Theta^{\mathrm{high}},\Theta^{\mathrm{low}})$}
\label{alg:ConstructManipulativeBallots}
\KwIn{honest scores $S[\cdot]$; boundary sets $\mathcal{O}^\ast,\mathcal{T}_w$; number of colluders $m$; number of candidates $x$; optimal level $k^\ast$; feasible cutoff $B$; AP parameters $\Theta^{\mathrm{high}}=(F^{\mathrm{high}}(\cdot),T^{\mathrm{high}},g,\alpha^{\mathrm{high}})$ and $\Theta^{\mathrm{low}}=(F^{\mathrm{low}}(\cdot),T^{\mathrm{low}},g,\alpha^{\mathrm{low}})$}
\KwOut{manipulative rankings $\pi_1,\ldots,\pi_m$}
$k' \gets k^\ast$\;
$(b^\downarrow,u^\uparrow)\leftarrow \textsc{ComputeDemandVectors}(S,\mathcal{O}^\ast,\mathcal{T}_w,B)$\;

\BlankLine
\textbf{Top block (boost): witness aggregate and decomposition}\;
$\hat b \leftarrow \textsc{LatticeAdjustUp}(b^\downarrow,g,\alpha^{\mathrm{high}})$ \tcp*{as in Lemma~4}
$y^{\mathrm{high}} \leftarrow \textsc{ExtendToBase}(\hat b, F^{\mathrm{high}}(\cdot), T^{\mathrm{high}}, g)$\;
$(\sigma_1,\ldots,\sigma_m)\leftarrow \textsc{RealizeAP}(y^{\mathrm{high}},\Theta^{\mathrm{high}})$ \tcp*{Appendix~G}

\BlankLine
\textbf{Bottom block (suppress): witness aggregate and decomposition}\;
$\hat u \leftarrow \textsc{LatticeAdjustDown}(u^\uparrow,g,\alpha^{\mathrm{low}})$ \tcp*{as in Lemma~5}
$d \leftarrow T^{\mathrm{low}}\mathbf{1} - \hat u$ \tcp*{slack demand (Lemma~5)}
$z \leftarrow \textsc{ExtendToBase}(d, F^{\mathrm{low}}(\cdot), T^{\mathrm{low}}, g)$\;
$y^{\mathrm{low}} \leftarrow T^{\mathrm{low}}\mathbf{1} - z$\;
$(\rho_1,\ldots,\rho_m)\leftarrow \textsc{RealizeAP}(y^{\mathrm{low}},\Theta^{\mathrm{low}})$ \tcp*{Appendix~G}

\BlankLine
\textbf{Assemble ballots}\;
\For{$v\leftarrow 1$ \KwTo $m$}{
Assign the top $k'$ positions of $\pi_v$ to $\mathcal{O}^\ast$ according to permutation $\sigma_v$\;
Assign the bottom $k'$ positions of $\pi_v$ to $\mathcal{T}_w$ according to permutation $\rho_v$\;
Fill the remaining positions $\{k'+1,\ldots,x-k'\}$ with an arbitrary permutation of $\mathcal{C}\setminus (\mathcal{O}^\ast\cup \mathcal{T}_w)$\;
}
\Return{$(\pi_1,\ldots,\pi_m)$}\;
\end{algorithm}

\paragraph{Principle.}
By Lemma~\ref{lem:canonical}, any successful manipulation at level~$k'$ (in particular, at $k'=k^\ast$) may be
assumed to use only:
\begin{itemize}
    \item the $mk'$ \emph{highest} positional scores to boost the designated
          outsiders $\mathcal{O}^\ast$, and
    \item the $mk'$ \emph{lowest} positional scores to suppress the weak
          winners $\mathcal{T}_w$,
\end{itemize}
with all middle positions irrelevant to feasibility.
Lemma~\ref{lem:independent} guarantees that the high-score and low-score pools are independent resources, and Theorem~\ref{thm:guaranteed-topk}
ensures that any coalition profile satisfying
\[
F_o \ge B \ \ \forall o \in \mathcal{O}^\ast, \qquad
F_t \le B-1 \ \ \forall t \in \mathcal{T}_w,
\]
guarantees displacement level $k'$: in every completion of the middle ranks and under any tie-breaking, the
final Top-$k$ winner set contains at least $k'$ outsiders (and hence excludes at least $k'$ honest winners).
When the scoring vector contains plateaus, the identities of the displaced incumbents and entering outsiders
may depend on tie-breaking; the construction below is therefore aimed at realizing the certified separation.

Constructing ballots requires more than sorting: the feasibility oracle used in the envelope search
( \textsc{APDemandFeasible}) certifies \emph{existence} of feasible aggregates but does not specify how to
distribute score copies across the $m$ ballots.

After we have found a feasible cutoff $B$, we therefore construct \emph{witness} aggregate contribution
vectors for the high-score pool and low-score pool, and then decompose each witness aggregate into explicit
ballot-wise permutations of the corresponding AP ladder. This decomposition is guaranteed by the
constructive proof of the AP--Ladder Lattice Theorem (Theorem~\ref{thm:ap-lattice-common}; see Appendix~\ref{sec:ProofAP}).

Algorithm~\ref{alg:ConstructManipulativeBallots} summarizes this witness-and-decomposition construction.

\paragraph{Subroutines in Algorithm~\ref{alg:ConstructManipulativeBallots}.}
The lattice adjustments \textsc{LatticeAdjustUp} and \textsc{LatticeAdjustDown} are the coordinatewise
rounding operations from Lemmas~\ref{sec:ProofBoostFeasibility} and \ref{sec:ProofSupressFeasibility} that map a demand vector into the unique congruence class
$L_{\mathrm{tot}} + g\mathbb{Z}$. The routine \textsc{ExtendToBase} takes any vector satisfying the Block--HLP
prefix constraints with total sum at most $F(k')$ and returns a dominating integer base point of total sum
$F(k')$ (polymatroid extension; see~\cite{edmonds1970}). Finally, \textsc{RealizeAP} (Appendix~\ref{sec:ProofAP}) decomposes any
realizable aggregate into ballot-wise permutations as guaranteed by the constructive proof of Theorem~\ref{thm:ap-lattice-common}.

\paragraph{Correctness.}
Fix a displacement level $k'$ (in the construction we take $k'=k^\ast$) and a feasible cutoff $B$.
The envelope/oracle phase certifies that the lattice-adjusted boost demand $\hat b(B)$ and the slack demand
$d(B)=T^{\mathrm{low}}\mathbf{1}-\hat u(B)$ satisfy the Block--HLP prefix constraints (Lemmas~\ref{sec:ProofBoostFeasibility} and \ref{sec:ProofSupressFeasibility}), hence each
can be extended to a base point in the corresponding polymatroid (subroutine \textsc{ExtendToBase}).
This produces witness aggregates $y^{\mathrm{high}} \ge \hat b(B)$ and $y^{\mathrm{low}} \le \hat u(B)$ with the correct
congruence class modulo $g$.

By the constructive AP--Ladder Lattice Theorem (Theorem~\ref{thm:ap-lattice-common}; Appendix~\ref{sec:ProofAP}), each witness aggregate can be
decomposed into $m$ ballot-wise permutations of the relevant AP ladder. Therefore Algorithm~\ref{alg:ConstructManipulativeBallots}  produces
coalition ballots whose induced contributions satisfy
\[
S_o+\Delta_o \ge B \ \ \forall o\in \mathcal{O}^\ast, \qquad
S_t+\Delta_t \le B-1 \ \ \forall t\in \mathcal{T}_w.
\]
Applying the guaranteed-displacement theorem (Theorem~\ref{thm:guaranteed-topk}) yields robust displacement at level $k'$
(regardless of tie-breaking).

\paragraph{Complexity.}
The construction calls \textsc{ExtendToBase} (polynomial in $k'$) and the decomposition routine
\textsc{RealizeAP} from Appendix~\ref{sec:ProofAP} (polynomial in $(m,k')$), and then assembles ballots in $O(mk')$ time.

\section{Experimental Evaluation}
\label{sec:experiments}

We evaluate our framework on synthetic and real election data, focusing on
four core questions:

\begin{enumerate}
  \item \textbf{Exactness.}
        Are the Block--HLP feasibility tests \emph{tight}?  
        On instances where brute-force enumeration is feasible, does the
        oracle agree with the ground truth on every cutoff and displacement
        level?

  \item \textbf{Envelope accuracy.}
        Do the predicted feasibility boundaries $(B_{\min},B_{\max})$ match the
        oracle pointwise across all tested profiles, confirming that the
        boost--suppress envelopes give an exact geometric description of
        realizable manipulation?

\item \textbf{Scalability.}
     How does the oracle scale with the number of candidates? We benchmark the exact
maximum-achievable-displacement oracle (computing $k^\ast$ and implemented via the
Block--HLP/AP-lattice feasibility tests of Section~\ref{sec:algorithm}) on synthetic elections with candidate counts from
$x = 10^3$ up to $x = 10^9$. Runtime scales near-linearly in $x$ and is effectively independent of the
coalition size $m$, with larger instances limited only by memory.

  \item \textbf{Baseline comparison.}
        How does our exact method compare to natural heuristics?
        We show that both a greedy constructive strategy and an LP-rounding
        baseline systematically underestimate achievable displacement, often by
        large margins.
\end{enumerate}


\subsection{Setup}
\label{subsec:setup}

\paragraph{Synthetic elections.}
We generate rankings for $n = 1000$ honest voters using the Mallows model with dispersion
$\phi \in \{0.01, 0.05, 0.2, 0.6, 0.95, 0.97, 0.98, 0.99, 1.00\}$.
Here smaller $\phi$ means stronger concentration around the reference ranking, and $\phi=1$ is uniform over rankings.
For each $\phi$, we draw 50 independent elections.
Unless stated otherwise, we use pure Borda as the positional scoring rule.

\paragraph{Candidate and committee sizes (matching the figures).}
The synthetic experiments reported in the figures use the following candidate and committee sizes:
\begin{itemize}
  \item \textbf{Figure~\ref{fig:diminishing-synthetic} (maximal displacement vs.\ coalition size):} $x=500$ candidates and a Top-$k$ rule with $k=250$.
  \item \textbf{Figure~\ref{fig:envelope} (boost/suppress envelopes):} $x=100$ candidates, $k=50$, and $m=20$ colluding voters (as stated in the caption).
  \item \textbf{Figure~\ref{fig:runtime_scaling} (scalability in the number of candidates):}
 $x=10^{3}$ to $10^{9}$ and coalition sizes
$m \in \{10^{3},10^{4},10^{5},10^{6},10^{7}\}$.
  \item \textbf{Figure~\ref{fig:baseline-comparison} (baseline comparison):} $x=400$ candidates and a Top-$k$ rule with $k=200$.
\end{itemize}
Coalition sizes $m$ are varied over the ranges shown on the corresponding plot axes (up to $m/n=0.4$ on the synthetic plots).

\paragraph{Committee size and displacement cap.}
In every Top-$k$ election, the displacement power is bounded by
\[
k^\ast(m) \;\le\; \min\{k,\, x-k\}.
\]
In our synthetic plots we set $k=x/2$ (so $x-k=k$), which makes the upper bound equal to $k$.
This is why the $k^\ast$ curves can plateau at $k=250$ in Figure~\ref{fig:diminishing-both} (where $x=500$) and
why the envelope plot in Figure~\ref{fig:envelope} ranges over displacement levels up to $k=50$ (where $x=100$).

\paragraph{Real datasets.}
We evaluate on three classes of real-world preference data from PrefLib~\cite{MaWa13a}:
\begin{itemize}
  \item 33 Irish elections ($n \approx 3\times 10^3$--$1.5\times 10^4$, $x \approx 10$--$20$);
  \item the SUSHI dataset ($n = 5000$, $x = 100$);
  \item the Glasgow City Council dataset ($n \approx 3\times 10^4$, $x \approx 12$).
\end{itemize}
To make the real-data plots comparable across datasets and consistent with the axis scale in Figure~\ref{fig:diminishing-both},
we set the committee size to
\[
k \;=\; \min\{10,\ \lfloor x/2\rfloor\},
\]
so that the maximal possible displacement is at most 10 on every instance.
For each election we treat the reported rankings as the honest profile and apply our displacement oracle for
coalition sizes $m$ ranging up to a $20\%$ fraction of the electorate (as plotted in Figure~\ref{fig:diminishing-real}).

\paragraph{Scoring rules.}
All experiments use pure Borda unless otherwise noted. For positional rules with different curvature
(e.g., truncated Borda or geometric scoring), we recompute the capacity vectors $F(t)$ accordingly.

\paragraph{Implementation and environment.}
Our algorithms are implemented in Python using optimized NumPy and Cython-backed routines for prefix-sum
computations and feasibility checks. All running times are measured on a workstation equipped with an
Intel(R) Core(TM) Ultra 9 285K CPU, 32\,GB RAM, running Ubuntu~24.04.3~LTS.
Reported times are averaged over 20 repetitions.

\subsection{Exactness on Small Instances}
\label{subsec:exactness}

Our theory characterizes feasibility using a separating cutoff $B$: displacement
of $k'$ winners is achieved only when the coalition can create a strict score gap
\[
\min_{o\in \mathcal{O}^\ast} F_o
\;\ge\;
\max_{t\in \mathcal{T}_w} F_t + 1.
\]
To confirm that our Block--HLP and AP-lattice feasibility tests are tight and exact, we compare them to a
brute-force ground truth on very small elections where full enumeration is possible.

For small instances where exhaustive verification is computationally feasible, we validate our exact
feasibility oracle against brute-force ground truth: in the unit-step regime ($g=1$) this oracle reduces to the
Block--HLP prefix test (Corollary~\ref{cor:unit-step-structural}), while in the sparse-step regime ($g>1$) it applies the Block--HLP prefix
test together with the AP congruence condition from Theorem~\ref{thm:ap-lattice-common}.

Specifically, for elections with up to $x\le 7$ candidates and $m\le 2$ colluding voters, we enumerate all
coalition ballot profiles up to permutation of the $m$ colluders, i.e., treat the coalition as an unordered
multiset of ballots.\footnote{
Enumerating $(x!)^m$ \emph{ordered} coalition profiles is infeasible even for moderate $(x,m)$.
We therefore enumerate coalition profiles \emph{modulo permutation of the $m$ colluders}, which reduces the
search space substantially while remaining exact: permuting the colluders does not change the induced
aggregate score vector (and hence does not change whether a separating cutoff exists under our strict-gap
definition).}
For each enumerated coalition profile, we compute the resulting final score vector and test whether a
separating cutoff $B$ exists for each $k'$, following the strict-threshold feasibility condition used throughout
the theoretical development.

Across 1,200 randomly generated instances spanning Borda ($g=1$), a sparse-step common-step AP rule with
$g>1$ (e.g., $3$-scaled Borda with positional scores $p_r = 3(x-r)$), truncated Borda, plurality, and $3$--$2$--$1$
scoring rules, the oracle's predicted maximal displacement $k^\ast$ matched the brute-force maximum in
\textbf{100\%} of cases on all brute-forceable configurations. This confirms that the boost/suppress prefix systems
and the $B_{\min}$--$B_{\max}$ feasibility envelopes are tight on these instances: no feasible displacement is missed,
and no infeasible displacement is ever certified.

\paragraph{Why test a sparse-step rule ($g>1$).}
For unit-step rules such as Borda, Theorem~\ref{thm:ap-lattice-common} reduces to a pure prefix-majorization test because the
congruence restriction is vacuous (Corollary~\ref{cor:unit-step-structural}). When $g>1$, however, the congruence condition becomes
substantive: there exist integer points that satisfy all Block--HLP prefix inequalities but are not realizable by
ballot permutations because they fall outside the unique residue class modulo $g$. Including a $g>1$ scoring
rule in the brute-force sweep therefore directly validates the ``prefix plus one congruence class'' claim that
distinguishes the AP-lattice theory from the continuous Block--HLP envelope.

\subsection{Maximal Displacement as Coalition Size Grows}

\begin{figure*}[t]
\centering
\begin{subfigure}[t]{0.45\linewidth}
    \centering
    \includegraphics[width=\linewidth]{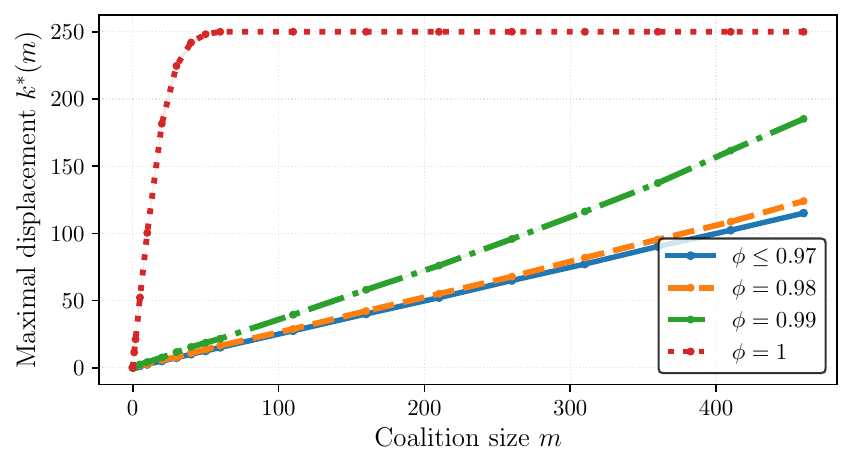}
    \caption{Synthetic (Mallows).}
    \label{fig:diminishing-synthetic}
\end{subfigure}
\hfill
\begin{subfigure}[t]{0.45\linewidth}
    \centering
    \includegraphics[width=\linewidth]{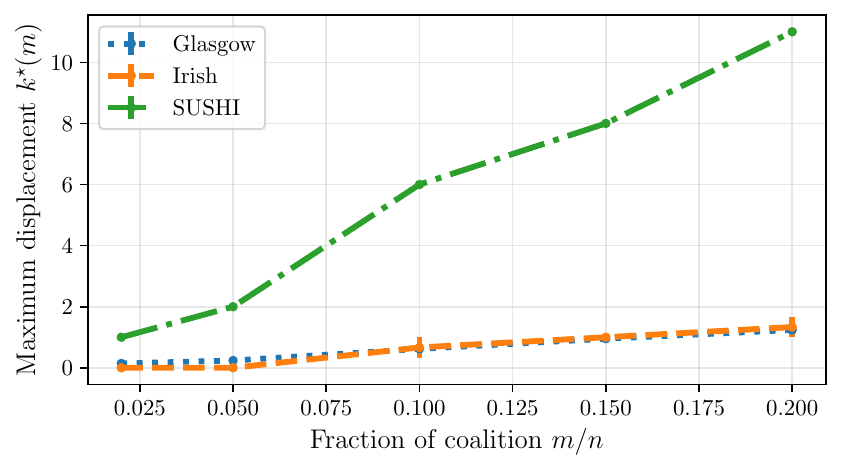}
    \caption{Real elections.}
    \label{fig:diminishing-real}
\end{subfigure}
\caption{
Maximal displacement $k^\ast$ as a function of coalition size.
\textbf{(a)} Synthetic Mallows profiles (Borda; $n=1000$; $x=500$; $k=250$), averaged over 50 trials for each dispersion parameter $\phi$ (see legend).
\textbf{(b)} Real PrefLib elections (Glasgow, Irish, SUSHI), plotting $k^\ast$ against coalition fraction $m/n$ up to $20\%$ (with $k=\min\{10,\lfloor x/2\rfloor\}$ for cross-dataset comparability).
}
\label{fig:diminishing-both}
\end{figure*}

Figures~\ref{fig:diminishing-synthetic} and~\ref{fig:diminishing-real} plot
how the coalition’s maximum achievable displacement $k^\ast(m)$ grows with the
number of colluding voters.

On synthetic Mallows elections (Figure~\ref{fig:diminishing-synthetic}),
the maximal displacement $k^\ast(m)$ is remarkably insensitive to $\phi$ over a broad range.
The curves for $\phi \in \{0.01, 0.05, 0.2, 0.6, 0.95, 0.97\}$ are visually indistinguishable (and are therefore shown as a single line labeled $\phi \le 0.97$),
and grow approximately linearly with the coalition size $m$.
Since smaller $\phi$ corresponds to stronger concentration (higher agreement) while $\phi \to 1$ approaches uniform randomness,
this collapse spans profiles ranging from highly structured to fairly noisy.

This behavior is consistent with the structure of the displacement oracle:
throughout this range, the binding constraint is primarily the coalition's aggregate score-supply geometry (the Block--HLP prefix capacities),
rather than small changes in the honest score gaps induced by $\phi$.
Only when $\phi$ becomes extremely close to $1$ do the honest score gaps near the Top-$k$ boundary compress enough to materially increase manipulability.

Beyond $\phi \approx 0.97$, the behavior changes visibly.
At $\phi = 0.98$ and $\phi = 0.99$, the curve shifts upward, reflecting the increased leverage available under near-uniform preferences without immediate saturation.
At $\phi = 1$ (uniform random preferences), the coalition quickly reaches the displacement cap (here $k=x/2$ in Figure~\ref{fig:diminishing-synthetic}), after which additional manipulators yield no further displacement.
Overall, the dependence on $\phi$ exhibits a sharp transition as $\phi \to 1$, rather than a smooth interpolation across the entire range.

Real elections (Figure~\ref{fig:diminishing-real}) show a similar
sublinear pattern but on a smaller scale.
SUSHI exhibits the largest displacement (reflecting its concentrated
preferences), while Irish and Glasgow profiles are noticeably harder to
manipulate.
Across all datasets, $k^\ast(m)$ grows more slowly than linearly and, once
the coalition becomes large enough, begins to flatten toward the
theoretical upper bound, consistent with the Block--HLP picture that
boosting capacity grows roughly linearly in $m$ while the “demand’’ to
promote additional outsiders grows faster.

\subsection{Feasibility Envelope}

\begin{figure}[t]
\centering
\includegraphics[width=0.5\linewidth]{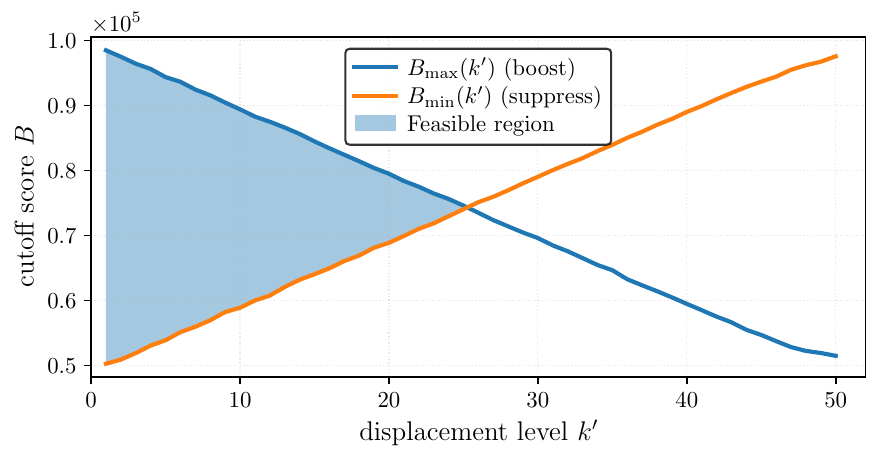}
\caption{Boost (blue) and suppress (orange) cutoff envelopes for a
synthetic election with $x=100$ candidates and $m=20$ colluding voters.
For each displacement level $k'$, the blue curve shows the largest
boost-feasible cutoff $B_{\max}(k')$, while the orange curve shows the
smallest suppress-feasible cutoff $B_{\min}(k')$.
Displacement at level $k'$ is feasible exactly when
$B_{\min}(k') \le B_{\max}(k')$, i.e., over the range of $k'$ where the
two envelopes overlap.
The total number of feasible displacement levels coincides with the
value $k^\ast$ returned by \textsc{MaximizeDisplacement}.}
\label{fig:envelope}
\end{figure}

Figure~\ref{fig:envelope} illustrates the feasibility-envelope
characterization for a single synthetic election instance.
For a fixed displacement level $k'$, the boost constraints induce the
interval $(-\infty, B_{\max}(k')]$ of score cutoffs for which the coalition
can raise all $k'$ targeted outsiders to at least $B$, while the suppress
constraints induce the interval $[B_{\min}(k'), \infty)$ of cutoffs for
which it can keep all $k'$ weak winners strictly below $B$.
Displacement at level $k'$ is feasible if and only if these two intervals
overlap.

The figure makes this geometric condition explicit: feasible displacement
levels correspond exactly to those values of $k'$ for which the boost and
suppress envelopes intersect.
For this instance, the overlap region is contiguous and its endpoint
coincides with the maximal displacement $k^\ast$ returned by
\textsc{MaximizeDisplacement}.
This visualization serves as a concrete illustration of how the abstract
boost--suppress conditions reduce feasibility testing to a simple
one-dimensional geometric criterion.

\begin{figure}[t]
    \centering
    \includegraphics[width=0.5\columnwidth]{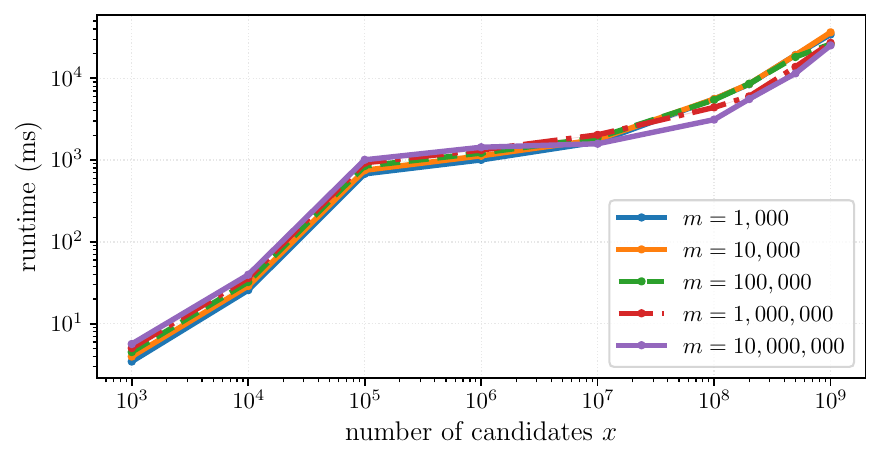}
  \caption{
\textbf{Scalability to large-scale elections.}
Runtime of the maximum-achievable-displacement oracle (computing $k^\ast$) as a function of the number of
candidates $x$ (log--log scale), ranging from $10^{3}$ up to $10^{9}$.
Results are shown for five representative coalition sizes
($m=10^{3}, 10^{4}, 10^{5}, 10^{6}, 10^{7}$), whose curves nearly overlap,
demonstrating that runtime is effectively independent of the coalition size.
The observed approximately linear scaling  in $x$ reflects the cost of
scanning the score array and extracting the boundary-relevant top and bottom
segments via linear-time selection, while the Block--HLP feasibility checks and
dual binary search contribute only negligible overhead.
}
    \label{fig:runtime_scaling}
\end{figure}

\subsection{Scalability in the Number of Candidates}
\label{sec:scaling}

To evaluate scalability in the number of candidates, we benchmark the exact displacement oracle on synthetic
elections with candidate counts ranging from $x = 10^3$ up to $x = 10^9$.
In each experiment, we explicitly generate an integer score array of length $x$ (32-bit) and apply linear-time selection to extract only the boundary-relevant top and bottom segments; the feasibility oracle then operates only on these $O(k')$-sized summaries.
The positional scoring vector is likewise represented explicitly to match the
general positional-rule interface.

Crucially, while score generation and boundary extraction scale linearly in $x$,
the geometric feasibility oracle itself operates only on the extracted boundary
sets and on prefix-capacity arrays of size $O(k')$.
As a result, the cost of the Block--HLP feasibility checks and the dual envelope
search is negligible relative to the one-pass scan and selection over the score
array.

Figure~\ref{fig:runtime_scaling} reports the end-to-end runtime as a function of
the number of candidates for five representative coalition sizes,
$m \in \{10^{3}, 10^{4}, 10^{5}, 10^{6}, 10^{7}\}$.
Across this range, runtime exhibits an approximately linear \emph{trend} on a
log--log scale, with moderate variability attributable to system and memory
effects, reaching on the order of $10^{4}$--$10^{5}$ milliseconds at
$x = 10^{9}$.
The curves for different coalition sizes nearly overlap throughout, indicating
that runtime is effectively independent of $m$.

These results empirically corroborate the theoretical analysis: the dominant
cost is the unavoidable linear pass over the score array, while the geometric
Block--HLP feasibility checks contribute only a small constant-factor overhead.
Even at candidate scales far exceeding any realistic real-world election, the
oracle remains stable and computationally practical.

We do not evaluate larger values of $x$ due solely to system memory constraints
associated with explicitly materializing the score array.
No algorithmic bottleneck was encountered within the tested range, and the
observed scaling behavior is consistent with continued linear-time dependence
on $x$ given sufficient memory.

\subsection{Comparison with Heuristic Baselines}
\label{sec:baseline-comparison}

We compare our exact displacement oracle (Section~\ref{sec:algorithm}) against
two natural baselines:

\begin{itemize}[leftmargin=1.3em]
    \item \textsc{GreedyPromote}: a myopic constructive heuristic that repeatedly
      identifies the currently strongest outsider still below the top-$k$
      threshold and allocates the coalition’s highest available scores to that
      candidate.  At each step it updates the running scoreboard and breaks
      ties in favor of outsiders with larger honest scores.  The remaining
      candidates on the ballot are ranked in decreasing order of their
      current scores, so that the highest-scoring incumbents receive the
      lowest possible positions.
    \item \textsc{LP-rounding}: formulates the displacement task as a linear
      program that relaxes the coalition’s per-ballot ranking constraints into
      a fractional score-allocation problem.  After solving the LP,
      randomized rounding is applied to convert the fractional allocations
      into integral Borda ballots.  The rounded ballots are then cast
      sequentially and the scoreboard updated after each step.
\end{itemize}

Both baselines have access to the same information (honest scores and coalition
size) but operate without the geometric Block--HLP / AP-lattice structure.

\begin{figure}[t]
    \centering
    \begin{subfigure}[t]{0.48\textwidth}
        \centering
        \includegraphics[width=\textwidth]{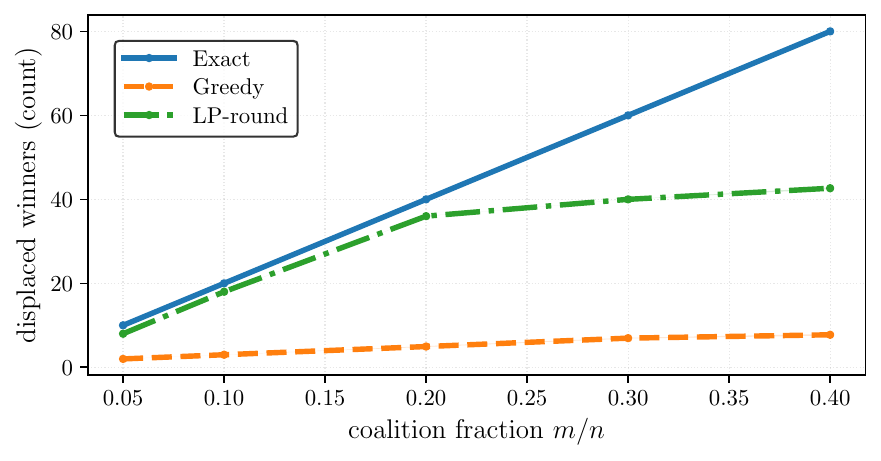}
        \caption{\textbf{Displacement optimality.}}
        \label{fig:baseline-optimality}
    \end{subfigure}
    \hfill
    \begin{subfigure}[t]{0.48\textwidth}
        \centering
        \includegraphics[width=\textwidth]{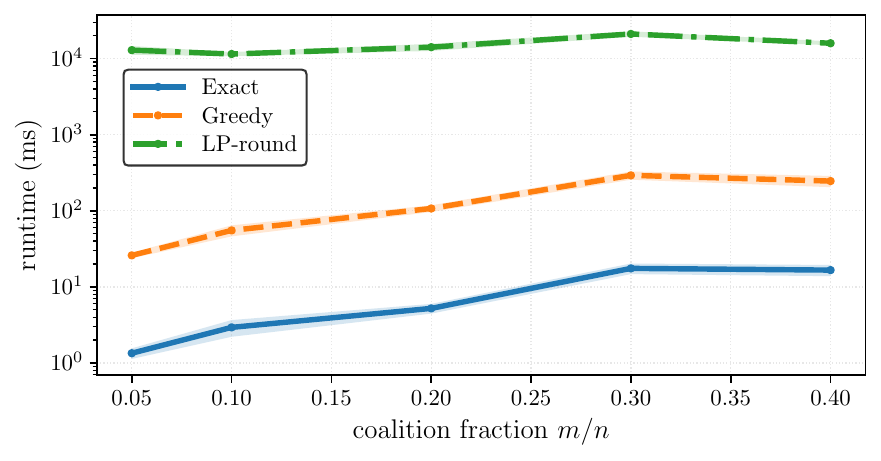}
        \caption{\textbf{Runtime comparison (log scale).}
       }
        \label{fig:baseline-runtime}
    \end{subfigure}

    \caption{
    \textbf{Comparison with heuristic baselines.}
    Synthetic Mallows elections with $x=400$ candidates and Borda scoring,
    averaged over 50 i.i.d.\ instances per coalition size.
    Shaded regions indicate 95\% confidence intervals.
    }
    \label{fig:baseline-comparison}
\end{figure}

\paragraph{Displacement quality.}
Figure~\ref{fig:baseline-optimality} reports the achieved displacement $k^\ast$ as a function of the coalition fraction $m/n$ for the exact
oracle and the two heuristic baselines.
The oracle exhibits near-linear growth across the entire range of coalition
sizes, consistently achieving the largest displacement.
For small coalitions, LP-rounding initially tracks the oracle closely and
achieves comparable displacement.
However, as the coalition grows, LP-rounding quickly saturates and flattens out,
reaching only about $40$ displaced winners at the largest coalition sizes—roughly
half of the oracle’s value.
The greedy heuristic performs substantially worse, saturating even earlier and
displacing fewer than $10$ winners even when $m/n$ approaches $0.4$.
This widening gap for larger coalitions shows that local or fractional score
reallocations fail to capture the global, prefix-aware coordination required
for optimal manipulation, which is precisely enforced by the Block--HLP
feasibility envelope.

\paragraph{Runtime comparison.}
Figure~\ref{fig:baseline-runtime} shows the running time of each method on a
logarithmic scale as a function of the coalition fraction $m/n$.
The exact oracle is extremely fast, with runtimes ranging from roughly
$1$--$20$\,ms across all coalition sizes.
The greedy heuristic is consistently slower, increasing from tens to several
hundred milliseconds as $m/n$ grows.
LP-rounding is dramatically slower than both alternatives, with runtimes on the
order of $10^4$ milliseconds throughout, reflecting the cost of repeatedly
solving large linear programs.
Notably, the oracle achieves strictly higher displacement while remaining one to
three orders of magnitude faster than greedy and roughly four orders of
magnitude faster than LP-rounding, demonstrating that the geometric structure is
algorithmically efficient as well as optimal.

\subsection{Summary}

Across both synthetic and real elections, our empirical evaluation
corroborates the theoretical guarantees of the Block--HLP framework and
demonstrates its practical effectiveness.

First, on all instances where exhaustive enumeration is feasible, the feasibility oracle achieves 100\% exactness,
confirming that the geometric characterization is tight: every feasible displacement is certified, and no
infeasible displacement is ever admitted.

Second, maximal displacement scales strongly with coalition size.
On synthetic Mallows elections, $k^\ast(m)$ grows approximately linearly in $m$
over a broad range of dispersions: the curves for all $\phi \le 0.97$ nearly
overlap and exhibit similar growth rates, despite $\phi=0.97$ already being
highly concentrated.
Only when $\phi$ becomes extremely close to $1$ does the behavior change
qualitatively, with rapid saturation at the committee-size cap.
Real-world elections show analogous trends at smaller scales, with
$k^\ast(m)$ increasing with $m$ at dataset-dependent rates.

Third, the algorithm scales smoothly to large candidate sets.
End-to-end runtime grows approximately linearly with the number of
candidates $x$ and is dominated by the unavoidable cost of scanning the
score array and extracting boundary-relevant segments.
Across multiple coalition sizes, runtime curves nearly overlap,
demonstrating that performance is effectively independent of the coalition
size and that the geometric feasibility checks contribute only a small
constant-factor overhead.

Finally, in comparison with natural heuristics, the exact oracle
decisively outperforms both greedy construction and LP-rounding.
While LP-rounding can initially track the oracle for small coalitions, it
quickly saturates and falls far short of the true feasible displacement as
the coalition grows, and incurs several orders of magnitude higher runtime.
The greedy heuristic performs substantially worse in both solution quality
and efficiency.

Collectively, these results demonstrate that the Block--HLP framework is
both theoretically sharp and practically scalable, providing a complete,
efficient, and provably optimal tool for analyzing coordinated
manipulation under positional scoring rules.

\section{Discussion}
\label{sec:discussion}

\subsection{Conceptual Contribution: A Canonical Decomposition of Manipulation}
A central conceptual contribution of this work is the identification of a
canonical geometric decomposition underlying coalition manipulation in
positional-scoring elections.  Despite the combinatorial richness of the
strategy space—each of the $m$ manipulative voters may rank $x$ candidates in
$x!$ possible ways—we show that the effect of the coalition collapses onto a
small, highly structured set of aggregate score allocations.

The first layer of this decomposition is \emph{continuous}: every achievable
aggregate score vector lies inside a polymatroid defined by the
Block--Hardy--Littlewood--Pólya (Block--HLP) prefix inequalities and a fixed
total mass constraint.  These convex conditions capture all feasible ways of
distributing the coalition's high and low scores without reference to the
individual ballots that produced them.

The second layer is \emph{discrete}: under the assumption that all colluding
ballots follow \emph{common-step} AP ladders, the
coalition’s realizable aggregates lie on a single congruence class modulo the
shared step size~$g$.  This produces a one-dimensional lattice structure:
every realizable aggregate satisfies $y_i \equiv L_{\mathrm{tot}} \pmod g$,
and no conflicting residue constraints arise across ballots.

Together, these two layers yield a complete characterization for
common-step AP ladders: feasible manipulations correspond exactly to the
intersection of a polymatroid with a rank-one lattice coset.  This
decomposition not only explains why the problem admits polynomial-time
algorithms in the common-step AP regime, but also delineates the structural
boundary between tractable scoring rules (with globally coherent step sizes)
and potentially intractable ones.  The remainder of this section explores the
consequences of this decomposition for complexity, general scoring systems, and
the broader landscape of coalition manipulation.

\subsection{The Algorithmic Frontier: Why Common-Step AP Ladders Yield Tractability}

The canonical decomposition has an immediate algorithmic consequence: 
\emph{when all colluding ballots share a common AP step size} $g$, displacement
feasibility reduces to evaluating the intersection of two monotone envelopes 
over a one-dimensional parameter~$B$.  This stands in sharp contrast to the 
apparent combinatorial explosion of the manipulation space, where each of the 
$m$ manipulative voters has $x!$ possible ballots and the coalition’s joint 
action space is of size $(x!)^m$.

The key to tractability is that common-step AP ladders impose two stabilizing 
geometric constraints:

\begin{enumerate}[leftmargin=1.3em]
\item \textbf{Prefix-submodularity of the score supply.}
      Each ballot’s contribution to the top $t$ positions forms a 
      nonincreasing prefix-sum sequence, and summing across ballots yields 
      the polymatroid defined by the Block--HLP inequalities.  
      This structure mirrors classic feasibility conditions in majorization 
      theory and eliminates the need to reason about individual permutations.

\item \textbf{Lattice coherence across ballots.}
      When all manipulative ballots use the \emph{same} step size~$g$, every 
      realizable aggregate score vector lies in a single residue class 
      modulo~$g$.  
      This global lattice coherence is essential: if even one ballot uses a 
      different step size, the shared residue structure collapses and the 
      polymatroid–lattice intersection used in our feasibility oracles no 
      longer holds.

\end{enumerate}

Together, these two properties collapse the search for feasible manipulations 
to a pair of scalar monotone searches: one determining the largest attainable 
boost cutoff $B_{\max}$ and one determining the smallest attainable suppress 
cutoff $B_{\min}$.  Their intersection determines displacement feasibility for 
any fixed~$k'$, and a second monotone search over $k'$ yields the maximum achievable
displacement $k^\ast$.

The resulting algorithm runs in polylogarithmic time in the score domain and 
near-linear time in the number of candidates.  In practice (as confirmed in 
Section~\ref{sec:experiments}), even extremely large elections can be processed 
within seconds.  
This level of scalability is surprising given the general intractability of 
manipulation under most scoring rules, and it derives entirely from the 
stabilizing effect of \emph{common-step} AP-ladder structure on both the 
continuous and discrete layers of the feasibility region.

Understanding where this tractability breaks is the focus of the next 
subsection, which proposes a three-region ``tractability frontier'' for positional scoring rules.

\subsection{The Coordination-Based Tractability Frontier}
\label{subsec:coordination-frontier}

The structural results in this paper indicate that the computational
complexity of coalition manipulation is governed not only by the positional
scoring rule, but also by the degree of \emph{alignment} among the coalition’s
ballot ladders.  The polynomial-time results established here rely critically
on the assumption that all $m$ colluding ballots share a \emph{common-step}
AP ladder.  We therefore measure deviations from this
structure using a parameter $C$, defined as the number of ballots whose score
ladders \emph{do not} share the common step size~$g$.

\medskip
\noindent
\textbf{The case $C=0$: common-step AP ladders (solved in this paper).}
When all ballots use the same AP step size, the aggregate score supply lies in
a \emph{single} residue class modulo~$g$, producing a one-dimensional lattice
coset embedded in the Block--HLP polymatroid.  This alignment collapses
displacement feasibility to the intersection of two monotone envelopes in a
scalar cutoff~$B$, yielding the polynomial-time characterization proved in this
paper.

\medskip
\noindent
\textbf{The case $C>0$: introducing irregular ladders.}
As soon as some ballots use ladders with \emph{different} step sizes, the
aggregate score supply no longer lies in a single residue class.  Instead,
these “irregular” ballots inject additional, independent congruence
constraints.  Geometrically, the Minkowski sum of the coalition’s ladders
thickens from a one-dimensional lattice line into a higher-dimensional lattice
region whose feasible points may occupy multiple incompatible residue classes.

Even a single irregular step ($C=1$) destroys global congruence alignment:
there is no longer a single modulo-$g$ lattice coset that contains all
realizable aggregates.  This eliminates the one-dimensional envelope structure
that underpins tractability.  With multiple irregular ballots ($C \ge 2$), the
aggregate lattice becomes genuinely multi-dimensional, and the interaction
between irregular congruence constraints and Block--HLP prefix inequalities
resembles the multi-dimensional packing constraints characteristic of classical
NP-hard problems.

\medskip
\noindent
Although a formal hardness boundary remains open, the structural evidence
suggests a sharp transition in computational behavior between the fully aligned
case $C=0$ (solved here in polynomial time) and settings with one or more
misaligned ladders ($C>0$), where feasibility is unlikely to reduce to a
one-dimensional monotone condition.

\medskip
\noindent\textbf{Conjecture (Coordination Threshold).}
\emph{There exists a universal constant $C^\star \ge 1$ such that displacement
feasibility under positional scoring rules is solvable in polynomial time for
all $C < C^\star$ (i.e., when fewer than $C^\star$ ballots deviate from the AP
structure), and becomes NP-hard once $C \ge C^\star$.  Structural evidence
based on Minkowski-sum geometry suggests that $C^\star$ is small---possibly as
low as~$2$.}

\medskip
\paragraph{Minkowski--Sum Thickening and Complexity.}
A geometric lens on the conjecture comes from examining the Minkowski sum of
the coalition’s score ladders.  
When all ballots share the \emph{same} AP step size
($C=0$), the Minkowski sum of their ladders remains a 
\emph{one-dimensional} arithmetic-progression coset lying on a single lattice
line inside the Block--HLP polymatroid; feasibility thus reduces to comparing
two monotone envelopes in a single cutoff parameter.

When some ballots deviate from the common step size, each such ``irregular''
ladder introduces an independent off-axis direction in the Minkowski sum.
Even if these ladders are AP individually, differing step sizes already break
the global lattice alignment and therefore behave as non-AP sources of
irregularity.  Each irregular block thickens the aggregate lattice from a line
into a low-dimensional tube.  Preliminary structural evidence suggests that
for $C=1$, this thickened region may still behave essentially
one-dimensionally, preserving enough structure for envelope-based feasibility
tests.

Once two or more irregular ladders are present ($C \ge 2$), the geometry
changes qualitatively.  
The Minkowski sum becomes a genuinely \emph{multi-dimensional} polytope whose
independent irregular directions interact with the Block--HLP polymatroid in a
way that mirrors multi-dimensional packing and partition polytopes.  
This geometric transition from one-dimensional to multi-dimensional feasible
regions provides strong structural motivation for the hypothesis that
manipulation complexity exhibits a sharp boundary at a small constant~$C^\star$.

\begin{itemize}[leftmargin=1.3em]
    \item \textbf{Inner circle ($C=0$): fully tractable.}
          All ballots share the same AP step size.  Feasible aggregates lie on
          a one-dimensional lattice line, and displacement is decided by
          monotone intervals.  This regime is completely solved in this paper.

    \item \textbf{Middle circle ($0 < C < C^\star$): conjectured tractable.}
          A small number of irregular (mixed-step or non-AP) ladders thicken
          the aggregate lattice but may still preserve an essentially
          one-dimensional envelope geometry.  
          We conjecture that polynomial-time algorithms extend throughout this
          regime.

    \item \textbf{Outer circle ($C \ge C^\star$): conjectured NP-hard.}
          Multiple irregular ladders generate a multi-dimensional Minkowski-sum
          polytope, causing feasibility constraints to behave like
          multi-dimensional knapsack problems.  
          We conjecture that manipulation feasibility becomes NP-hard beyond
          this boundary.
\end{itemize}

This conjecture deliberately leaves the exact value of $C^\star$ open.
Our results establish tractability at $C=0$ and provide structural evidence
that the complexity transition should occur at a small constant, with the cases
$C=1$ and $C=2$ forming important directions for future work.

\begin{figure}[t]
\centering
\begin{tikzpicture}[scale=1.1]

\colorlet{inner}{green!25}
\colorlet{middle}{yellow!30}
\colorlet{outer}{red!20}

\filldraw[fill=outer,  draw=black, thick] (0,0) circle (3);
\filldraw[fill=middle, draw=black, thick] (0,0) circle (2);
\filldraw[fill=inner,  draw=black, thick] (0,0) circle (1);

\node at (0,  2.35) {\large $\mathbf{C \ge C^\star}$};
\node at (0,  1.35) {\large $\mathbf{0 < C < C^\star}$};
\node at (0,  0) {\large $\mathbf{C = 0}$};

\begin{scope}[shift={(4.1,0)}]
    \node[anchor=west, align=left] at (0,2.0) {\textbf{Legend:}};
    
    \node[anchor=west, align=left] at (0,1.25) {
        \textbf{$C = 0$ (inner circle)}\\[-0.2em]
        pure displacement; fully tractable; proved in this paper
    };

    \node[anchor=west, align=left] at (0,-0.1) {
        \textbf{$0 < C < C^\star$ (middle circle)}\\[-0.2em]
        few coordinated blocks; conjectured polytime
    };

    \node[anchor=west, align=left] at (0,-1.45) {
        \textbf{$C \ge C^\star$ (outer circle)}\\[-0.2em]
        multi-block coordination; conjectured NP-hard
    };
\end{scope}

\end{tikzpicture}

\caption{Conjectured slack-based complexity frontier for targeted replacement.  
When $k = k'$ (fully targeted), the problem is NP-hard for $k = k' = 1$
(known) and conjectured NP-hard for all $k = k' > 1$.  For targeted partial
replacement, we conjecture a slack threshold $C_{\mathrm{slack}}$ such that
instances with $0 < k - k' < C_{\mathrm{slack}}$ remain hard, while those
with $k - k' \ge C_{\mathrm{slack}}$ may become tractable under common-step
AP ladders.}
\label{fig:three-circles}
\end{figure}

\subsection{A Second Complexity Frontier: Targeted Full vs.\ Targeted Partial Replacement}
\label{subsec:targeted-replacement}

The analysis in this paper concerns the \emph{maximum displacement} problem:
the coalition is allowed to replace \emph{any} subset of $k'$ current winners
with \emph{any} subset of $k'$ outsiders, and the objective is to maximize $k'$,
yielding the maximum achievable displacement $k^\ast$.  
Our structural results (Section~\ref{sec:decomposition}) show that there is
always an optimal manipulation that targets the $k'$ \emph{weakest} winners and
the $k'$ \emph{strongest} outsiders, so the search for maximum displacement can
be restricted without loss of generality to this canonical “weakest vs.\ strongest”
pairing.  
This non-targeted setting admits the canonical ordering that drives the
one-dimensional cutoff structure under common-step AP ladders.

In contrast, \emph{targeted replacement} asks the coalition to replace
\emph{prescribed} winners with \emph{prescribed} outsiders.
This objective is fundamentally incompatible with the maximum-displacement
problem considered in this paper: the goal is no longer to maximize how many
winners can be displaced, but to achieve one specific replacement outcome.
The two formulations address different questions and are not reducible to one
another.
Accordingly, the meaningful computational question in the targeted setting is
not maximum displacement but rather:
\emph{what is the minimal coalition size required to achieve the prescribed
replacement?}

Below we outline a second complexity frontier---one that emerges solely from
the structure of the replacement objective.

\paragraph{Targeted partial replacement ($k>k'\ge1$).}
Here the coalition must replace only $k'$ designated winners with $k'$
designated outsiders, while the remaining $k-k'$ winners are untargeted.
This setting differs fundamentally from both maximum displacement and fully
targeted replacement.

Unlike maximum displacement, the coalition is required to place \emph{specific}
outsiders into the top-$k$; unlike fully targeted replacement, the presence of
$k-k'$ untargeted winners provides \emph{slack}: some accidental promotions or
demotions of other candidates may occur without violating the objective.
We distinguish:
\begin{itemize}[leftmargin=1.3em,itemsep=2pt]
  \item \emph{Weak targeted replacement}: the required outsider--winner swaps
        occur, but additional candidates may cross the top-$k$ boundary.
  \item \emph{Strong targeted replacement}: exactly the designated outsiders
        enter and exactly the designated winners leave the top-$k$.
\end{itemize}

Slack raises the possibility that weak and strong targeted replacement coincide
in terms of minimal coalition size once $k-k'$ is large enough.  
This motivates a slack-threshold view of tractability.

\begin{quote}
\emph{Conjecture (Slack-based tractability frontier).}
There exists a slack threshold $C_{\mathrm{slack}}\ge 1$ such that, under
common-step AP ladders, computing the minimal coalition size for \emph{targeted
replacement} is NP-hard whenever $k-k' < C_{\mathrm{slack}}$ and polynomial-time
solvable whenever $k-k' \ge C_{\mathrm{slack}}$.
\end{quote}

The intuition is that sufficient slack allows accidental crossings of the
top-$k$ boundary to be absorbed or neutralized without increasing coalition
power, potentially restoring a one-dimensional cutoff structure.  In contrast,
fully targeted replacement ($k'=k$) has no slack at all and forces the
coalition to satisfy $k$ independent outsider--winner dominance constraints,
inducing a genuinely multi-dimensional feasibility region.

If this conjecture holds, then the \emph{minimal coalition size} needed for
targeted partial replacement should admit a one-dimensional characterization
under AP ladders, and thus be solvable in polynomial time—unlike the fully
targeted case.

\paragraph{Frontier Summary.}
These considerations suggest a second complexity frontier, governed by the
slack $k-k'$ in the replacement objective.  Informally:
\[
\textbf{Targeted full replacement } (k'=k)
   \ \Rightarrow\ 
   \begin{cases}
      \text{NP-hard for $k=k'=1$ under Borda and related rules (known);}\\[2pt]
      \text{conjectured NP-hard for all $k=k'>1$.}
   \end{cases}
\]

\[
\textbf{Targeted partial replacement with }
   k-k' \ge C_{\mathrm{slack}}
   \ \Rightarrow\
   \text{conjectured tractable under common-step AP ladders};
\]

\[
\textbf{Targeted partial replacement with }
   k-k' < C_{\mathrm{slack}}
   \ \Rightarrow\
   \text{conjectured hard}.
\]

The constant $C_{\mathrm{slack}}$ is not yet identified.
Preliminary evidence from our structural analysis of side-effect elimination
suggests that $C_{\mathrm{slack}}$ may be as small as~2, though a complete
characterization remains open.
Informally, when slack is large, many pairwise outsider--winner dominance
constraints appear to collapse into a single global cutoff condition,
restoring the one-dimensional geometry that enables tractability in maximum
displacement; when slack is small, the feasible region remains fundamentally
multi-dimensional, preventing such a reduction.

\begin{figure}[t]
\centering
\begin{tikzpicture}[x=1.4cm,y=1.2cm]

\draw[->,thick] (0,0) -- (7.6,0) node[anchor=west] {$k - k'$ (slack)};
\draw[thick] (0,-0.1) -- (0,0.1);
\node[below] at (0,-0.1) {\scriptsize $0$};

\begin{scope}
  \fill[red!10] (0,0.5) rectangle (2.0,1.7);
  \draw[red!50!black,thick] (0,0.5) rectangle (2.0,1.7);
  \node at (1.0,1.25) {\small \textbf{NP-hard}};
  \node at (1.0,0.9) {\tiny $k = k' = 1$ (known)};
\end{scope}

\begin{scope}
  \fill[red!10] (2.0,0.5) rectangle (4.0,1.7);
  \draw[red!50!black,thick] (2.0,0.5) rectangle (4.0,1.7);
  \node at (3.0,1.25) {\small \textbf{NP-hard?}};
  \node at (3.0,0.9) {\tiny $k = k' > 1$};
\end{scope}

\begin{scope}
  \fill[red!10] (4.0,0.5) rectangle (5.8,1.7);
  \draw[red!50!black,thick] (4.0,0.5) rectangle (5.8,1.7);
  \node at (4.9,1.25) {\small \textbf{Hard?}};
  \node at (4.9,0.9) {\tiny $0 < k - k' < C_{\mathrm{slack}}$};
\end{scope}

\begin{scope}
  \fill[green!12] (5.8,0.5) rectangle (7.6,1.7);
  \draw[green!50!black,thick] (5.8,0.5) rectangle (7.6,1.7);
  \node at (6.7,1.25) {\small \textbf{Tractable?}};
  \node at (6.7,0.9) {\tiny $k - k' \ge C_{\mathrm{slack}}$};
\end{scope}

\node at (2.0,2.15) {\scriptsize targeted full replacement};

\node at (6.0,2.15) {\scriptsize targeted partial replacement};

\node at (4.9,1.95) {\scriptsize low slack};
\node at (6.7,1.95) {\scriptsize high slack};

\node at (5.8,0.2) {\scriptsize $C_{\mathrm{slack}}$};

\end{tikzpicture}

\caption{Conjectured slack-based complexity frontier for targeted replacement.  
When $k = k'$ (fully targeted), the problem is NP-hard for $k = k' = 1$
(known) and conjectured NP-hard for all $k = k' > 1$.  For targeted partial
replacement, we conjecture a slack threshold $C_{\mathrm{slack}}$ such that
instances with $0 < k - k' < C_{\mathrm{slack}}$ remain NP-hard, while those
with $k - k' \ge C_{\mathrm{slack}}$ become tractable under common-step
AP ladders.}
\label{fig:slack-frontier}
\end{figure}
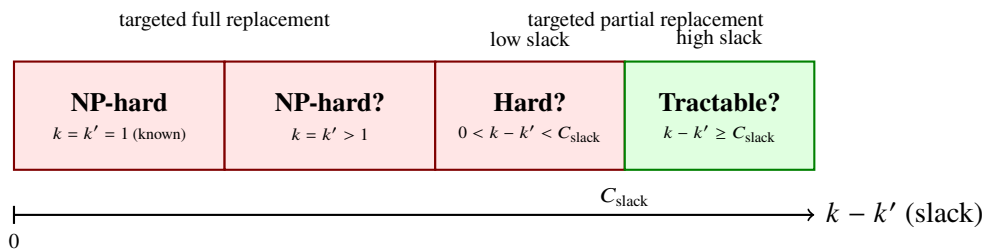

Figure~\ref{fig:slack-frontier} summarizes the conjectured frontier and the
role of slack in separating the hard and potentially tractable regimes.

\paragraph{Consistency with classical hardness ($k=k'=1$).}
Classical NP-hardness results assume fully targeted replacement with no slack.
By contrast, maximum displacement and the $k>k'$ targeted-partial regime rely
on canonical ordering and slack-based normalization, which largely remove the
need for enforcing individual outsider--winner dominance constraints.  Thus there is no inconsistency: targeted
full replacement remains hard, while non-targeted and slack-rich settings admit
structure.

\subsection{Broader Perspective and Conceptual Implications}

The results of this paper contribute to a growing body of work showing that the
computational structure of coalition manipulation is governed less by the
specific form of the positional scoring rule and more by the \emph{geometry of
feasible score redistributions}.  At a high level, our analysis reveals two
general principles that extend beyond the displacement setting studied here.

\paragraph{Manipulation as structured resource allocation.}
Although the coalition's action space is combinatorially enormous---each of the
$m$ manipulative voters may submit any of $x!$ possible rankings---their
aggregate effect collapses onto a highly structured feasibility region governed
by prefix-sum (submodular) constraints and a fixed total-mass constraint.
Viewed through this lens, coalition manipulation becomes a resource allocation
problem over a polymatroid: the coalition distributes a limited supply of
high-score positions subject to dominance constraints reminiscent of classical
majorization theory.  This continuous geometric structure abstracts away the
underlying ballots and exposes the essential variables that determine
manipulability.

\paragraph{Tractability via one-dimensional monotone structure.}
A second guiding principle is that manipulation becomes algorithmically
tractable whenever feasible outcomes admit a \emph{monotone
one-parameter characterization}.  In the displacement problem, this parameter
is a scalar cutoff $B$ separating the targeted outsiders from the targeted
winners.  Once feasibility reduces to checking whether two monotone envelopes
overlap on the $B$-axis, the entire problem collapses to two polynomial-time
feasibility oracles and a pair of binary searches.  This perspective explains
how large portions of the manipulation landscape can be efficiently decided
despite the exponential size of the coalition’s action space.

\paragraph{Implications for strategic multiagent systems.}
These two principles---polymatroidal structure and one-dimensional monotonicity---
suggest that similar tractability phenomena may arise in other multiagent
strategic settings.  Examples include:
\begin{itemize}[leftmargin=1.3em]
    \item \emph{Mechanism design and matching}, where threshold-based
          feasibility regions often collapse to scalar deviation constraints;
    \item \emph{Federated learning and fairness}, where client influence and
          robustness constraints can become tractable once reduced to a
          single-order statistic or cutoff parameter;
    \item \emph{Committee selection and multiwinner voting}, where feasible
          committees may be characterized through prefix-dominance or
          majorization conditions.
\end{itemize}
The canonical decomposition established here---a continuous polymatroid region
intersected with a discrete feasibility layer---provides a template for
analyzing strategic behavior in these and related domains.  It highlights that
complexity transitions often arise not from the size of the strategy space but
from the \emph{loss of monotonicity} or \emph{loss of geometric coherence} in
the aggregate feasible region.

Understanding these transitions, identifying their generality, and mapping the
boundary between tractable and intractable manipulation problems constitute
promising directions for future research.

\paragraph{Implications beyond voting.}
Many socio-technical systems aggregate ranked or scored signals from distributed
participants: peer-review scoring pipelines, recommender-system ranking,
credit-attribution platforms, moderation signals, and other integrity workflows.
In these domains, coordinated actors effectively control multiple copies of a
score ladder, and the geometric tools developed in this paper provide a
tractable way to quantify feasible influence, detect anomalous redistributions,
and certify whether observed aggregates are consistent with any legitimate
coalition.  Thus the Block--HLP envelope and AP-lattice structure offer a
general mathematical toolkit for modeling coordinated behavior across a wide
spectrum of ranking, evaluation, and multi-signal aggregation systems.

\section{Conclusion and Future Work}
\label{sec:conclusion}

This paper develops a unified geometric framework for analyzing coalition
manipulation in multiwinner elections under positional scoring rules.
By decomposing manipulation into independent \emph{boost} and \emph{suppress}
phases, we show that displacement feasibility at any fixed queried level~$k'$
is governed exactly by two prefix-majorization systems comparing coalition demand
vectors with its rank-prefix supply.
  This yields the Block--HLP
characterization, which identifies precisely when the coalition can raise all
targeted outsiders above a common cutoff~$B$ while simultaneously forcing all
targeted winners strictly below it.

For \emph{common-step} AP scoring vectors, the feasible
region refines into a discrete prefix-and-congruence lattice. In this
common-step AP regime, the Block--HLP conditions are also \emph{sufficient},
enabling an exact polynomial-time realizability oracle and a dual binary-search
algorithm for computing the maximum achievable displacement $k^\ast$.

Empirical evaluations confirm these results: on all instances where exhaustive
enumeration is feasible, the exact Block--HLP oracle agrees with brute-force
ground truth, and the Block--HLP envelopes match the oracle exactly on all tested
instances. The algorithm further scales to extremely large elections (e.g.,
$10^9$ candidates in under $28$ seconds), with memory rather than computation
being the only limiting factor.

Taken together, these results provide a complete geometric description of the
maximum displacement problem under common-step positional scoring rules, and an
exact algorithm for computing the maximum achievable displacement $k^\ast$.
Several research directions follow naturally:

\begin{itemize}[leftmargin=1.5em]

\item \emph{Tie-breaking and stochastic displacement.}
  Our strict-separation model guarantees success under \emph{all}
  tie-breaking rules, including adversarial ones.  Relaxing this to
  lexicographic, neutral, or randomized tie-breaking raises questions about
  expected displacement, probabilistic guarantees, and optimal stochastic
  manipulation strategies.  New geometric techniques may be needed to
  characterize “expected envelopes’’ under such uncertainty.

\item \emph{Beyond common-step AP structure and the coordination frontier.}
  When the coalition's ballots deviate from a shared step size, the aggregate
  feasible region loses its one-dimensional lattice structure.
  Section~\ref{subsec:coordination-frontier} suggests a sharp computational
  transition governed by the number of such deviations.  Determining whether a
  small universal threshold exists---beyond which feasibility becomes
  NP-hard---is a central open problem.  Mapping this \emph{coordination-based
  tractability frontier} and identifying exact structural assumptions needed
  for polynomial-time feasibility remain major challenges.

\item \emph{Minimal coalition size and the targeted-replacement frontier.}
  Computing minimum-size manipulating coalitions is NP-hard in the
single-winner case, and we conjecture that hardness persists for all
\emph{targeted full replacement} tasks ($k = k'$), even under common-step AP ladders.
By contrast, when $k > k' \ge 1$ and the coalition aims for
\emph{targeted partial displacement}, the feasible region may retain a one-dimensional cutoff structure under common-step AP ladders, which could admit efficient and exact algorithms.
  Understanding exactly which $(k,k')$ pairs under common-step AP ladders admit such monotone envelope characterizations—and therefore may support efficient and exact minimal-coalition algorithms—remains an important open question.

\item \emph{Variable-size committees and endogenous selection.}
  Our analysis assumes a fixed Top-$k$ rule.  Extending the geometric
  framework to systems with endogenous committee size—such as approval voting
  with quotas, threshold-based selection, or proportional representation—will
  require new definitions of displacement and new geometric separators.
  Understanding how prefix-majorization interacts with dynamic committee sizes
  is a promising direction.

\end{itemize}

Beyond these algorithmic questions, the interval-majorization viewpoint offers
tools for \emph{manipulation detection} and auditing.  Because the Block--HLP
conditions define tight feasible envelopes, observed election outcomes that
deviate from these envelopes may provide geometric signatures of coordinated
behavior.  This connection suggests new methods for analyzing robustness in
ranked data, aggregation systems, and federated learning environments where
bounded strategic influence is a central concern.

In summary, the Block--HLP theory provides a complete geometric foundation for
understanding coalition feasibility in positional scoring rules.  Combined with
the discrete structure exploited by our realizability oracle, it yields both
analytical clarity and practically scalable algorithms.  The broader landscape
of coalition manipulation—especially beyond common-step structure and beyond
fixed committee sizes—offers a rich set of open problems at the interface of
social choice, convex geometry, and computational complexity.

\bibliographystyle{ACM-Reference-Format}
\bibliography{Reference}
\appendix

\section{Proof of Lemma~\ref{lem:canonical}}\label{sec:ProofLemma1}
\begin{proof}[Full Proof of Lemma~\ref{lem:canonical}]

Recall Lemma~\ref{lem:canonical} states that if displacement at level $k'$ is feasible under a positional scoring rule, then there exists a successful
manipulation in which (i) every top-$k'$ position on every coalition ballot is assigned to a candidate in
$\mathcal{O}^\ast$, and (ii) every bottom-$k'$ position on every coalition ballot is assigned to a candidate in
$\mathcal{T}_w$.

\paragraph{Setup and notation.}
Let $p=(p_1,\dots,p_x)$ be the positional scoring vector with $p_1\ge \cdots \ge p_x$.
Let $\pi=(\pi_1,\dots,\pi_m)$ be a manipulation profile (one complete ranking $\pi_v$ per colluding voter).
For a candidate $c$ and ballot $\pi_v$, let $\mathrm{pos}_{\pi_v}(c)\in\{1,\dots,x\}$ denote the position of $c$
in $\pi_v$.
Define the coalition's total score contribution to candidate $c$ under manipulation profile $\pi$ by
\[
\Delta_c(\pi)\;:=\;\sum_{v=1}^m p_{\mathrm{pos}_{\pi_v}(c)}.
\]
The final score is $F_c(\pi):=S_c+\Delta_c(\pi)$, where $S_c$ is the honest score.

Fix a displacement level $k'$ and a separating cutoff $B$ such that $\pi$ is successful, i.e.,
\begin{equation}
\label{eq:sep_appA}
F_o(\pi)=S_o+\Delta_o(\pi)\ge B \quad \forall\, o\in \mathcal{O}^\ast,
\qquad
F_t(\pi)=S_t+\Delta_t(\pi)\le B-1 \quad \forall\, t\in \mathcal{T}_w.
\end{equation}

\paragraph{A basic swap fact.}
Consider a single ballot $\pi_v$ and two candidates $a,b$ that occupy adjacent positions $j-1$ and $j$.
Swapping $a$ upward to $j-1$ and $b$ downward to $j$ changes only their coalition contributions:
\[
\Delta_a \leftarrow \Delta_a + (p_{j-1}-p_j),\qquad
\Delta_b \leftarrow \Delta_b - (p_{j-1}-p_j),
\]
and leaves all other $\Delta_c$ unchanged. Since $p_{j-1}\ge p_j$, we have $p_{j-1}-p_j\ge 0$.
Thus, moving a candidate upward weakly increases its $\Delta$, and moving a candidate downward weakly
decreases its $\Delta$.

\paragraph{Violation counters.}
For each coalition ballot $\pi_v$, define:
\begin{itemize}
\item $H_v$: the number of candidates in the top-$k'$ positions $\{1,\dots,k'\}$ of $\pi_v$
that are \emph{not} in $\mathcal{O}^\ast$;
\item $L_v$: the number of candidates in the bottom-$k'$ positions $\{x-k'+1,\dots,x\}$ of $\pi_v$
that are \emph{not} in $\mathcal{T}_w$.
\end{itemize}
Let $H:=\sum_{v=1}^m H_v$ and $L:=\sum_{v=1}^m L_v$.
We will transform $\pi$ via adjacent swaps into a canonical manipulation profile $\pi^{\mathrm{can}}$ such that $H=L=0$ while maintaining
\eqref{eq:sep_appA} throughout. Since each local correction strictly decreases either $H$ or $L$ and both are nonnegative integers,
the process terminates.

\subparagraph{Part I: Fixing misallocated top-$k'$ positions (making $H=0$).}
Assume $H>0$ and pick a ballot $\pi_v$ with $H_v>0$. Let $i\le k'$ be the \emph{largest} index such that
the occupant $c:=\pi_v[i]$ satisfies $c\notin \mathcal{O}^\ast$. By maximality of $i$, every position
$j\in\{i+1,\dots,k'\}$ is occupied by a candidate in $\mathcal{O}^\ast$.
Since $|\mathcal{O}^\ast|=k'$ but the top-$k'$ block of $\pi_v$ contains $c\notin\mathcal{O}^\ast$, there exists
at least one candidate $o\in\mathcal{O}^\ast$ appearing at some position $j>k'$ on the same ballot.

We modify $\pi_v$ using adjacent swaps in two stages, ensuring that the candidate ejected from the top-$k'$
block is exactly the misallocated $c$.

\emph{Stage 1 (move $c$ to the boundary).}
Starting from its current position $i\le k'$, repeatedly swap $c$ downward by adjacent swaps until $c$
reaches position $k'$. Each such swap moves $c$ to a weakly lower-scoring position, so $\Delta_c$ weakly
decreases. Moreover, the only candidates swapped upward past $c$ during this bubbling process are the
candidates originally in positions $i+1,\dots,k'$, all of which lie in $\mathcal{O}^\ast$ by the choice of $i$.
In particular, no candidate in $\mathcal{T}_w$ is ever moved upward in Stage~1.

\emph{Stage 2 (insert $o$ at the boundary).}
From the resulting ballot, repeatedly swap $o$ upward by adjacent swaps until $o$ reaches position $k'$.
This places $o$ into the top-$k'$ block and pushes $c$ from position $k'$ to position $k'+1$, thereby
removing $c$ from the top-$k'$ block.

\emph{Effect on the constraints.}
For the boost constraints, Stage~1 can only decrease $\Delta_c$ for $c\notin\mathcal{O}^\ast$ and weakly
increase the contributions of candidates in $\mathcal{O}^\ast$ that move upward; Stage~2 weakly increases
$\Delta_o$ for $o\in\mathcal{O}^\ast$ and the only candidate moved downward out of the top-$k'$ block is
$c\notin\mathcal{O}^\ast$. Hence all inequalities $S_{o'}+\Delta_{o'}\ge B$ for $o'\in\mathcal{O}^\ast$ remain valid.
For the suppress constraints, Stage~1 does not move any $t\in\mathcal{T}_w$ upward (by the choice of $i$),
and Stage~2 moves every candidate passed by $o$ downward by one position (so their $\Delta$ weakly
decreases). Therefore every inequality $S_t+\Delta_t\le B-1$ for $t\in\mathcal{T}_w$ remains valid
throughout.

\emph{Progress.}
After Stage~2, position $k'$ is occupied by $o\in\mathcal{O}^\ast$ and $c\notin\mathcal{O}^\ast$ has been moved
to position $k'+1$. Thus the number $H_v$ of non-$\mathcal{O}^\ast$ candidates in the top-$k'$ block decreases
by at least one, so $H$ strictly decreases. Repeating this operation while $H>0$ yields a successful
manipulation profile (still satisfying~(1)) with $H=0$, i.e., every top-$k'$ position on every ballot is
occupied by a candidate in $\mathcal{O}^\ast$.

\subparagraph{Part II: Fixing misallocated bottom-$k'$ positions (making $L=0$).}
We now start from a successful manipulation profile in which $H=0$ (i.e., every top-$k'$ position on every
ballot is occupied by a candidate in $\mathcal{O}^\ast$), and assume $L>0$. Pick a ballot $\pi_v$ with $L_v>0$.
Let $\ell \ge x-k'+1$ be the \emph{smallest} index within the bottom-$k'$ block such that the occupant
$c:=\pi_v[\ell]$ satisfies $c\notin \mathcal{T}_w$. By minimality of $\ell$, every position
$j\in\{x-k'+1,\dots,\ell-1\}$ is occupied by a candidate in $\mathcal{T}_w$.
Since $|\mathcal{T}_w|=k'$ but the bottom-$k'$ block of $\pi_v$ contains $c\notin\mathcal{T}_w$, there exists
at least one candidate $t\in\mathcal{T}_w$ appearing at some position $j<x-k'+1$ on the same ballot.

We modify $\pi_v$ using adjacent swaps in two stages, ensuring that the candidate ejected from the bottom-$k'$
block is exactly the misallocated $c$.

\emph{Stage 1 (move $c$ to the boundary).}
Starting from its current position $\ell\ge x-k'+1$, repeatedly swap $c$ upward by adjacent swaps until $c$
reaches position $x-k'+1$. Each such swap moves $c$ to a weakly higher-scoring position, so $\Delta_c$ weakly
increases. Moreover, the only candidates swapped downward past $c$ during this bubbling process are the
candidates originally in positions $x-k'+1,\dots,\ell-1$, all of which lie in $\mathcal{T}_w$ by the choice of
$\ell$. Hence, during Stage~1, every $t\in\mathcal{T}_w$ that is affected is moved \emph{downward} (so its
$\Delta_t$ weakly decreases), and no candidate in $\mathcal{T}_w$ is moved upward.

\emph{Stage 2 (insert $t$ at the boundary).}
From the resulting ballot, repeatedly swap $t$ downward by adjacent swaps until $t$ reaches position
$x-k'+1$. This places $t$ into the bottom-$k'$ block and pushes $c$ from position $x-k'+1$ to position
$x-k'$, thereby removing $c$ from the bottom-$k'$ block.

\emph{Effect on the constraints.}
For the suppress constraints, in Stage~1 the only candidate whose $\Delta$ increases is $c\notin\mathcal{T}_w$,
while any $t'\in\mathcal{T}_w$ swapped downward past $c$ has $\Delta_{t'}$ weakly decreased; in Stage~2, the
candidate $t\in\mathcal{T}_w$ moves downward, so $\Delta_t$ weakly decreases, and the only candidate moved
upward out of the bottom-$k'$ block is $c\notin\mathcal{T}_w$. Hence every inequality
$S_{t'}+\Delta_{t'}\le B-1$ for $t'\in\mathcal{T}_w$ remains valid throughout.

For the boost constraints, note first that, since $H=0$ holds before Part~II, every candidate of $\mathcal{O}^\ast$
already occupies a top-$k'$ position on $\pi_v$. Both Stage~1 and Stage~2 act entirely within the suffix of the
ballot from positions $j\ge x-k'+1$, and thus never move any $\mathcal{O}^\ast$ candidate \emph{out of} the top-$k'$
block. Moreover, any candidates swapped upward during Stage~1 are moved to positions no higher than
$x-k'+1>k'$, so the coalition contributions of $\mathcal{O}^\ast$ candidates (which come from top-$k'$ positions) are
unchanged. Therefore all inequalities $S_o+\Delta_o\ge B$ for $o\in\mathcal{O}^\ast$ remain valid.

\emph{Progress.}
After Stage~2, position $x-k'+1$ is occupied by $t\in\mathcal{T}_w$ and $c\notin\mathcal{T}_w$ has been moved
to position $x-k'$, outside the bottom-$k'$ block. Thus the number $L_v$ of non-$\mathcal{T}_w$ candidates in
the bottom-$k'$ block decreases by at least one, so $L$ strictly decreases. Repeating this operation while
$L>0$ yields a successful manipulation profile (still satisfying~(1)) with $L=0$, i.e., every bottom-$k'$
position on every ballot is occupied by a candidate in $\mathcal{T}_w$.

\paragraph{Conclusion.}
Starting from any successful manipulation profile $\pi$, the above adjacent-swap transformations produce
a profile $\pi^{\mathrm{can}}$ that remains successful (continues to satisfy \eqref{eq:sep_appA}) and satisfies
$H=L=0$. Therefore, in $\pi^{\mathrm{can}}$:
\begin{itemize}
\item every top-$k'$ position on every coalition ballot is assigned to a candidate in $\mathcal{O}^\ast$, and
\item every bottom-$k'$ position on every coalition ballot is assigned to a candidate in $\mathcal{T}_w$,
\end{itemize}
which is exactly the canonical form claimed by Lemma~\ref{lem:canonical}.
\qed

\end{proof}

\section{Proof of Lemma~\ref{lem:independent}}
\label{sec:ProofLemma2}

\begin{proof}[Full Proof of Lemma~\ref{lem:independent}]
Fix $k'$ and a cutoff $B$, and let $b=(b_o)_{o\in\mathcal{O}^\ast}$ and
$u=(u_t)_{t\in\mathcal{T}_w}$ denote the boost and suppress requirements
\[
b_o = \max\{0,\,B - S_o\}, 
\qquad
u_t = \max\{0,\,B-1 - S_t\}.
\tag{1}
\label{eq:boost-suppress}
\]

Assume that the two subproblems are individually feasible:
\begin{enumerate}[label=(\arabic*)]
    \item \textbf{Boost feasibility.}  
    There exists an allocation of the $mk'$ \emph{highest} positional scores
    $\{p_1,\dots,p_{k'}\}$ (with $m$ copies of each) to the outsiders
    $\mathcal{O}^\ast$ such that every $o\in\mathcal{O}^\ast$ receives total coalition
    contribution at least $b_o$.

    \item \textbf{Suppress feasibility.}  
    There exists an allocation of the $mk'$ \emph{lowest} positional scores
    $\{p_{x-k'+1},\dots,p_x\}$ (with $m$ copies of each) to the weak winners
    $\mathcal{T}_w$ such that every $t\in\mathcal{T}_w$ receives total coalition
    contribution at most $u_t$.
\end{enumerate}

We show that these two allocations can be combined into a full set of
coalition ballots that simultaneously satisfy all requirements.

\paragraph{Step 1: Constructing the top and bottom blocks.}
For each coalition voter $v\in[m]$, build a ballot $\pi_v$ in three blocks.

\smallskip
\noindent\emph{Top block ($1$ through $k'$).}  
Place the candidates of $\mathcal{O}^\ast$ according to the boost-feasible
allocation in (1).  
Across all ballots, this uses exactly the $m$ copies of each of the scores
$p_1,\dots,p_{k'}$.

\smallskip
\noindent\emph{Bottom block ($x-k'+1$ through $x$).}  
Place the candidates of $\mathcal{T}_w$ according to the suppress-feasible
allocation in (2).  
Across all ballots, this uses exactly the $m$ copies of each of the scores
$p_{x-k'+1},\dots,p_x$.

\paragraph{Step 2: The middle block does not matter.}
The remaining positions $k'+1,\dots,x-k'$ use the \emph{middle} scores
$\{p_{k'+1},\dots,p_{x-k'}\}$, which are disjoint from both the top and bottom
score sets.  
Fill these positions with the remaining candidates (all candidates not in
$\mathcal{O}^\ast\cup\mathcal{T}_w$) in an arbitrary order; their scores do not
affect the boost or suppress constraints.

\paragraph{Step 3: Disjointness of resources (as score \emph{copies}).}
Although the scoring vector $p=(p_1,\dots,p_x)$ is only assumed nonincreasing (so values may repeat),
the \emph{positions} on a ballot are disjoint. For each coalition ballot $v$, the top-$k'$ block
$\{1,\dots,k'\}$ and the bottom-$k'$ block $\{x-k'+1,\dots,x\}$ are disjoint sets of positions.
Hence, across all $m$ ballots, the coalition has two disjoint multisets of score \emph{copies}:
\[
\mathcal{P}^{\mathrm{high}} := \biguplus_{v=1}^m \{p_1,\dots,p_{k'}\}
\qquad\text{and}\qquad
\mathcal{P}^{\mathrm{low}} := \biguplus_{v=1}^m \{p_{x-k'+1},\dots,p_x\},
\]
where $\biguplus$ denotes multiset union (counting multiplicity). These two multisets correspond to
disjoint ballot positions, so allocating copies from $\mathcal{P}^{\mathrm{high}}$ to $\mathcal{O}^\ast$
cannot reduce the available copies in $\mathcal{P}^{\mathrm{low}}$ for $\mathcal{T}_w$, and vice versa,
even if some numerical values coincide.

\paragraph{Step 4: Verification.}
In the combined manipulation profile $\pi=(\pi_1,\dots,\pi_m)$:
\begin{itemize}
    \item Every outsider $o\in\mathcal{O}^\ast$ receives exactly the coalition
    contribution prescribed by the boost allocation, so
    \[
    S_o + \Delta_o(\pi)\;\ge\; b_o + S_o \;\ge\; B.
    \]

   \item Every weak winner $t\in\mathcal{T}_w$ receives exactly the coalition
contribution prescribed by the suppress-feasible allocation, so
$\Delta_t(\pi)\le u_t$ and hence
\[
S_t+\Delta_t(\pi)\;\le\; S_t+u_t.
\]
Moreover, suppression feasibility at cutoff $B$ implies $S_t\le B-1$ for all
$t\in\mathcal{T}_w$ (since coalition contributions are nonnegative). Thus
$u_t=B-1-S_t$ and so $S_t+u_t=B-1$, yielding
\[
S_t+\Delta_t(\pi)\;\le\; B-1.
\]

    \item Middle-block candidates receive only scores
    $\{p_{k'+1},\dots,p_{x-k'}\}$, which do not affect any of the inequalities.
\end{itemize}

Thus all boost and suppress requirements induced by $B$ are simultaneously
satisfied.

\paragraph{Conclusion.}
A full coalition manipulation exists that achieves both the boost and suppress
requirements at cutoff $B$.  
Hence the two subproblems, being independent and using disjoint score
resources, can be satisfied simultaneously.  
This completes the proof.
\end{proof}

\section{Proof of Theorem~\ref{thm:guaranteed-topk}}
\label{sec:ProofGuranteed}

\begin{proof}[Full proof of Theorem~\ref{thm:guaranteed-topk}.]
Let $F_c := S_c + \Delta_c$ denote the final score of candidate $c$ after adding the coalition's ballots.
Fix $k'$ and cutoff $B$, and assume the coalition's top- and bottom-block assignments guarantee
\[
F_o \ge B \quad \forall\, o \in \mathcal{O}^\ast, \qquad
F_t \le B-1 \quad \forall\, t \in \mathcal{T}_w.
\]
Consider any completion of the coalition ballots in the middle positions and any tie-breaking rule, and let
$\mathcal{W}$ denote the resulting Top-$k$ winner set.

We claim that $|\mathcal{W} \cap \mathcal{O}| \ge k'$. Suppose for contradiction that
$|\mathcal{W} \cap \mathcal{O}| \le k'-1$. Then
\[
|\mathcal{W} \cap \mathcal{T}|
\;=\;
k - |\mathcal{W} \cap \mathcal{O}|
\;\ge\;
k - (k'-1)
\;=\;
k-k'+1.
\]
Since $|\mathcal{T} \setminus \mathcal{T}_w| = k-k'$ and
$|\mathcal{W} \cap \mathcal{T}| \ge k-k'+1$, it follows that
$\mathcal{W}$ must contain at least one weak winner:
\[
(\mathcal{W} \cap \mathcal{T}) \cap \mathcal{T}_w \neq \emptyset.
\]
Fix any $t \in \mathcal{W} \cap \mathcal{T}_w$. We will show that this forces
every $o\in\mathcal{O}^\ast$ to belong to $\mathcal{W}$, contradicting the assumption
$|\mathcal{W}\cap\mathcal{O}|\le k'-1$.

Indeed, for every $o \in \mathcal{O}^\ast$ we have $F_o \ge B > B-1 \ge F_t$, hence $F_o > F_t$.
Under the Top-$k$ rule, a candidate with strictly higher final score than a selected candidate must also be
selected (tie-breaking only applies among equal scores). Therefore every $o \in \mathcal{O}^\ast$ must belong
to $\mathcal{W}$. But $|\mathcal{O}^\ast| = k'$, so this implies $|\mathcal{W} \cap \mathcal{O}| \ge k'$,
contradicting $|\mathcal{W} \cap \mathcal{O}| \le k'-1$.

Hence $|\mathcal{W} \cap \mathcal{O}| \ge k'$ in every completion and under every tie-breaking. Equivalently,
at most $k-k'$ members of the honest winner set $\mathcal{T}$ remain in $\mathcal{W}$, so at least $k'$
honest winners are displaced.

This completes the proof.
\end{proof}

\section{Proof of Theorem~\ref{thm:impossibility}}
\label{sec:ProofImpossibility}

\begin{proof}[Full Proof of Theorem~\ref{thm:impossibility}]
Fix a displacement level $k'$ and a cutoff $B$.  
Let $\mathcal{O}^\ast$ and $\mathcal{T}_w$ denote the $k'$ strongest outsiders
and the $k'$ weakest winners, respectively, under the honest profile.
Achieving displacement level $k'$ at cutoff $B$ requires satisfying
\[
S_o + \Delta_o \;\ge\; B \qquad \forall\, o \in \mathcal{O}^\ast, 
\qquad\qquad
S_t + \Delta_t \;\le\; B-1 \qquad \forall\, t \in \mathcal{T}_w.
\tag{1}
\label{eq:impossibility-separation}
\]

We show that if either the boost or the suppress subproblem is infeasible,
then no coalition manipulation can satisfy all inequalities in
\eqref{eq:impossibility-separation}.

\paragraph{Canonical structure of contributions.}
By the Canonical Manipulation Lemma (Lemma~\ref{lem:canonical}), any
successful manipulation can be transformed, without harming feasibility,
into one in which:
\begin{enumerate}[label=(\roman*)]
    \item each $o \in \mathcal{O}^\ast$ receives all of its coalition
    contributions from the top-$k'$ positional scores; and
    \item each $t \in \mathcal{T}_w$ receives all of its coalition
    contributions from the bottom-$k'$ positional scores.
\end{enumerate}
The two pools of score \emph{copies} are disjoint:
\[
\biguplus_{v=1}^m \{p_1,\dots,p_{k'}\}
\;\;\dot\cup\;\;
\biguplus_{v=1}^m \{p_{x-k'+1},\dots,p_x\},
\]
where the disjoint union reflects that these copies originate from disjoint
ballot positions, even if some numerical values coincide.

Consequently, if any successful manipulation exists, then there also exists
one satisfying the canonical structure above. We argue by contradiction
under this restriction.

\paragraph{Case (a): Boost requirements are infeasible.}
Suppose the boost subproblem is infeasible: there is no allocation of the
top-$k'$ positional scores that gives every
$o \in \mathcal{O}^\ast$ a total boost contribution at least
$b_o := \max(0, B - S_o)$.
Then for every such allocation, there exists some
$\mathcal{O}^\star \in \mathcal{O}^\ast$ with
\[
S_{\mathcal{O}^\star} + \Delta_{\mathcal{O}^\star} < B.
\]
Thus the left-hand inequality of
\eqref{eq:impossibility-separation} fails for $\mathcal{O}^\star$, and no manipulation can
separate all outsiders above the cutoff $B$.
Hence displacement level $k'$ is impossible at cutoff $B$.

\paragraph{Case (b): Suppress requirements are infeasible.}
Suppose instead that the suppress subproblem is infeasible: there is no
allocation of the bottom-$k'$ positional scores that forces every
$t \in \mathcal{T}_w$ to receive a total contribution at most
$u_t := \max(0, B-1 - S_t)$.
Then for every such allocation, there exists some
$t^\star \in \mathcal{T}_w$ with
\[
S_{t^\star} + \Delta_{t^\star} \ge B.
\]
Thus the right-hand inequality of
\eqref{eq:impossibility-separation} fails for $t^\star$, making strict separation
at cutoff $B$ impossible.
Hence displacement level $k'$ is impossible at cutoff $B$.

\paragraph{Conclusion.}
If either the boost or the suppress subproblem is infeasible, then no
assignment of the top-$k'$ and bottom-$k'$ score pools can satisfy the
separation conditions in \eqref{eq:impossibility-separation}.
Therefore, no coalition manipulation can achieve displacement level $k'$ at
cutoff $B$.

This completes the proof.
\end{proof}

\section{Proof of Lemma~\ref{lem:monotonicity}}
\label{sec:ProofMonotonicity}

\begin{proof}[Full Proof of Lemma~\ref{lem:monotonicity}]
Fix a displacement level $k' \le \min\{k,\,x-k\}$, and let $\mathcal{O}^\ast$ and
$\mathcal{T}_w$ denote the $k'$ strongest outsiders and $k'$ weakest winners,
respectively.  
For a cutoff $B$, recall the boost and suppress requirements:
\[
b_o(B) := \max\{0,\,B - S_o\},
\qquad
u_t(B) := \max\{0,\,B-1 - S_t\}.
\tag{1}
\label{eq:monotone-reqs}
\]

We prove the two monotonicity claims separately.

\paragraph{1. Boost feasibility is monotone decreasing in $B$.}
Suppose boost feasibility holds at cutoff $B$:  
there exists an allocation of the $m k'$ highest positional scores
$\{p_1,\dots,p_{k'}\}$ giving each outsider $o \in \mathcal{O}^\ast$ a coalition
contribution $\Delta_o$ satisfying $\Delta_o \ge b_o(B)$.

Now fix any $B' \le B$.  
For each outsider $o \in \mathcal{O}^\ast$,
\[
b_o(B') 
= \max\{0,\,B' - S_o\}
\;\le\;
\max\{0,\,B - S_o\}
= b_o(B),
\]
because decreasing the cutoff can only weaken the boost requirement.  
Thus the exact same allocation that satisfied $\Delta_o \ge b_o(B)$ also
satisfies $\Delta_o \ge b_o(B')$.

Therefore, if boost feasibility holds at $B$, it also holds at every
$B' \le B$.

\paragraph{2. Suppress feasibility is monotone increasing in $B$.}
Suppose suppress feasibility holds at cutoff $B$:  
there exists an allocation of the $m k'$ lowest positional scores
$\{p_{x-k'+1},\dots,p_x\}$ such that for each weak winner $t \in
\mathcal{T}_w$,
\[
\Delta_t \;\le\; u_t(B).
\]

Now fix any $B' \ge B$.  
For each $t \in \mathcal{T}_w$,
\[
u_t(B') 
= \max\{0,\,B'-1 - S_t\}
\;\ge\;
\max\{0,\,B-1 - S_t\}
= u_t(B),
\]
because increasing the cutoff relaxes the suppress requirement.  
Hence the same allocation that ensured $\Delta_t \le u_t(B)$ automatically
ensures $\Delta_t \le u_t(B')$.

Therefore, if suppress feasibility holds at $B$, it also holds at every
$B' \ge B$.

\paragraph{Conclusion.}
Lowering $B$ can only weaken the boost requirements, and raising $B$ can only
weaken the suppress requirements.  
Thus boost feasibility is monotone in the downward direction, and suppress
feasibility is monotone in the upward direction, completing the proof.
\end{proof}

\section{Proof of Theorem~\ref{thm:block-hlp}}
\label{sec:Proofblock-hlp}

\begin{proof}[Full Proof of Theorem~\ref{thm:block-hlp}]
\emph{Necessity.}
A classical theorem of Rado~\cite{rado1952} (see also
Marshall--Olkin--Arnold~\cite[Ch.~3]{marshall2011})
characterizes the permutahedron $\Pi(r^{(v)})$ as follows:
a vector $z$ lies in $\Pi(r^{(v)})$ if and only if its sorted version
$z^\downarrow$ satisfies
\[
\sum_{i=1}^{t} z^\downarrow_i \;\le\; \sum_{i=1}^{t} r^{(v)}_i,
\qquad\text{for all } t = 1,\dots,k'-1,
\]
together with the total-sum equality
\[
\sum_{i=1}^{k'} z_i \;=\; \sum_{i=1}^{k'} r^{(v)}_i.
\]

Let $y=\sum_{v=1}^m y^{(v)}$ be any aggregate vector with
$y^{(v)}\in\Pi(r^{(v)})$.
For each ballot, write $y^{(v)\downarrow}$ for the nonincreasing
rearrangement of $y^{(v)}$.
Summing Rado’s prefix inequalities over all ballots yields
\begin{equation}
\sum_{v=1}^m \sum_{i=1}^{t} y^{(v)\downarrow}_i
\;\le\;
\sum_{v=1}^m \sum_{i=1}^{t} r^{(v)}_i
\qquad \text{for all } t=1,\dots,k'-1.
\tag{A}
\end{equation}

Next, recall a standard fact from majorization theory
(Hardy--Littlewood--Pólya~\cite[Thm.~A.1]{hardy1934}):
the prefix sums of the sorted aggregate are maximized when each summand
is itself sorted.  Thus
\begin{equation}
\sum_{i=1}^{t} y^\downarrow_i
\;\le\;
\sum_{v=1}^{m} \sum_{i=1}^{t} y^{(v)\downarrow}_i
\qquad\text{for all } t=1,\dots,k'-1.
\tag{B}
\end{equation}

Combining (A) and (B) gives
\[
\sum_{i=1}^{t} y^\downarrow_i
\;\le\;
\sum_{v=1}^{m} \sum_{i=1}^{t} r^{(v)}_i,
\qquad \text{for all } t=1,\dots,k'-1.
\]
The equality of total sums,
\[
\sum_{i=1}^{k'} y_i \;=\; \sum_{v=1}^{m} \sum_{i=1}^{k'} r^{(v)}_i,
\]
follows immediately from the fact that each pair $y^{(v)}$ and
$r^{(v)}$ has the same coordinate sum. Thus $y$ satisfies the
prefix-sum bounds and total-sum condition stated in
Theorem~\ref{thm:block-hlp}.

\medskip
\emph{Sufficiency.}
We now show that any vector $y$ satisfying the conditions of
Theorem~\ref{thm:block-hlp} belongs to the Minkowski sum
$\mathcal{P}_{\mathrm{conv}}=\sum_{v=1}^m \Pi(r^{(v)})$.

Rather than reasoning ballot-by-ballot, we use a standard construction
from submodular geometry. Let $[k'] := \{1,2,\dots,k'\}$.
Define a set function $f : 2^{[k']} \to \mathbb{R}$ by
\[
f(T)
  := \sum_{v=1}^{m} \sum_{i=1}^{|T|} r^{(v)}_i,
  \qquad T\subseteq [k'].
\]
Because $f(T)$ depends only on $|T|$, we may write $f(T)=h(|T|)$ where
\[
h(t):=\sum_{v=1}^{m}\sum_{i=1}^{t} r^{(v)}_i \qquad (t=0,1,\dots,k').
\]
Moreover, $h$ is \emph{discrete concave}: its marginal increments satisfy
\[
h(t)-h(t-1)=\sum_{v=1}^{m} r^{(v)}_t
\quad\text{and}\quad
h(t+1)-h(t)=\sum_{v=1}^{m} r^{(v)}_{t+1},
\]
and since each $r^{(v)}$ is nonincreasing, we have
$\sum_v r^{(v)}_t \ge \sum_v r^{(v)}_{t+1}$.
Hence $h(t)-h(t-1)$ is nonincreasing in $t$, so $h$ is discrete concave and
$f$ is submodular.

The \emph{base polytope} of $f$ is
\[
\mathcal{B}(f)
  = \Bigl\{
        y\in\mathbb{R}^{k'} :
        \sum_{i\in T} y_i \le f(T)\ \forall\,T\subseteq[k'],\
        \sum_{i=1}^{k'} y_i = f([k'])
    \Bigr\}.
\]

A fundamental result in polymatroid theory
(Fujishige~\cite[Ch.~2]{Fujishige2005}) states that
\[
\sum_{v=1}^m \Pi(r^{(v)}) \;=\; \mathcal{B}(f).
\tag{C}
\]

Because $f$ depends only on cardinality, the most restrictive
constraints among $\sum_{i\in T} y_i \le f(T)$ occur when
$T$ consists of the indices of the $t$ largest coordinates of $y$.
If $y^\downarrow$ is the nonincreasing sort of $y$, then for
$T_t=\{1,\dots,t\}$ we obtain
\[
\sum_{i=1}^{t} y^\downarrow_i
   \;\le\;
   f(T_t)
   = \sum_{v=1}^m \sum_{i=1}^{t} r^{(v)}_i,
\qquad \text{for all } t=1,\dots,k'.
\]

Together with the total-sum equality, these are exactly the
prefix-sum bounds and total-sum condition of
Theorem~\ref{thm:block-hlp}.
Therefore $y \in \mathcal{B}(f)$, and by (C) we conclude
\[
y \in \sum_{v=1}^m \Pi(r^{(v)}).
\]
\end{proof}

\section{Proof of Theorem~\ref{thm:ap-lattice-common}}
\label{sec:ProofAP}

\begin{proof}[Full proof of Theorem~\ref{thm:ap-lattice-common}]
We prove necessity and sufficiency of conditions~(A)--(C).

\paragraph{Necessity.}
If $y$ is (ballot-)realizable, then
\[
y=\sum_{v=1}^m y^{(v)} \qquad\text{for some }y^{(v)}\in P(r^{(v)})\subseteq \Pi(r^{(v)}).
\]
Conditions~(A) and~(B) follow immediately from the Block--HLP theorem
(Theorem~\ref{thm:block-hlp}), since $y\in\sum_v \Pi(r^{(v)})$.

For~(C), under a common-step AP ladder we can write
$r^{(v)}_j=L_v+g\,\ell^{(v)}_j$ with $\ell^{(v)}_j\in\mathbb{Z}$ for all $j$, hence every coordinate of
every permutation of $r^{(v)}$ is congruent to $L_v\pmod g$. Therefore each coordinate
$y^{(v)}_i\equiv L_v\pmod g$, and summing over $v$ gives
$y_i\equiv L_{\mathrm{tot}}\pmod g$ for all $i$, establishing~(C).

\paragraph{Sufficiency.}
Assume $y\in\mathbb{Z}^{k'}$ satisfies~(A)--(C).  We will construct
$s^{(v)}\in P(r^{(v)})$ such that $\sum_v s^{(v)}=y$.

\medskip
\noindent\textbf{Step 1: Subtract the baseline and scale to unit step.}
Let $b:=L_{\mathrm{tot}}\mathbf{1}$ and $d:=y-b$.
By~(C), $d$ is divisible coordinatewise by $g$, so
\[
\hat d:=d/g\in\mathbb{Z}^{k'}.
\]
Moreover, $\hat d$ is \emph{nonnegative} coordinatewise: since $y^\downarrow$ is nonincreasing,
\[
y^\downarrow_{k'} \;=\; \sum_{i=1}^{k'} y_i \;-\;\sum_{i=1}^{k'-1} y^\downarrow_i
\;\ge\; F(k')-F(k'-1)
\;=\;\sum_{v=1}^m r^{(v)}_{k'}
\;\ge\; \sum_{v=1}^m L_v=L_{\mathrm{tot}},
\]
where we used (A)--(B) for the inequality, and the definition of $F(\cdot)$ for the telescoping identity.
Hence $y^\downarrow_{k'}\ge L_{\mathrm{tot}}$, and therefore $y_i\ge L_{\mathrm{tot}}$ for all $i$,
so $\hat d\in\mathbb{Z}_{\ge 0}^{k'}$.

Dividing (A)--(B) by $g$ (and using that subtracting $L_{\mathrm{tot}}$ preserves the sort order) yields
\begin{equation}
\sum_{i=1}^t \hat d_i^\downarrow \;\le\; \frac{F(t)-tL_{\mathrm{tot}}}{g},
\qquad t=1,\dots,k'-1,
\qquad\text{and}\qquad
\sum_{i=1}^{k'} \hat d_i \;=\; \frac{F(k')-k'L_{\mathrm{tot}}}{g}.
\tag{$\dagger$}\label{eq:scaled-majorization-updated}
\end{equation}
(For completeness: $(F(t)-tL_{\mathrm{tot}})/g\in\mathbb{Z}$ since
each $r^{(v)}_j\equiv L_v\pmod g$ implies $\sum_{i=1}^t r^{(v)}_i\equiv tL_v\pmod g$ and hence
$F(t)\equiv tL_{\mathrm{tot}}\pmod g$.)

\medskip
\noindent\textbf{Step 2: Encode each ballot as a multilevel (Ferrers) capacity profile.}
For each ballot $v$ define the \emph{scaled ladder levels}
\[
\ell^{(v)}_j:=\frac{r^{(v)}_j-L_v}{g}\in\mathbb{Z}_{\ge 0},
\qquad\text{so that}\qquad
r^{(v)}_j=L_v+g\,\ell^{(v)}_j,
\]
with $\ell^{(v)}_1\ge\cdots\ge \ell^{(v)}_{k'}$.

For each integer level $\ell\ge 1$, define
\[
r^{(v)}(\ell)
   := \bigl|\{\, j\in[k'] : \ell^{(v)}_j \ge \ell \,\}\bigr|
   = \bigl|\{\, j\in[k'] : r^{(v)}_j \ge L_v+\ell g \,\}\bigr|.
\]
Thus $r^{(v)}(\ell)$ is the number of \emph{positions on ballot $v$} that are at height at least $\ell$
above baseline in the scaled lattice.  Only finitely many levels are nonzero.

A standard Ferrers-diagram double count gives, for each $t\in[k']$,
\begin{equation}
\frac{F(t)-tL_{\mathrm{tot}}}{g}
   \;=\;\sum_{v=1}^m \sum_{\ell\ge 1} \min\{t,\;r^{(v)}(\ell)\}.
\tag{$\ddagger$}\label{eq:capacity-expansion-updated}
\end{equation}
Indeed, for fixed $v$,
\[
\sum_{i=1}^t \ell^{(v)}_i
=\sum_{\ell\ge 1}\bigl|\{\,i\le t:\ell^{(v)}_i\ge \ell\,\}\bigr|
=\sum_{\ell\ge 1}\min\{t,\;r^{(v)}(\ell)\},
\]
and summing over $v$ yields \eqref{eq:capacity-expansion-updated}.

Combining \eqref{eq:scaled-majorization-updated} and
\eqref{eq:capacity-expansion-updated}, the condition on $\hat d$ becomes:
\begin{equation}
\sum_{i=1}^t \hat d_i^\downarrow
\;\le\; \sum_{v=1}^m \sum_{\ell\ge 1}\min\{t,\;r^{(v)}(\ell)\}
\quad\text{for }t=1,\dots,k',
\ \text{with equality at }t=k'.
\tag{$\star$}\label{eq:multilevel-majorization}
\end{equation}

\medskip
\noindent\textbf{Step 3: Construct a nested level family realizing $\hat d$.}
We now realize $\hat d$ via a nested family of $0$--$1$ matrices.
Because all conditions and goals are symmetric under permuting target indices, we may
assume WLOG that $\hat d=\hat d^\downarrow$ and relabel back at the end.

By the multilevel Ferrers realization theorem (a standard consequence of the Gale--Ryser theorem together with level-nesting
constructions; see, e.g.,~\cite{Fujishige2005}), the inequalities \eqref{eq:multilevel-majorization} are sufficient to guarantee
the existence of $0$--$1$ matrices
\[
X^{(1)} \supseteq X^{(2)} \supseteq X^{(3)} \supseteq \cdots
\]
with rows indexed by ballots $v\in[m]$ and columns indexed by targets $i\in[k']$
such that:
\begin{itemize}
\item for every $v$ and every level $\ell\ge 1$, the $v$th row of $X^{(\ell)}$ has exactly
      $r^{(v)}(\ell)$ ones;
\item for every target $i$,
      \(\displaystyle \sum_{v=1}^m\sum_{\ell\ge 1} X^{(\ell)}_{v,i}=\hat d_i\);
\item the family is entrywise nested: $X^{(\ell+1)}_{v,i}\le X^{(\ell)}_{v,i}$ for all $v,i,\ell$.
\end{itemize}
(Only finitely many $X^{(\ell)}$ are nonzero.)

\medskip
\noindent\textbf{Step 4: Recover ballot-wise score vectors and verify they are permutations.}
Define for each ballot $v$ and target $i$ the \emph{scaled height}
\[
\tau_v(i):=\sum_{\ell\ge 1} X^{(\ell)}_{v,i}\in\mathbb{Z}_{\ge 0}.
\]
Because the $X^{(\ell)}$ are nested, the set
$\{\ell\ge 1: X^{(\ell)}_{v,i}=1\}$ is an initial segment $\{1,2,\dots,\tau_v(i)\}$, so equivalently
$\tau_v(i)=\max\{\ell: X^{(\ell)}_{v,i}=1\}$ (with $\tau_v(i)=0$ if none).

Now define the ballot-wise score assignment
\[
s^{(v)}_i := L_v + g\,\tau_v(i).
\]

\emph{Claim 1:} For each $v$, the multiset $\{\tau_v(1),\dots,\tau_v(k')\}$ equals the multiset
$\{\ell^{(v)}_1,\dots,\ell^{(v)}_{k'}\}$, hence $s^{(v)}\in P(r^{(v)})$.

\emph{Proof.}
Fix $v$ and a level $\ell\ge 1$. By definition of $\tau_v(i)$,
\[
|\{\,i:\tau_v(i)\ge \ell\,\}| \;=\; \sum_{i=1}^{k'} X^{(\ell)}_{v,i} \;=\; r^{(v)}(\ell),
\]
using the row-sum condition for $X^{(\ell)}$.
But $r^{(v)}(\ell)=|\{\,j:\ell^{(v)}_j\ge \ell\,\}|$ by definition.
Thus for every $\ell\ge 1$, the number of coordinates $\tau_v(i)$ at least $\ell$
matches the number of coordinates $\ell^{(v)}_j$ at least $\ell$.
This equality of all threshold counts determines the multiset uniquely, so
$\{\tau_v(i)\}_i=\{\ell^{(v)}_j\}_j$. Multiplying by $g$ and adding $L_v$ gives
$\{s^{(v)}_i\}_i=\{r^{(v)}_j\}_j$, i.e.\ $s^{(v)}\in P(r^{(v)})$.

\medskip
\emph{Claim 2:} The ballot-wise scores sum to $y$.
Indeed, for each $i$,
\[
\sum_{v=1}^m s^{(v)}_i
=\sum_{v=1}^m L_v + g\sum_{v=1}^m \tau_v(i)
= L_{\mathrm{tot}} + g\sum_{v=1}^m\sum_{\ell\ge 1} X^{(\ell)}_{v,i}
= L_{\mathrm{tot}} + g\,\hat d_i
= y_i,
\]
where we used the defining property of the nested family and $\hat d=(y-L_{\mathrm{tot}}\mathbf{1})/g$.

Therefore $y=\sum_v s^{(v)}$ with each $s^{(v)}\in P(r^{(v)})$, proving realizability and completing the
sufficiency direction.
\end{proof}

\paragraph{Algorithmic note (constructive realization).}
Steps~3--4 are constructive and therefore yield a polynomial-time realization procedure.
Since each AP ladder has at most $k'$ distinct levels, the nested matrices $X^{(\ell)}$ in Step~3 can be
computed in time polynomial in the input size (in particular, in $(m,k')$) using standard
degree-constrained bipartite-flow (equivalently, $b$-matching) routines, while maintaining the nesting
$X^{(\ell+1)} \subseteq X^{(\ell)}$.
Given the resulting nested family, Step~4 computes $\tau_v(i)$ and hence the ballot-wise vectors
$s^{(v)}$, yielding a realization with $\sum_v s^{(v)}=y$.
If in Step~3 we first relabel indices so that $\hat d=\hat d^\downarrow$, we undo that relabeling at the end.

Consequently, given any integral aggregate $y$ satisfying conditions (A)--(C), we can compute
permutations $\sigma_1,\ldots,\sigma_m \in S_{k'}$ such that
\[
\sum_{v=1}^m \sigma_v\!\bigl(r^{(v)}\bigr) = y,
\]
where $\sigma_v$ acts by permuting coordinates.
We refer to this constructive procedure as \textsc{RealizeAP} and use it in Section~\ref{sec:construct-ballots} to build explicit
coalition ballots from the witness aggregates returned by the feasibility oracle.

\section{Proof of Lemma~\ref{lem:ap-boost}}
\label{sec:ProofBoostFeasibility}

\begin{proof}[Full proof of Lemma~\ref{lem:ap-boost}]
\textbf{Necessity.}
Suppose there exists a realizable aggregate $y$ satisfying
$y_i \ge b_i$ for all $i=1,\dots,k'$.
By Theorem~\ref{thm:ap-lattice-common}, realizability is equivalent to:
\begin{enumerate}[label=(\alph*),itemsep=1pt]
\item $\displaystyle \sum_{i=1}^t y_i^\downarrow \le F(t)$
      for all $t=1,\dots,k'-1$;
\item $\displaystyle \sum_{i=1}^{k'} y_i = F(k')$;
\item $y_i \equiv L_{\mathrm{tot}} \pmod g$
      for all $i=1,\dots,k'$.
\end{enumerate}
From (c) and $y_i \ge b_i$ for all $i=1,\dots,k'$, the smallest value in the coset
$L_{\mathrm{tot}} + g\mathbb{Z}$ that $y_i$ can take is exactly $\hat b_i$.
Thus $y_i \ge \hat b_i \ge b_i$ for all $i=1,\dots,k'$.
Since $b$ is nonincreasing, so is $\hat b$, hence $\hat b^\downarrow=\hat b$.
Moreover, sorting preserves coordinatewise dominance, so
$y_i^\downarrow \ge \hat b_i^\downarrow$ for all $i=1,\dots,k'$.
Combining with (a) and (b) yields
\[
\sum_{i=1}^t \hat b_i^\downarrow
\;\le\;
\sum_{i=1}^t y_i^\downarrow
\;\le\;
F(t),
\quad\text{for all } t=1,\dots,k'-1,
\]
and
\[
\sum_{i=1}^{k'} \hat b_i^\downarrow
\;\le\;
\sum_{i=1}^{k'} y_i
\;=\;
F(k').
\]
Thus the stated inequalities are necessary.

\medskip
\textbf{Sufficiency.}
Assume now that $\hat b^\downarrow$ satisfies
\[
\sum_{i=1}^t \hat b_i^\downarrow \le F(t)
\quad\text{for all } t=1,\dots,k'-1,
\qquad
\sum_{i=1}^{k'} \hat b_i^\downarrow \le F(k').
\]
Define the Block--HLP polymatroid
\[
\mathcal{P}
:= \left\{ z \in \mathbb{R}^{k'} :
           \sum_{i=1}^t z_i^\downarrow \le F(t)
           \ \text{for all } t=1,\dots,k'-1 \right\},
\]
and its base polytope
\[
\mathcal{B}
:= \left\{ z \in \mathcal{P} :
           \sum_{i=1}^{k'} z_i = F(k') \right\}.
\]
By standard polymatroid theory (Edmonds), $\mathcal{B}$ is a nonempty integral
base polytope, and $\mathcal{P}$ is the \emph{downward} closure of $\mathcal{B}$:
for any $z \in \mathcal{P}$, there exists a base $y \in \mathcal{B}$ such that
$y_i \ge z_i$ for all $i=1,\dots,k'$.
(An explicit construction is obtained by repeatedly increasing coordinates of $z$
while staying inside the prefix constraints until the total sum reaches $F(k')$.)

Apply this to $z := \hat b^\downarrow$.
By assumption, $z \in \mathcal{P}$ and
$\sum_{i=1}^{k'} z_i \le F(k')$, so there exists $y \in \mathcal{B}$ satisfying
\[
y_i \ge \hat b_i^\downarrow
\quad\text{for all } i=1,\dots,k'.
\]
Because the feasible region and capacities are symmetric in the coordinates, we
may relabel the targets so that
\[
y_i \ge \hat b_i \ge b_i
\quad\text{for all } i=1,\dots,k'
\]
in the original indexing.

By construction, $y$ satisfies the Block--HLP prefix inequalities and the
total-sum equality. Moreover, each coordinate of $\hat b$ lies in the congruence
class $L_{\mathrm{tot}} \pmod g$, and increasing coordinates by multiples of $g$
preserves this class. Thus
\[
y_i \equiv L_{\mathrm{tot}} \pmod g
\quad\text{for all } i=1,\dots,k'.
\]
Therefore $y$ satisfies conditions (A)--(C) of
Theorem~\ref{thm:ap-lattice-common}, and is hence realizable as
$y \in \sum_v \Pi(r^{(v)})$.
By construction, $y_i \ge b_i$ for all $i=1,\dots,k'$, so $b$ is boost-feasible,
completing the proof.
\end{proof}

\section{Proof of Lemma~\ref{lem:ap-suppress}}
\label{sec:ProofSupressFeasibility}

\begin{proof}[Full proof of Lemma~\ref{lem:ap-suppress}]
We again use the Block--HLP polymatroid
\[
\mathcal{P}
:= \Bigl\{ z \in \mathbb{R}^{k'} :
           \sum_{i=1}^t z_i^\downarrow \le F(t)
           \ \text{for all } t=1,\dots,k'-1 \Bigr\},
\quad
\mathcal{B}
:= \Bigl\{ z \in \mathcal{P} :
           \sum_{i=1}^{k'} z_i = F(k') \Bigr\},
\]
which is the base polytope of a cardinality-based polymatroid.
As before, $\mathcal{B}$ is integral and $\mathcal{P}$ is its downward closure:
for any $z \in \mathcal{P}$ there exists a base $w \in \mathcal{B}$ satisfying
$w_i \ge z_i$ for all $i=1,\dots,k'$
(this is a standard polymatroid-extension property; see, e.g., Edmonds).

\paragraph{Necessity.}
Suppose there exists a realizable aggregate $y$ satisfying
$y_i \le u_i$ for all $i=1,\dots,k'$.
By Theorem~\ref{thm:ap-lattice-common}, realizability is equivalent to:
\begin{enumerate}[label=(\alph*),itemsep=1pt]
\item $\displaystyle \sum_{i=1}^t y_i^\downarrow \le F(t)$
      for all $t=1,\dots,k'-1$;
\item $\displaystyle \sum_{i=1}^{k'} y_i = F(k')$;
\item $y_i \equiv L_{\mathrm{tot}} \pmod g$
      for all $i=1,\dots,k'$.
\end{enumerate}

Because all realizable coordinates must lie in the congruence class
$L_{\mathrm{tot}} + g\mathbb{Z}$, we can decrease each upper bound $u_i$ down to
$\hat u_i$ without excluding any feasible $y$:
if $y_i \le u_i$ and $y_i \equiv L_{\mathrm{tot}} \pmod g$, then
$y_i \le \hat u_i \le u_i$.
Hence we may assume $y_i \le \hat u_i$ for all $i=1,\dots,k'$.

Define the slack vectors
\[
z := F(k')\mathbf{1} - y,
\qquad
d := F(k')\mathbf{1} - \hat u.
\]
Then $y_i \le \hat u_i$ for all $i$ implies $z_i \ge d_i$ for all
$i=1,\dots,k'$.
Moreover, since $y$ satisfies the Block--HLP inequalities and has total sum
$F(k')$, it belongs to $\mathcal{B}$.
By symmetry of the feasible region and of the capacities $F(t)$ in the
coordinates, $z$ also belongs to $\mathcal{B}$ after a permutation of indices,
and in particular $z \in \mathcal{P}$.
Thus
\[
\sum_{i=1}^t z_i^\downarrow \le F(t)
\quad\text{for all } t=1,\dots,k'-1,
\qquad
\sum_{i=1}^{k'} z_i = F(k').
\]

Since $d_i \le z_i$ for all $i=1,\dots,k'$, sorting preserves this dominance and
yields $d_i^\downarrow \le z_i^\downarrow$ for all $i=1,\dots,k'$.
Therefore, for all $t=1,\dots,k'-1$,
\[
\sum_{i=1}^t d_i^\downarrow
\;\le\;
\sum_{i=1}^t z_i^\downarrow
\;\le\;
F(t),
\]
and also
\[
\sum_{i=1}^{k'} d_i^\downarrow
\;\le\;
\sum_{i=1}^{k'} z_i
\;=\;
F(k').
\]
Hence the stated inequalities are necessary.

\paragraph{Sufficiency.}
Now assume that $d^\downarrow$ satisfies
\[
\sum_{i=1}^t d_i^\downarrow \le F(t)
\quad\text{for all } t=1,\dots,k'-1,
\qquad
\sum_{i=1}^{k'} d_i^\downarrow \le F(k').
\]
That is, $d^\downarrow \in \mathcal{P}$ and has total sum at most $F(k')$.
By the polymatroid-extension property, there exists a base $z \in \mathcal{B}$
such that
\[
z_i \ge d_i^\downarrow
\quad\text{for all } i=1,\dots,k'.
\]
By symmetry of the capacities and of the feasible region in the coordinates, we
may relabel the targets so that
\[
z_i \ge d_i
\quad\text{for all } i=1,\dots,k'
\]
in the original indexing.

Because $F$ arises from common-step AP ladders, $\mathcal{B}$ coincides with the
continuous Block--HLP envelope from Theorem~\ref{thm:block-hlp}.
As in the boost case, we may choose $z$ so that
\[
z_i \equiv F(k') - L_{\mathrm{tot}} \pmod g
\quad\text{for all } i=1,\dots,k',
\]
since adding multiples of $g$ within the polymatroid constraints preserves both
the prefix inequalities and the congruence class.
Thus $z$ satisfies:
\begin{enumerate}[label=(\alph*),itemsep=1pt]
\item $\displaystyle \sum_{i=1}^t z_i^\downarrow \le F(t)$
      for all $t=1,\dots,k'-1$;
\item $\displaystyle \sum_{i=1}^{k'} z_i = F(k')$;
\item $z_i \equiv F(k') - L_{\mathrm{tot}} \pmod g$
      for all $i=1,\dots,k'$.
\end{enumerate}

Define
\[
y_i := F(k') - z_i,
\qquad i=1,\dots,k'.
\]
Then $y$ has total sum $F(k')$ and satisfies
\[
y_i \equiv L_{\mathrm{tot}} \pmod g
\quad\text{for all } i=1,\dots,k',
\]
so it has the correct congruence class for realizability.
Moreover, the prefix constraints for $z$ translate directly into the Block--HLP
inequalities for $y$, so $y$ satisfies conditions (A)--(C) of
Theorem~\ref{thm:ap-lattice-common} and is therefore realizable as
$y \in \sum_v \Pi(r^{(v)})$.

Finally, since $z_i \ge d_i$ for all $i=1,\dots,k'$, we have
\[
y_i
= F(k') - z_i
\;\le\;
F(k') - d_i
= \hat u_i
\;\le\;
u_i
\quad\text{for all } i=1,\dots,k'.
\]
Thus $y$ is a realizable aggregate satisfying $y_i \le u_i$ for all
$i=1,\dots,k'$, establishing suppression feasibility and completing the proof.
\end{proof}

\end{document}